%% file: hypercontract.tex
\documentclass[11pt]{article}

\def\showauthornotes{0}
\def\showtableofcontents{1}
\def\showkeys{0}
\def\showdraftbox{0}
\def\showcolorlinks{1}
\def\usemicrotype{1}
\def\showfixme{0}

\input{macros}

\let\pref=\prettyref

\renewcommand{\L}[1]{L_{#1}}

\newcommand{\ptoqnorm}[3]{\|#1\|_{#2\rightarrow #3}}
\newcommand{\tfnorm}[1]{\ptoqnorm{#1}{2}{4}}
\newcommand{\ptoqcnorm}[3]{\cnorm{#1}_{#2\rightarrow #3}}
\newcommand{\tfcnorm}[1]{\ptoqcnorm{#1}{2}{4}}

\newcommand{\cp}{\textup{\textsf{cp}}}
\newcommand{\tensorsdp}{\textup{\textsf{Tensor-SDP}}\xspace}

\newcommand{\frv}{\mbox{f.r.v.}\xspace}
\newcommand{\pef}{\mbox{p.e.f.}\xspace}

\newcommand{\lrvar}[1]{$#1$-\frv}

\DeclareMathOperator*{\lift}{\textsf{Lift}}
\DeclareMathOperator{\coeff}{\mathit{coeff}}

\DeclareMathOperator*{\pE}{\Tilde{\Esymb}}

\newcommand{\bd}{\Phi}
\newcommand{\Holder}{H\"{o}lder}
\title{Hypercontractivity, Sum-of-Squares Proofs, and their Applications}

\author{Boaz Barak\thanks{Microsoft Research New England.}  \and
  Fernando G.S.L. Brand\~ao\thanks{Department of Computer Science,
    University College London} \and Aram W. Harrow\thanks{Department
    of Physics, Massachusetts Institute of Technology. Much of the
    work was done while visiting Microsoft Research New England and
    while working at the University of Washington.} \and Jonathan
  Kelner\thanks{Department of Mathematics, Massachusetts Institute of
    Technology.} \and David Steurer\thanks{Microsoft Research New
    England.} \and Yuan Zhou\thanks{Department of Computer Science,
    Carnegie Mellon University. Much of the work done while the author
    was an intern at Microsoft Research New England.} \vspace{-2ex}}

\date{\today\vspace{-4ex}}

\begin{document}

\maketitle

\draftbox \thispagestyle{empty}

\begin{abstract}

We study the computational complexity of approximating the \emph{$2$-to-$q$} norm of linear operators (defined as $\norm{A}_{2\to q} = \max_{v\neq
0}\norm{Av}_q/\norm{v}_2$) for $q > 2$, as well as connections between this question and issues arising in quantum information theory and the study of Khot's
Unique Games Conjecture (UGC). We show the following:

\begin{enumerate}

\vspace{-0.5ex}\item For any constant even integer $q \geq 4$, a graph $G$ is a \emph{small-set expander} if and only if the projector into the span of the
  top eigenvectors of $G$'s adjacency matrix has bounded $2\to q$ norm.  As a corollary, a good approximation to the $2\to q$ norm will refute the
  \emph{Small-Set Expansion Conjecture} --- a close variant of the UGC.  We also show that such a good approximation can be computed in $\exp(n^{2/q})$
  time, thus obtaining a different proof of the known subexponential algorithm for \smallsetexpansion.

\vspace{-0.5ex} \item Constant rounds of the ``Sum of Squares'' semidefinite programing  hierarchy certify an upper bound on the $2\to 4$ norm of the
projector to low-degree polynomials over the Boolean cube, as well certify the unsatisfiability of the ``noisy cube'' and ``short code'' based instances of
\uniquegames considered by prior works. This improves on the previous upper bound of  $\exp(\log^{O(1)} n)$ rounds (for the ``short code''), as well as
separates the ``Sum of Squares''/``Lasserre'' hierarchy from weaker hierarchies that were known to require $\omega(1)$ rounds.

\vspace{-0.5ex}\item We show reductions between computing the $2\to 4$ norm and computing the \emph{injective tensor norm} of a tensor, a problem with
connections to quantum information theory. Three corollaries are: \textbf{(i)} the $2\to 4$ norm is NP-hard to approximate to precision inverse-polynomial
in the dimension, \textbf{(ii)} the $2\ra 4$ norm does not have a good approximation (in the sense above) unless 3-SAT can be solved in time
$\exp(\sqrt{n}\poly\log(n))$, and \textbf{(iii)} known algorithms for the quantum separability problem imply a non-trivial additive approximation for the
$2\to 4$ norm.


\end{enumerate}


\end{abstract}
\thispagestyle{empty}
\clearpage
\thispagestyle{empty}

\ifnum\showtableofcontents=1
{
\small
\tableofcontents
\thispagestyle{empty}
 }
\fi

\setcounter{page}{0}

\clearpage
\setcounter{page}{1}

\section{Introduction} \label{sec:intro}

For a function $f\from \Omega\to \R$ on a (finite) probability
space $\Omega$, the \emph{$p$-norm} is defined as $\norm{f}_p=(\E_\Omega
f^p)^{1/p}$.\footnote{
We follow the convention to use \emph{expectation} norms for \emph{functions} (on probability
spaces) and \emph{counting} norms, denoted as $\cnorm{v}_p=(\sum_{i=1}^n |v_i|^p)^{1/p}$, for
\emph{vectors} $v\in\R^m$.
All normed spaces here will be finite dimensional.
We distinguish between expectation and counting norms to avoid recurrent
normalization factors.}
The \emph{$p\ra q$ norm} $\ptoqnorm{A}{p}{q}$ of a linear operator $A$
between vector spaces of such functions is the smallest number $c\ge 0$
such that $\norm{Af}_q \le c \norm{f}_p$ for all functions $f$ in the
domain of $A$.  We also define the \emph{$p\to q$ norm} of a \emph{subspace} $V$ to be
the maximum of $\norm{f}_q/\norm{f}_p$ for $f\in V$; note that for $p=2$ this is the same
as the norm of the projector operator into  $V$.

%
In this work, we are interested in the case $p<q$ and we will call such
$p\to q$ norms \emph{hypercontractive}.\footnote{We use this name because a
  bound of the form $\ptoqnorm{A}{p}{q} \leq 1$ for $p<q$ is often called a
  \emph{hypercontractive inequality}.}
Roughly speaking, for $p<q$, a function $f$ with large $\norm{f}_q$ compared
to $\norm{f}_p$ can be thought of as ``spiky'' or somewhat sparse (i.e.,
much of the mass concentrated in small portion of the entries).
Hence finding a function $f$ in a linear subspace $V$ maximizing
$\norm{f}_q/\norm{f}_2$ for some $q>2$ can be thought of as a geometric
analogue of the problem finding the shortest word in a linear code.
This problem is equivalent to computing the $2\to q$ norm of the projector
$P$ into $V$ (since $\norm{P f}_2\le \norm{f}_2$).
 Also when $A$ is a
normalized adjacency matrix of a graph (or more generally a Markov
operator), upper bounds on the $p\to q$ norm are known as
\emph{mixed-norm}, \emph{Nash} or \emph{hypercontractive inequalities} and
can be used to show rapid mixing of the corresponding random walk (e.g.,
see the surveys \cite{Gross75,SaloffCoste97}). Such bounds also have
many applications to theoretical computer science, which are described
in the survey~\cite{PunyashlokaBiswalSurvey}.

However, very little is known about the complexity of computing these norms. This is in contrast to the case of  $p\to q$ norms for $p\geq q$, where much more is known both in terms of algorithms and lower bounds, see~\cite{Steinberg05,KindlerNS08,BhaskaraV11}.


\section{Our Results}

We initiate a study of the computational complexity of approximating the $2\to 4$ (and more generally $2\to q$ for $q>2$) norm. While there are still many more questions than answers on this topic, we are able to show some new algorithmic and hardness results, as well as connections to both Khot's unique games
conjecture~\cite{Khot02a} (UGC) and questions from quantum information theory. In particular our paper gives some conflicting evidence regarding the validity
of the UGC and its close variant--- the small set expansion hypothesis (SSEH) of \cite{RaghavendraS10}. (See also our conclusions section.)

First, we show in \thmref{hardness} that approximating the $2\to 4$ problem to within any constant factor cannot be done in polynomial time (unless SAT can be
solved in $\exp(o(n))$ time) but yet this problem is seemingly related to the \uniquegames and \smallsetexpansion problems. In
particular, we show that approximating the $2\to 4$ norm is \smallsetexpansion- hard but yet has a subexponential algorithm closely related to the
\cite{AroraBS10} algorithm for \uniquegames and \smallsetexpansion  . Thus the computational difficulty of this problem can be considered as some indirect
evidence \emph{supporting} the validity of the UGC (or perhaps some weaker variants of it). To our knowledge, this is the first evidence of this kind for the
UGC.

On the other hand, we show that a natural polynomial-time algorithm (based on an SDP hierarchy) that solves the previously proposed hard instances for
\uniquegames . The previous best algorithms for some of these instances took almost exponential ( $\exp(\exp(\log^{\Omega(1)}n))$ ) time, and in fact they were
shown to require super-polynomial time for some hierarchies. Thus this result suggests that this algorithm could potentially refute the UGC, and hence can be
construed as evidence \emph{opposing} the UGC's validity.

\subsection{Algorithms}

We show several algorithmic results for the $2\to 4$ (and more generally $2\to q$) norm.

\subsubsection{Subexponential algorithm for ``good'' approximation}

For $q\geq 2$, we say that an algorithm provides a \emph{$(c,C)$-approximation} for the $2\to q$ norm if on input an operator $A$, the algorithm can distinguish between the case that $\norm{A}_{2\to q} \leq c \sigma$ and the case that $\norm{A}_{2\to q} \geq C \sigma$, where $\sigma = \sigma_{\min}(A)$ is the minimum nonzero singular value of $A$. (Note that since we use the expectation norm, $\norm{Af}_q \geq \norm{Af}_2 \geq \sigma\norm{f}_2$ for every function $f$ orthogonal to the Kernel of $A$.) We say that an algorithm provides a \emph{good approximation} for the $2\to q$ norm if it provides a $(c,C)$-approximation for some (dimension independent) constants $c<C$.  The motivation behind this definition is to capture the notion of a \emph{dimension independent} approximation factor, and is also motivated by Theorem~\ref{thm:sse-hyper} below, that relates a good approximation for the $2\to q$ norm to solving the \smallsetexpansion problem.

We show the following:

\begin{theorem} \torestate{\label{thm:subexp} For every $1<c<C$, there is a $\poly(n)\exp(n^{2/q})$-time algorithm that computes a $(c,C)$-approximation for the $2\to q$ norm of any linear operator whose range is $\R^n$.}
\end{theorem}

Combining this with our results below, we get as a corollary a subexponential algorithm for the \smallsetexpansion problem matching the parameters of \cite{AroraBS10}'s algorithm. We note that this algorithm can be achieved by the ``Sum of Squares'' SDP hierarchy described below (and probably weaker hierarchies as well, although we did not verify this).

\subsubsection{Polynomial algorithm for specific instances}

We study a natural semidefinite programming (SDP) relaxation for computing the $2\to 4$ norm of a given linear operator which we call \tensorsdp .\footnote{We
use the name \tensorsdp for this program since it will be a canonical relaxation of the polynomial program  $\max_{\norm{x}_2=1}{\iprod{T,x^{\ot 4}}}$ where
$T$ is the $4$-tensor such that $\iprod{T,x^{\ot 4}}=\norm{Ax}_4^4$. See Section~\ref{sec:prelim:short} for more details.}  While \tensorsdp\ is very unlikely
to provide a poly-time constant-factor approximation for the $2\to 4$ norm in general (see Theorem~\ref{thm:hardness} below), we do show that it provides such approximation on two
very different types of instances:

\begin{itemize}

\item We show that \tensorsdp certifies  a constant upper bound on the ratio $\tfnorm{A}/\norm{A}_{2\to 2}$ where  $A:\R^n\to\R^m$  is a \emph{random
    linear operator} (e.g., obtained by a matrix with entries chosen as i.i.d Bernoulli variables) and  $m \geq \Omega(n^2\log n)$.  In contrast, if $m =
    o(n^2)$ then this ratio is $\omega(1)$, and hence this result is almost tight in the sense of obtaining ``good approximation'' in the sense mentioned
    above. We find this interesting, since random matrices seem like natural instances; indeed for superficially similar problems such shortest codeword,
    shortest lattice vector (or even the $1\to 2$ norm), it seems hard to efficiently certify bounds on random operators. \Bnote{Add here a correct
    footnote how about how bounding the $1$ to $2$ norm of a random \emph{generating} (not projector) matrix, is related to certifying RIP of a random
    matrix. Perhaps this also answers the reviewer's complaint below.}

\Dnote{we need to be careful by what we mean with $1\to 2$ norm. I think
  even the $1\to 2$ norm of generators is not the right thing. (the 1-norm
  is basis dependent so the $1\to 2$ norm would depend on which generator
  matrix you use.) i think we should state these things in term of $1$ vs
  $2$ subspace distortion and point out that unlike $2$ vs $4$ distortion
  it does not (obviously) correspond to an operator norm. }

\Dnote{(reviewer complained about ``superficially''.) I guess
  ``superficially similar'' is justified in the introduction when we say
  that a large gap between norms corresponds to spiky functions (akin to
  low-weight codewords). the similarity is really only superficial because
  for example no good codes exist in the $2$ vs $4$ norm sense. comparing
  $1$ and $2$ norms in a subspace seems more similar to shortest codeword problems.}

\item We show that \tensorsdp gives a good approximation of the $2\to 4$ norm of the operator projecting a function $f:\{\pm 1\}^n\to\R$ into its low-degree component:

     \begin{theorem} \label{thm:hyper-poly}  Let $\cP_d$ be the linear operator that maps a function $f:\{\pm 1\}^n\to\R$ of the form $f = \sum_{\alpha\subseteq[n]}\Hat{f}_\alpha \chi_{\alpha}$ to its low-degree part $f' = \sum_{|\alpha|\leq d} \Hat{f}_\alpha \chi_{\alpha}$ (where $\chi_{\alpha}(x)=\prod_{i\in\alpha} x_i$). Then  $\tensorsdp(\cP_d) \leq 9^{d}$.
     \end{theorem}

The fact that $\cP_d$ has bounded  $2\to 4$ norm is widely used in the
literature relating to the UGC. Previously, no general-purpose algorithm
was known to efficiently certify this fact.

\end{itemize}

\subsubsection{Quasipolynomial algorithm for additive approximation}

We also consider the generalization of \tensorsdp to a natural SDP \emph{hierarchy}.  This is a convex relaxation that starts from an initial SDP and tightens it by adding additional constraints. Such hierarchies are generally paramaterized by a
number $r$ (often called the \emph{number of rounds}), where the $1^{\text{st}}$
round corresponds to the initial SDP, and the $n^{\text{th}}$ round (for discrete problems where $n$ is
the instance size) corresponds to the exponential brute force algorithm
that outputs an optimal answer. Generally, the $r^{\text{th}}$-round of each such
hierarchy can be evaluated in $n^{O(r)}$ time (though in some cases
$n^{O(1)}2^{O(r)}$ time suffices~\cite{BarakRS11}). See Section~\ref{sec:SoS}, as well as the surveys~\cite{ChlamtacT10,Laurent03} and the
papers~\cite{SheraliA90,LovaszS91,RaghavendraS09c,KhotPS10} for more  information about these hierarchies.

We call the hierarchy we consider here the \emph{Sum of Squares} (SoS) hierarchy. It is not novel but rather a variant of the hierarchies studied by several
authors including Shor~\cite{Shor87}, Parrilo~\cite{Parrilo00,Parrilo03}, Nesterov~\cite{Nesterov00} and Lasserre~\cite{Lasserre01}. (Generally in our context
these hierarchies can be made equivalent in power, though there are
some subtleties involved; see \cite{Laurent09} and Appendix~\ref{app:hierarchies} for more
details.) We describe the SoS hierarchy formally in Section~\ref{sec:SoS}. We show that \tensorsdp 's extension to several rounds of the SoS hierarchy gives a
non-trivial \emph{additive} approximation:


    \begin{theorem}\torestate{\label{thm:BCY}  Let $\tensorsdp^{(d)}$ denote the $n^{O(d)}$-time algorithm by extending \tensorsdp to $d$ rounds of the Sum-of-Squares hierarchy. Then for all $\e$, there is $d=O(\log (n)/\e^2)$ such that
\[
\tfnorm{A}^4 \leq \tensorsdp^{(d)}(A) \leq
\tfnorm{A}^4 + \eps\|A\|_{2\ra 2}^2 \|A\|_{2\ra \infty}^2 \mper
\]}
    \end{theorem}

The term $\|A\|_{2\ra 2}^2 \|A\|_{2\ra \infty}^2$ is a natural upper bound on $\tfnorm{A}^4$ obtained using \Holder 's inequality.   Since $\|A\|_{2 \ra 2}$ is the largest singular value of $A$, and $\|A\|_{2\ra \infty}$ is the largest 2-norm of any row of $A$, they can be computed quickly. Theorem~\ref{thm:BCY} shows that one can improve this upper bound by a factor of $\e$ using run time  $\exp(\log^2(n) / \eps^2)$). Note however that in the special case (relevant to the UGC)  that $A$ is a projector to a subspace $V$, $\norm{A}_{2\to 2}=1$ and $\norm{A}_{2\to \infty} \geq \sqrt{\dim(V)}$ (see Lemma~\ref{lem:two-to-infty}), which unfortunately means that Theorem~\ref{thm:BCY} does not give any new algorithms in that setting.

Despite Theorem~\ref{thm:BCY} being a non-quantum algorithm for for an
ostensibly non-quantum problem, we actually achieve it using the
results of  Brand\~{a}o, Christandl and Yard~\cite{BrandaoCY11} about
the quantum separability problem.  In fact, it turns out that the SoS
hierarchy extension of \tensorsdp is equivalent to techniques that
have been used to approximate separable states~\cite{DohertyPS04}.
We find this interesting both because there are few positive general
results about the convergence rate of SDP hierarchies, and because the
techniques of \cite{BrandaoCY11}, based on entanglement measures of
quantum states, are different from typical ways of proving correctness
of semidefinite programs, and in particular different techniques from
the ones we use to analyze \tensorsdp in other settings. This
connection also means that integrality gaps for \tensorsdp would imply
new types of entangled states that pass most of the known tests for
separability.

%
%
%

\subsection{Reductions}

We relate the question of computing the hypercontractive norm with two other problems considered in the literature: the \emph{small set expansion} problem~\cite{RaghavendraS10,RaghavendraST10},  and the \emph{injective tensor norm} question studied in the context of quantum information theory~\cite{HarrowM10,BrandaoCY11}.

\subsubsection{Hypercontractivity and small set expansion}

Khot's \emph{Unique Games Conjecture}~\cite{Khot02a} (UGC) has been
the focus of intense research effort in the last few years. The
conjecture posits the hardness of approximation for a certain
constraint-satisfaction problem, and shows promise to settle many open
questions in the theory of approximation algorithms. Many works have
been devoted to studying the plausibility of the UGC, as well as
exploring its implications and obtaining unconditional results
inspired or motivated by this effort. Tantalizingly, at the moment we
have very little insight on whether this conjecture is actually true,
and thus producing evidence on the UGC's truth or falsity is a central
effort in computational complexity. Raghavendra and
 Steurer~\cite{RaghavendraS10} proposed  a hypothesis closely related to the UGC
called the \emph{Small-Set Expansion} hypothesis (SSEH). Loosely
speaking, the SSEH states that it is NP-hard to certify that a given
graph $G=(V,E)$ is a \emph{small-set expander} in the sense that
subsets with size $o(|V|)$ vertices have almost all their neighbors
outside the set. \cite{RaghavendraS10}  showed that SSEH implies UGC. While a reduction in the other direction is not known, all currently known algorithmic and integrality gap results apply to both problems equally well (e.g., \cite{AroraBS10,RaghavendraST12}), and thus the two conjectures are likely to be equivalent.

We show, loosely speaking, that a graph is a small-set expander if and
only if the projection operator to the span of its top eigenvectors
has bounded $2\to 4$ norm.   To make this precise, if $G=(V,E)$ is a
regular graph, then let $P_{\ge \lambda}(G)$ be the projection operator into the span of the eigenvectors of $G$'s normalized adjacency
matrix with eigenvalue at least $\lambda$, and $\bd_G(\delta)$ be
$\min_{S\subseteq V, |S|\leq \delta |V|} \Pr_{(u,v)\in E}[ v\not\in S
| u\in S]$.

Then we relate small-set expansion to the $2\ra 4$ norm (indeed the
$2\ra q$ norm for even $q\geq 4$) as follows:
\begin{theorem} \torestate{\label{thm:sse-hyper} For every regular graph $G$,  $\lambda >0$ and even $q$,
\begin{enumerate}
\item \emph{(Norm bound implies expansion)} For all $\delta>0 ,\e>0$,
  $\Vert P_{\ge\lambda}(G) \Vert_{2 \to q} \leq \e/\delta^{(q-2)/2q}$ implies that
 $\bd_G(\delta) \geq 1-\lambda - \e^2$.

\item \emph{(Expansion implies norm bound)} There are constants $c_1, c_2 > 0$ such that for all $\delta>0$, $\bd_G(\delta) > 1 - c_1\lambda^{2q} 2^{- c_2q}$
  implies $\Vert P_{\ge\lambda}(G) \Vert_{2 \rightarrow q} \leq 2/\sqrt{\delta}$.
\end{enumerate}}
\end{theorem}

%
%

While one direction (bounded hypercontractive norm implies small-set
expansion) was already known,\footnote{While we do not know who was
  the first to point out this fact explicitly, within theoretical CS it was implicitly used in several results relating the Bonami-Beckner-Gross hypercontractivity of the Boolean noise operator to isoperimetric properties, with one example being O'Donnell's proof of the soundness of \cite{KhotV05}'s integrality gap (see \cite[Sec~9.1]{KhotV05}).} to our knowledge the other direction is novel. As a corollary we show that the SSEH implies that there is no good approximation for the $2\to 4$ norm.

\subsubsection{Hypercontractivity and the injective tensor norm}

We are able to make progress in understanding both the complexity of the $2\ra 4$ norm and the quality of our SDP relaxation by relating the $2\ra 4$ norm to
several natural questions about tensors.  An $r$-tensor can be thought of as an $r$-linear form on $\R^n$, and the {\em injective tensor norm}
$\|\cdot\|_{\inj}$ of a tensor is given by maximizing this form over all unit vector inputs.   See \secref{ITN} for a precise definition. When $r=1$, this norm
is the 2-norm of a vector and when $r=2$, it is the operator norm (or $2\!\ra\! 2$-norm) of a matrix, but for $r=3$ it becomes NP-hard to calculate. One
motivation to study this norm comes from quantum mechanics, where computing it is equivalent to a number of long-studied problems concerning entanglement and
many-body physics~\cite{HarrowM10}.  More generally, tensors arise in a vast range of practical problems involving multidimensional data~\cite{vanLoan09} for
which the injective tensor norm is both of direct interest and can be used as a subroutine for other tasks, such as tensor decomposition~\cite{delaVegaKKV05}.

It is not hard to show that $\tfnorm{A}^4$ is actually equal to
$\norm{T}_{{\rm inj}}$ for some $4$-tensor $T=T_A$. Not all
$4$-tensors can arise this way, but we show that the injective tensor
norm problem for general tensors can be reduced to those of the form
$T_A$.  Combined with known results about the hardness of tensor
computations, this reduction implies the following hardness
result.  To formulate the theorem, recall that the Exponential Time
Hypothesis (ETH)~\cite{ImpagliazzoPZ98} states that 3-SAT instances of length $n$ require time
$\exp(\Omega(n))$ to solve.

\begin{theorem}[informal version]\label{thm:hardness}
  Assuming ETH, then for any $\eps,\delta$ satisfying $2\eps+\delta<1$,
  the $2\ra 4$ norm of an $m\times m$ matrix $A$ cannot be
  approximated to within a $\exp(\log^\eps(m))$ multiplicative factor in time less
  than $m^{\log^{\delta}(m)}$ time.  This hardness result holds even
  with $A$ is a projector.
\end{theorem}

While we are primarily concerned with the case of $\Omega(1)$ approximation factor, we note that poly-time approximations to within multiplicative factor
$1+1/n^{1.01}$ are not possible unless $\p=\np$.  This, along with
\thmref{hardness}, is restated more formally as \thmref{hardness-formal} in
\secref{hardness}. Theorem~\ref{thm:hardness} yields as a corollary that, assuming ETH, there is no polynomial-time algorithm
obtaining a good approximation for the $2\ra 4$ norm. We note that these results hold under weaker assumptions than the ETH; see \secref{hardness} as well.

Previously no hardness results were known for the $2\ra 4$ norm, or
any $p\ra q$ norm with $p<q$, even for calculating the norms exactly.
However, hardness of approximation results for $1+1/\poly(n)$
multiplicative error have been proved for other polynomial
optimization problems~\cite{Ben-TalN98}.


\subsection{Relation to the Unique Games Conjecture}

Our results and techniques have some relevance to  the unique games conjecture.  Theorem~\ref{thm:sse-hyper} shows that obtaining a good approximation for the
$2\to q$ norm is \smallsetexpansion hard, but Theorem~\ref{thm:subexp} shows that this problem is not ``that much harder'' than \uniquegames and
\smallsetexpansion since it too has a subexponential algorithm. Thus, the $2\to q$ problem is in some informal sense ``of similar flavor'' to the \uniquegames
/ \smallsetexpansion . On the other hand, we actually are able to show in Theorem~\ref{thm:hardness} \emph{hardness} (even if only quasipolynomial) of this
problem, whereas a similar result for \uniquegames or \smallsetexpansion would be a major breakthrough. So there is a sense in which these results can be
thought of as some ``positive evidence'' in favor of at least weak variants of the UGC. (We emphasize however that there are inherent difficulties in extending
these results for \uniquegames, and it may very well be that obtaining a multiplicative approximation to the $2\to 4$ of an operator is significantly harder
problem than \uniquegames or \smallsetexpansion .)  In contrast, our positive algorithmic results show that perhaps the $2\to q$ norm can be thought of as a
path to refuting the UGC. In particular we are able to extend our techniques to show a polynomial time algorithm can approximate the canonical hard instances
for \uniquegames considered in prior works.

\begin{theorem}\label{thm:i-gaps}
(Informal) Eight rounds of the SoS relaxation certifies that it is possible to satisfy at most $1/100$ fraction of the constraints of \uniquegames instances of
the ``quotient noisy cube'' and ``short code'' types considered in~\cite{RaghavendraS09c,KhotS09,KhotPS10,BarakGHMRS11}
\end{theorem}

These instances are the same ones for which previous works showed that weaker hierarchies such as ``SDP+Sherali Adams'' and ``Approximate Lasserre'' require
$\omega(1)$ rounds to certify that one cannot satisfy almost all the constraints~\cite{KhotV05,RaghavendraS10,KhotS09,BarakGHMRS11}. In fact, for the ``short
code'' based instances of \cite{BarakGHMRS11} there was no upper bound known better than $\exp(\log^{\Omega(1)} n)$ on the number of rounds required to certify
that they are not almost satisfiable, regardless of the power of the hierarchy used.

This is significant since the current best known algorithms for \uniquegames utilize SDP hierarchies~\cite{BarakRS11,GuruswamiS11b},\footnote{Both these works
showed SDP-hierarchy-based algorithms matching the performance of the subexponential algorithm of \cite{AroraBS10}. \cite{GuruswamiS11b} used the Lasserre
hierarchy, while \cite{BarakRS11} used the weaker ``SDP+Sherali-Adams'' hierarchy.} and the instances above were the only known evidence that polynomial time
versions of these algorithms do not refute the unique games conjecture. Our work also show that strong ``basis independent'' hierarchies such as Sum of
Squares~\cite{Parrilo00,Parrilo03} and Lasserre~\cite{Lasserre01} can in fact do better than the seemingly only slightly weaker variants.\footnote{The only
other result of this kind we are aware of is~\cite{KarlinMN11}, that show that Lasserre gives a better approximation ratio than the linear programming
Sherali-Adams hierarchy for the knapsack problem. We do not know if weaker semidefinite hierarchies match this ratio, although knapsack of course has a simple
dynamic programming based PTAS.}

\section{The SoS hierarchy} \label{sec:SoS}


For our algorithmic results in this paper we consider a semidefinite programming (SDP) hierarchy that we call the \emph{Sum of Squares} (SoS) hierarchy. This is not a novel algorithm and essentially the same hierarchies were
considered by many other researchers (see the survey \cite{Laurent09}).  Because different works sometimes used slightly different definitions, in this section
we formally define the hierarchy we use as well as explain the intuition behind it. While there are some subtleties involved, one can think of this hierarchy
as equivalent in power to the programs considered by Parrilo, Lasserre and others, while stronger than hierarchies such ``SDP+Sherali-Adams'' and ``Approximate
Lasserre'' considered in~\cite{RaghavendraS09c,KhotPS10,BarakRS11}.

The SoS SDP is a relaxation for polynomial equations. That is, we consider a system of the following form: maximize $P_0(x)$ over $x \in \R^n$ subject to
$P_i^2(x)=0$ for $i=1\ldots m$ and $P_0,\ldots,P_m$ polynomials of degree at most $d$.\footnote{This form is without loss of generality, as one can translate
an inequality constraint of the form $P_i(x) \geq 0$ into the equality constraint $(P_i(x) - y^2)^2=0$ where $y$ is some new auxiliary variable. It is useful to show equivalences between various hierarchy formulations; see also Appendix~\ref{app:hierarchies}.}  For $r \geq 2d$, the \emph{$r$-round SoS SDP} optimizes over $x_1,\ldots,x_n$ that can be thought of as formal
variables rather than actual numbers. For these formal variables, expressions of the form $P(x)$ are well defined and correspond to a real number (which can be
computed from the SDP solution) as long as $P$ is a polynomial of degree at most $r$. These numbers obey the \emph{linearity} property which is that $(P+Q)(x)
= P(x) + Q(x)$, and, most importantly, the \emph{positivity} property that $P^2(x)\geq 0$ for every polynomial $P$ of degree at most $r/2$.  These expressions
satisfy all initial constraints (i.e., $P_i^2(x)=0$ for $i=1\ldots m$) and the \emph{value} of the SDP is set to be the expression $P_0(x)$. The above means
that to show that the SoS relaxation has value at most $v$ it is sufficient to give any proof that derives from the constraints $\{ P_i^2(x) =0  \}_{i=1..m}$
the conclusion that $P_0(x) \leq v$ using only the linearity and positivity properties, without using any polynomials of degree larger than $r$ in the
intermediate steps. In fact, such a proof always has the form \be v - P_0(x) = \sum_{i=1}^k R_i(x)^2 + \sum_{i=1}^m P_i(x) Q_i(x), \label{eq:SOS-proof}\ee
where $R_1,\ldots,R_k,Q_1,\ldots,Q_m$ are arbitrary polynomials satisfying $\deg R_i \leq r/2, \deg P_iQ_i \leq r$. The polynomial $\sum_i R_i(x)^2$ is a SoS
(sum of squares) and optimizing over such polynomials (along with the $Q_1,\ldots,Q_m$) can be achieved with a semi-definite program.

\Bnote{removed: This shows why it is crucial that $k$ be unbounded (apart from the trivial $k\leq \binom{n+r-1}{r}$) so that we don't have to impose a rank
constraint on the matrices we optimize over.

I think the reader would find this confusing at this point}

\paragraph{Pseudo-expectation view}  For more intuition about the SoS hierarchy, one can imagine that instead of being formal variables, $x_1,\ldots,x_n$ actually correspond to correlated random variables $X_1,\ldots,X_n$ over $\R^n$, and the expression  $P(x)$ is set to equal the expectation $\E [ P(X) ]$. In this case, the linearity and positivity properties are obviously satisfied by these expressions, although other properties that would be obtained if $x_1,\ldots,x_n$ were simply numbers might not hold. For example, the property that $R(x) = P(x)Q(x)$ if $R=P\cdot Q$ does not necessarily hold, since its not always the case that $E[XY] = E[X]E[Y]$ for every three random variables $X,Y,Z$. So, another way to describe the $r$-round SoS hierarchy is that the expressions $P(x)$ (for $P$ of degree at most $r$) satisfy some of the constraints that would have been satisfied if these expressions corresponded to expectations over some correlated random variables $X_1,\ldots,X_N$. For this reason, we will use the notation $\pE_x P(x)$ instead of  $P(x)$ where we refer to the functional $\pE$ as a level-$r$ \emph{pseudo-expectation functional} (or $r$-\pef for short). Also, rather than describing $x_1,\ldots,x_n$ as formal variables, we will refer to them as level-$r$ \emph{fictitious random variables} (or $r$-\frv for short) since in some sense they look like true correlated random variables up to their $r^{\text{th}}$ moment.

We can now present our formal definition of pseudo-expectation and the SoS hierarchy:\footnote{We use the name ``Sum of Squares'' since the positivity
condition below is the most important constraint of this program. However, some prior works used this name for the \emph{dual} of the program we define here.
As we show in Appendix~\ref{app:hierarchies}, in many cases of interest to us there is no duality gap.}


\begin{definition}  Let $\pE$ be a functional that maps a polynomial $P$ over $\R^n$ of degree at most $r$ into a real number which we denote by $\pE_x P(x)$ or $\pE P$ for short.  We say that $\pE$ is a \emph{level-$r$ pseudo-expectation functional} ($r$-\pef for short)  if it satisfies:
\begin{description}

\item[Linearity] For every polynomials $P,Q$ of degree at most $r$ and $\alpha,\beta \in \R$, $\pE( \alpha P + \beta Q) = \alpha \pE P + \beta \pE Q$.

\item[Positivity] For every polynomial $P$ of degree at most $r/2$, $\pE P^2 \geq 0$.

\item[Normalization] $\pE 1 = 1$ where on the RHS, $1$ denotes the degree-$0$ polynomial that is the constant $1$.
\end{description}
\end{definition}

\begin{definition} \label{def:pseudo-SDP}
Let $P_0,\ldots, P_m$ be polynomials over $\R^n$ of degree at most $d$, and let $r \geq 2d$. The value of the  $r$-round SoS SDP for the program ``$\max P_0$ subject to $P_i^2=0$ for $i=1\ldots m$'', is equal to the maximum of $\pE P_0$ where $\pE$ ranges over all level $r$ pseudo-expectation functionals satisfying  $\pE P_i^2 = 0$ for $i=1\ldots m$.
\end{definition}

The functional $\pE$  can be represented by a table of size $n^{O(r)}$ containing the
pseudo-expectations of every monomial of degree at most $r$ (or some other
linear basis for polynomials of degree at most $r$). For a linear functional $\pE$, the map $P\mapsto \pE P^2$ is a quadratic
form. Hence, $\pE$ satisfies the positivity condition if and only if the corresponding quadratic form is positive semidefinite.
It follows that the convex set of level-$r$ pseudo-expectation functionals over $\R^n$ admits an $n^{O(r)}$-time
separation oracle, and hence the $r$-round SoS relaxation can be solved up to accuracy $\e$ in time
$(mn\cdot \log(1/\e))^{O(r)}$.

As noted above, for every random variable $X$ over $\R^n$, the functional $\pE P \seteq \E P(X)$ is
a level-$r$ pseudo-expectation functional for every $r$.  As $r\ra\infty$, this hierarchy of pseudo-expectations will converge to the
expectations of a true random variable~\cite{Lasserre01}, but the
convergence is in general not guaranteed to happen in a finite number of
steps~\cite{KlerkL11}.


Whenever there can be ambiguity about what are the variables of the polynomial $P$ inside an $r$-\pef $\pE$, we will use the notation $\pE_x P(x)$ (e.g.,
$\pE_x x_3^2$ is the same as $\pE P$ where $P$ is the polynomial $x \mapsto x_3^2$). As mentioned above, we call the inputs $x$ to the polynomial
\emph{level-$r$ fictitious random variables} or $r$-\frv for short.

\begin{remark} The main difference between the SoS hierarchy and weaker SDP hierarchies considered in the literature such as SDP+Sherali Adams and the Approximate Lasserre hierarchies~\cite{RaghavendraS09c,KhotPS10} is that the SoS hierarchy treats all polynomials equally and hence is agnostic to the choice of basis. For example, the approximate Lasserre hierarchy can also be described in terms of pseudo-expectations, but these pseudo-expectations are only defined for monomials, and are allowed some small error. While they can be extended linearly to other polynomials, for non-sparse polynomials that error can greatly accumulate.
\end{remark}

%
%

\subsection{Basic properties of pseudo-expectation}

For two polynomials $P$ and $Q$, we write $P \sle Q$ if $Q=P + \sum_{i=1}^m
R_i^2$ for some polynomials $R_1,\ldots,R_m$.

If $P$ and $Q$ have degree at most $r$, then $P \sle Q$ implies that $\pE P \leq \pE Q$ for every $r$-\pef $\pE$.  This follows using linearity and positivity, as
well as the (not too hard to verify) observation that if $Q-P = \sum_i R_i^2$ then it must hold that $\deg(R_i) \leq \max\{\deg(P),\deg(Q)\}/2$ for every $i$.

We would like to understand how polynomials behave on linear subspaces of
$\R^n$.
A map $P\from \R^n\to \R$ is \emph{polynomial} over a linear subspace
$V\sse\R^n$ if $P$ restricted to $V$ agrees with a polynomial in the
coefficients for some basis of $V$.
Concretely, if $g_1,\ldots,g_m$ is an (orthonormal) basis of $V$, then $P$
is \emph{polynomial} over $V$ if $P(f)$ agrees with a polynomial in
$\iprod{f,g_1},\ldots,\iprod{f,g_m}$.
We say that $P\sle Q$ holds over a subspace $V$ if $P-Q$, as a polynomial
over $V$, is a sum of squares.
\Anote{More generally, we could say that $P\sle_Z Q$ for some variety Z if $Q-P$ is a SoS plus an element from $I(Z)$.  This is what Parrilo et al do (following the various Positivstellensatzes). But perhaps we don't need this generalization.}
\begin{lemma}
  Let $P$ and $Q$ be two polynomials over $\R^n$ of degree at most $r$,
  and let $B\from \R^n\to \R^k$ be a linear operator.
  Suppose that $P\sle Q$ holds over the kernel of $B$.
  Then, $\pE P \le \pE Q$ holds for any $r$-\pef $\pE$ over $\R^n$
  that satisfies $\pE_f \snorm{Bf}=0$.
\end{lemma}

\begin{proof}
  Since $P\sle Q$ over the kernel of $B$, we can write
  $Q(f)=P(f)+\sum_{i=1}^m R_i^2(f)+\sum_{j=1}^k(B f)_j S_j(f)$ for
  polynomials $R_1,\ldots,R_m$ and $S_1,\ldots,S_k$ over $\R^n$.
  By positivity, $\pE_f R_i^2(f)\ge 0$ for all $i\in[m]$.
  We claim that $\pE_f (Bf)_j S_j(f)=0$ for all $j\in[k]$ (which would
  finish the proof).
  This claim follows from the fact that $\pE_f (B f)_j^2 =0$ for all
  $j\in[k]$ and \pref{lem:pseudo-expectation-cauchy-schwarz} below.
\end{proof}

\begin{lemma}[Pseudo Cauchy-Schwarz]
  \label{lem:pseudo-expectation-cauchy-schwarz}
  Let $P$ and $Q$ be two polynomials of degree at most $r$.
  Then, $\pE PQ\le \sqrt{\pE P^2}\cdot \sqrt{\pE Q^2}$ for any degree-$2r$
  pseudo-expectation functional $\pE$.
\end{lemma}

\begin{proof}
  We first consider the case $\pE P^2,\pE Q^2 >0$.
  Then, by linearity of $\pE$, we may assume that $\pE P^2=\pE Q^2=1$.
  Since $2 PQ\sle P^2 + Q^2$ (by expanding the square $(P-Q)^2$), it
  follows that $\pE PQ \le \tfrac12 \pE P^2 +\tfrac12\pE Q^2 =1$ as
  desired.
  It remains to consider the case $\pE P^2=0$.
  In this case, $2\alpha PQ \sle P^2+\alpha^2Q^2$ implies that $\pE PQ\le
  \alpha\cdot  \tfrac12\pE Q^2$ for all $\alpha>0$.
  Thus $\pE PQ=0$, as desired.
\end{proof}

\lemref{pseudo-expectation-cauchy-schwarz} also explains why our SDP in Definition~\ref{def:pseudo-SDP} is dual to the one in \eq{SOS-proof}.  If $\pE$ is a level-$r$ pseudo-expectation functional satisfying $\pE [P_i^2]=0$, then \lemref{pseudo-expectation-cauchy-schwarz} implies that $\pE [P_iQ_i]=0$ for all $Q_i$ with $\deg P_iQ_i\leq r$.

\Anote{Sorry, I wrote this thinking of the non-pseudo Cauchy-Schwarz:
  Another interpretation of \lemref{pseudo-expectation-cauchy-schwarz} is that any proof based solely on Cauchy-Schwarz of low-degree polynomials can be found by the SoS hierarchy.  We will give a concrete example of this principle in \secref{why-SOS}.}

Appendix~\ref{app:pseudo-expectation} contains some additional useful
facts about pseudo-expectation functionals.  In particular, we will make
repeated use of the fact that they satisfy another Cauchy-Schwarz
analogue: namely, for any level-2 f.r.v.'s $f,g$, we have $\pE_{f,g}
\iprod{f,g} \leq \sqrt{\pE_f \snorm{f}}\sqrt{\pE_g \snorm{g}}$.  This
is proven in \lemref{cauchy-Schwarz}.

\subsection{Why is this SoS hierarchy useful?}\label{sec:why-SOS}

Consider the following example. It is known that if $f:\{ \pm 1\}^{\ell}\to\R$ is a degree-$d$ polynomial then

\begin{equation}
9^d\left( \E_{w\in \{\pm 1\}^{\ell}} f(w)^2 \right)^2  \geq \E_{w\in \{\pm 1\}^n} f(w)^4 \mcom \label{eq:hypercontractive-poly}
 \end{equation} (see e.g.~\cite{ODonnell07}). Equivalently, the linear operator $\cP_d$ on $\R^{\{\pm 1\}^{\ell}}$ that projects a function into the degree $d$ polynomials satisfies $\tfnorm{\cP_d} \leq 9^{d/4}$. This fact is known as the hypercontractivity of low-degree polynomials, and was used in several integrality gaps results such as~\cite{KhotV05}. By following the proof of (\ref{eq:hypercontractive-poly}) we show in Lemma~\ref{lem:hypercontractivity1} that a stronger statement is true:
 \begin{equation}
 9^d\left( \E_{w\in \{\pm 1\}^{\ell}} f(w)^2 \right)^2    = \E_{w\in \{\pm 1\}^{\ell}} f(w)^4  + \sum_{i=1}^m Q_i(f)^2 \mcom \label{eq:hypercontractive-poly2}
 \end{equation}
 where the $Q_i$'s are polynomials of degree $\leq 2$ in the $\binom{\ell}{d}$ variables $\{ \Hat{f}(\alpha) \}_{\alpha \in \binom{[\ell]}{d}}$ specifying the coefficients of the polynomial $f$. By using the positivity constraints, (\ref{eq:hypercontractive-poly2}) implies that  (\ref{eq:hypercontractive-poly}) holds even in the $4$-round SoS relaxation where we consider the coefficients of $f$ to be given by $4$-f.r.v. This proves Theorem~\ref{thm:hyper-poly}, showing that the SoS relaxation certifies that $\tfnorm{\cP_d} \leq 9^{d/4}$.

\begin{remark}
Unfortunately to describe the result above, we needed to use the term ``degree'' in two different contexts. The SDP relaxation considers polynomial expressions
of degree at most $4$ \emph{in the coefficients of $f$}. This is a different notion of degree than the degree $d$ of $f$ itself as a polynomial over
$\R^{\ell}$. In particular the variables of this SoS program are the $\binom{\ell}{d}$ coefficients  $\{ \Hat{f}(\alpha) \}_{\alpha \in \binom{[\ell]}{d}}$.
Note that for every fixed $w$, the expression $f(w)$ is a linear polynomial over these variables, and hence the expressions $\left( \E_{w\in \{\pm 1\}^{\ell}}
f(w)^2 \right)^2$  and $\E_{w\in \{\pm 1\}^{\ell}} f(w)^4$ are degree $4$ polynomials over the variables.
\end{remark}

\Bnote{maybe add a remark here how when we use this as a corollary to argue about the SoS relaxation, that relaxation is given in another basis, and so we really use here the basis-independence of the algorithm.}

\Dnote{not clear what \pref{eq:hypercontractive-poly2} actually means.
  it makes sense only if we think of it over the subspace of low-degree
  Fourier polynomials.
  Again the notation $f \in \cS^4_{2^n}$ is problematic.
  }

While the proof of (\ref{eq:hypercontractive-poly2}) is fairly simple, we find the result--- that hypercontractivity of polynomials is efficiently certifiable---somewhat surprising. The reason is that hypercontractivity serves as the basis of the integrality gaps results which are exactly instances of maximization problems where the objective value is low but this is supposedly hard to certify.  In particular, we consider integrality gaps for \uniquegames considered before in the literature. All of these instances follow the framework initiated by Khot and Vishnoi~\cite{KhotV05}. Their idea was inspired by \uniquegames hardness proofs, with the integrality gap obtained by composing an initial instance with a gadget. The proof that these instances have ``cheating'' SDP solutions is obtained by ``lifting'' the completeness proof of the gadget. On the other hand, the soundness property of the gadget, combined with some isoperimetric results, showed that the instances do not have real solutions. This approach of lifting completeness proofs of reductions was used to get other integrality gap results as well~\cite{Tulsiani09}. We show that the SoS hierarchy allows us to lift a  certain \emph{soundness} proof for these instances, which includes a (variant of) the invariance principle of~\cite{MosselOO05}, influence-decoding a la~\cite{KhotKMO04}, and hypercontractivity of low-degree polynomials. It turns out all these results can be proven via sum-of-squares type arguments and hence lifted to the SoS hierarchy.

\section{Overview of proofs} \label{sec:overview}

We now give a very high level overview of the tools we use to obtain our results, leaving details to the later sections and appendices.

\subsection{Subexponential algorithm for the 2-to-q norm}

Our subexponential algorithm for obtaining a good approximation for the $2\to q$ norm is extremely simple. It is based on the observation that a subspace $V
\subseteq \R^n$ of too large a dimension must contain a function $f$ such that $\norm{f}_q \gg \norm{f}_2$. For example, if $\dim(V) \gg \sqrt{n}$, then there
must be $f$ such that $\norm{f}_4 \gg \norm{f}_2$. This means that if we want to distinguish between, say, the case that $\tfnorm{V} \leq 2$ and $\tfnorm{V}
\geq 3$,  then we can assume without loss of generality that $\dim(V) = O(\sqrt{n})$ in which case we can solve the problem in $\exp(O(\sqrt{n}))$ time. To get
intuition, consider the case that $V$ is spanned by an orthonormal basis $f^1,\ldots,f^d$ of functions whose entries are all in $\pm 1$. Then clearly we can
find coefficients $a_1,\ldots,a_d \in \{ \pm 1\}$ such that the first coordinate of $g=\sum a_j f^j$ is equal to $d$, which means that its $4$-norm is at least
$(d^4/n)^{1/4} = d/n^{1/4}$. On the other hand, since the basis is orthonormal, the $2$-norm of $g$ equals $\sqrt{d}$ which is $\ll d/n^{1/4}$ for $d \gg
\sqrt{n}$.

Note the similarity between this algorithm and \cite{AroraBS10}'s algorithm for \smallsetexpansion , that also worked by showing that if the dimension of the top eigenspace of a graph is too large then it cannot be a small-set expander. Indeed, using our reduction of \smallsetexpansion to the $2\to q$ norm, we can reproduce a similar result to~\cite{AroraBS10}.

\subsection{Bounding the value of SoS relaxations}

We show that in several cases, the SoS SDP hierarchy gives strong
bounds on various instances. At the heart of these results is a
general approach of ``lifting'' proofs about one-dimensional objects
into the SoS relaxation domain.  Thus we transform the prior proofs
that these instances have small objective value, into a proof that the
SoS relaxation also has a small objective  The crucial observation is
that many proofs boil down to the simple fact that a sum of squares of
numbers is always non-negative. It turns out that this ``sum of
squares'' axiom is surprisingly powerful (e.g. implying a version of
the Cauchy--Schwarz inequality given by
\lemref{cauchy-Schwarz})\Anote{I changed this from the
  pseudo-expectation C-S, to the normal C-S, since I think that's the
  one that used in proofs to say nontrivial things, whereas the
  pseudo-expectation C-S is used to say things like $P^2=0$ implies
  $PQ=0$ for all $Q$.}, and many proofs boil down to essentially this principle.

\subsection{The 2-to-4 norm and small-set expansion}

Bounds on the $p\to q$ norm of operators for $p<q$ have been used to show fast convergence of Markov chains. In particular, it is known that if the projector to the top eigenspace of a graph $G$ has bounded $2\to 4$ norm, then that graph is a \emph{small-set expander} in the sense that sets of $o(1)$ fraction of the vertices have most of their edges exit the set. In this work we show a converse to this statement, proving that if $G$ is a small-set expander, then the corresponding projector has bounded $2\to 4$ norm. As mentioned above, one corollary of this result is that a good (i.e., dimension-independent) approximation to the $2\to 4$ norm will refute the Small-Set Expansion hypothesis of~\cite{RaghavendraS10}.

We give a rough sketch of the proof. Suppose that $G$ is a sufficiently strong small-set expander, in the sense that every set $S$ with $|S| \leq \delta |V(G)|$ has all but a tiny fraction of the edges $(u,v)$ with $u\in S$ satisfying $v\not\in S$. Let $f$ be a function in the eigenspace of $G$ corresponding to eigenvalues larger than, say $0.99$. Since  $f$ is in the top eigenspace, for the purposes of this sketch let's imagine that it satisfies 
\be \forall x\in V, \E_{y\sim x} f(y) \geq 0.9 f(x)
\label{eq:SSE-cond},\ee
 where the expectation is over a random neighbor $y$ of $x$.  Now, suppose that  $\E f(x)^2 = 1$ but $\E f(x)^4 = C$ for some $C \gg \poly(1/\delta)$. That means that most of the contribution to the $4$-norm comes from the set $S$ of vertices $x$ such that $f(x) \geq (1/2)C^{1/4}$, but $|S| \ll \delta |V(G)|$. Moreover, suppose for simplicity that $f(x) \in ((1/2)C^{1/4},2C^{1/4})$, in which case the condition \textbf{(*)}  together with the small-set expansion condition that for most vertices $y$ in $\Gamma(S)$ (the neighborhood of $S$) satisfy $f(y) \geq C^{1/4}/3$, but the small-set expansion condition, together with the regularity of the graph imply that $|\Gamma(S)| \geq 200|S|$ (say), which implies that $\E f(x)^4 \geq 2C$---a contradiction.

The actual proof is more complicated, since we can't assume the condition \eq{SSE-cond}.  Instead we will approximate it  it by assuming that $f$ is the function in the top eigenspace that \emph{maximizes} the ratio $\norm{f}_4/\norm{f}_2$.   See Section~\ref{sec:sse} for the details.

%

\subsection{The 2-to-4 norm and the injective tensor norm}

To relate the $2\to 4$ norm to the injective tensor norm, we start by  establishing equivalences between the $2\ra 4$ norm and a variety of different tensor
problems. Some of these are straightforward exercises in linear algebra, analogous to proving that the largest eigenvalue of $M^TM$ equals the square of the
operator norm of $M$.

One technically challenging reduction is between the problem of
optimizing a general degree-4 polynomial $f(x)$ for $x\in\R^n$ and a
polynomial that can be written as the sum of fourth powers of linear
functions of $x$.
Straightforward approaches will magnify errors by $\poly(n)$ factors,
which would make it impossible to rule out a PTAS for the $2\ra 4$
norm.  This would still be enough to prove that a $1/\poly(n)$
additive approximation is \np-hard.
However, to handle constant-factor approximations, we will instead use a variant of a reduction in \cite{HarrowM10}. This will allow us to map a general tensor optimization problem (corresponding to a general degree-4 polynomial) to a $2\ra 4$ norm calculation without losing very much precision.

To understand this reduction, we first introduce the $n^2\times n^2$ matrix $A_{2,2}$ (defined in \secref{ITN}) with
the property that $\tfnorm{A}^4=\max z^T A_{2,2}z$, where the maximum
is taken over unit vectors $z$ that can be written in the form
$x \otimes y$.  Without this last restriction, the maximum would
simply be the operator norm of $A_{2,2}$.  Operationally, we can think
of $A_{2,2}$ as a quantum measurement operator, and vectors of the
form $x\otimes y$ as unentangled states (equivalently we say that
vectors in this form are tensor product states, or simply ``product
states'').  Thus the difference
between $\tfnorm{A}^4$ and $\|A_{2,2}\|_{2\ra 2}^2$ can be thought of as
the extent to which the measurement $A_{2,2}$ can notice the difference between
product states and (some) entangled state.

Next, we define a matrix $A'$ whose rows are of the form $(x'\ot
y')^*\sqrt{A_{2,2}}$, where $x',y'\in\bbR^n$ range over a distribution
that approximates the uniform distribution.  If $A'$ acts on a vector
of the form $x \ot y$, then the maximum output 4-norm (over $L_2$-unit
vectors $x,y$) is precisely $\|A\|_{2\ra 4}$.  Intuitively, if $A'$
acts on a highly ``entangled" vector $z$, meaning that $\iprod{z,x\ot
  y}$ is small for all unit vectors $x,y$, then $\|A'z\|_4$ should be
small.  This is because $z$ will have small overlap with $x' \ot y'$,
and $A_{2,2}$ is positive semi-definite, so its off-diagonal entries
can be upper-bounded in terms of its operator norm.  These arguments
(detailed in \secref{hardness}) lead to only modest bounds on $A'$,
but then we can use an amplification argument to make the $2\ra 4$
norm of $A'$ depend more sensitively on that of $A$, at the cost of
blowing up the dimension by a polynomial factor.

The reductions we achieve also permit us, in \secref{BCY-app}, to relate our \tensorsdp
algorithm with the sum-of-squares relaxation used by
Doherty, Parrilo, and Spedalieri~\cite{DohertyPS04} (henceforth DPS).
We show the two relaxations are essentially equivalent, allowing us to import results proved, in some cases, with techniques from quantum information theory.
One such result, from \cite{BrandaoCY11}, requires relating $A_{2,2}$ to a quantum measurement of the 1-LOCC form.  This means that there are two $n$-dimensional subsystems, combined via tensor products, and $A_{2,2}$ can be implemented as a measurement on the first subsystem followed by a measurement on the second subsystem that is chosen conditioned on the results of the first measurement.  The main result of \cite{BrandaoCY11} proved that such 1-LOCC measurements exhibit much better behavior under DPS, and they obtain nontrivial approximation guarantees with only $O(\log(n)/\eps^2)$ rounds.   Since this is achieved by DPS, it also implies an upper bound on the error
of \tensorsdp.  This upper bound is $\eps Z$, where $Z$ is the smallest number for which $A_{2,2}\leq Z M$ for some 1-LOCC measurement $M$.  While $Z$ is not believed to be efficiently computable, it is at least $\|A_{2,2}\|_{2\ra 2}$, since any measurement $M$ has $\|M\|_{2\ra 2}\leq 1$.
To upper bound $Z$, we can explicitly construct $A_{2,2}$ as a quantum
measurement.  This is done by the following protocol.  Let
$a_1,\ldots,a_m$ be the rows of $A$.  One party performs the quantum
measurement with outcomes $\{\alpha a_ia_i^T\}_{i=1}^m$ (where
$\alpha$ is a normalization factor) and transmits the outcome $i$ to
the other party.  Upon receiving message $i$, the second party does
the two outcome measurement $\{\beta a_ia_i^T, I -\beta a_i a_i^T\}$
and outputs 0 or 1 accordingly, where $\beta$ is another normalization
factor.  The measurement $A_{2,2}$ corresponds to the ``0'' outcomes.
For this to be a physically realizable 1-LOCC measurement, we need
$\alpha\leq \|A^TA\|_{2\ra 2}$ and $\beta\leq \|A\|_{2\ra\infty}^2$.
Combining these ingredients, we obtain the approximation guarantee in
\thmref{BCY}.  More details on this argument are in \secref{BCY-proof}.


\Dnote{section about organization of the paper is iffalsed.}

\Anote{This should be a section or subsection, not a section*.}

\subsection{Definitions and Notation} \label{sec:prelim:short}

Let $\cU$ be some finite set. For concreteness, and without loss of generality, we can let $\cU$ be the set $\{1,\ldots,n\}$, where $n$ is some positive
integer. We write $\E_\cU f$ to denote the average value of a function $f\from \cU\to \R$ over a random point in $\cU$ (omitting the subscript $\cU$ when it is
clear from the context).  We let $\L2(\cU)$ denote the space of functions $f\from \cU \to \R$ endowed with the inner product $\iprod{f,g}=\E_{\cU} fg$ and its
induced norm $\norm{f}=\iprod{f,f}^{1/2}$. For $p \geq 1$, the \emph{$p$-norm} of a function $f\in \L2(\cU)$ is defined as $\norm{f}_p \seteq \Paren{\E |f|^p
}^{1/p}$. A convexity argument shows $\norm{f}_p \leq \norm{f}_q$ for $p\leq q$. If $A$ is a linear operator mapping functions from $\L2(\cU)$ to $\L2(\cV)$,
and $p,q \geq 1$, then the $p$-to-$q$ norm of $A$ is defined as $\ptoqnorm{A}{p}{q} = \max_{0\neq f\in \L2(\cU)} \norm{Af}_q/\norm{f}_p$. If $V \subseteq
\L2(\cU)$ is a linear subspace, then we denote $\ptoqnorm{V}{p}{q} = \max_{f\in V} \norm{f}_q/\norm{f}_p$.

\paragraph{Counting norms} In most of this paper we use
\emph{expectation norms} defined as above, but in some contexts the
\emph{counting norms} will be more convenient. We will usually stick to the convention that \emph{functions} use expectation norms while \emph{vectors} use the counting norms. For a vector $v\in \C^{\cU}$ and $p \geq 1$, the \emph{$p$ counting norm} of $v$, denoted $\cnorm{v}_p$, is defined to be $\left( \sum_{i\in\cU}  |v_i|^p \right)^{1/p}$. The \emph{counting inner product} of two vectors $u,v \in \R^{\cU}$, denoted as $\ciprod{u,v}$, is defined to be $\sum_{i\in\cU} u_iv_i^*$.

\section{The \tensorsdp algorithm}

There is a very natural SoS program for the $2\to 4$ norm for a given linear operator $A\from \L2(\cU)\to \L2(\cV)$:

\begin{center}
\fbox{\begin{minipage}{3in}
\textsc{Algorithm $\tensorsdp^{(d)}(A)$:}
\begin{description}
\item[\quad] Maximize $\pE_f  \norm{Af}_4^4 $ subject to
  \begin{itemize}
  \item $f$ is a $d$-\frv over $\L2(\cU)$,
  \item $\pE_f (\snorm{f} - 1)^2=0$.
  \end{itemize}
\end{description}
\end{minipage}}
\end{center}

Note that $\norm{Af}_4^4$ is indeed a degree $4$ polynomial in the
variables $\{ f(u) \}_{u\in\cU}$. The $\tensorsdp^{(d)}$ algorithm makes
sense for $d\geq 4$, and we denote by $\tensorsdp$ its most basic version
where $d=4$. The \tensorsdp algorithm applies not just to the $2\to 4$
norm, but to optimizing general polynomials over the unit ball of $\L2(\cU)$ by replacing $\norm{Af}_4^4$ with an arbitrary polynomial $P$.

While we do not know the worst-case performance of the \tensorsdp algorithm, we do know that it performs well on random instances (see Section~\ref{sec:random}), and (perhaps more relevant to the UGC) on the projector to low-degree polynomials (see Theorem~\ref{thm:hyper-poly}). The latter is a corollary of the following result:

\Dnote{we should explicitly state and prove that our sdp relaxation works for the
  projector to low-degree polynomials. maybe best to add this as a
  corollary of the lemma below.}

\begin{lemma}\label{lem:hypercontractivity1}
  Over the space of $n$-variate Fourier polynomials\footnote{%
    An $n$-variate Fourier polynomial
    with degree at most $d$ is a function $f\from\sbits^n\to\R$ of the form
    \begin{math}
      f=\sum_{\alpha\subseteq [n], |\alpha|\leq d}
      \Hat{f}_\alpha\chi_{\alpha}
    \end{math}
    where $\chi_{\alpha}(x)=\prod_{i\in\alpha}x_i$.  } %
  $f$ with degree at most $d$,
  \begin{displaymath}
    \E f^4 \sle 9^{d} \Paren{ \E f^2 }^2\mcom
\end{displaymath}
 where the expectations are over
  $\sbits^n$.
\end{lemma}


\begin{proof}
  The result is proven by a careful variant of the standard inductive proof
  of the hypercontractivity for low-degree polynomials (see
  e.g.~\cite{ODonnell07}).
  We include it in this part of the paper since it is the simplest example
  of how to ``lift'' known proofs about functions over the reals into
  proofs about the fictitious random variables that arise in semidefinite
  programming hierarchies.
  To strengthen the inductive hypothesis, we will prove the more general
  %
  %
  statement that for $f$ and $g$ being $n$-variate Fourier polynomials with
  degrees at most $d$ and $e$, it holds that $\E f^2g^2 \sle 9^{\frac{d+e}2}
  \left(\E f^2\right)\left( \E g^2 \right)$.
  (Formally, this polynomial relation is over the linear space of pairs of
  $n$-variate Fourier polynomials $(f,g)$, where $f$ has degree at most $d$
  and $g$ has degree at most $e$.)
  The proof is by induction on the number of variables.

  If one of the functions is constant (so that $d=0$ or $e=0$), then $\E
  f^2g^2 = (\E f^2)(\E g^2)$, as desired.
  Otherwise, let $f_0,f_1,g_0,g_1$ be Fourier polynomials depending only on
  $x_1,\ldots,x_{n-1}$ such that $f(x)=f_0(x) + x_nf_1(x)$ and $g(x)=g_0(x) +
  x_ng_1(x)$.
  The Fourier polynomials $f_0,f_1,g_0,g_1$ depend linearly on $f$ and $g$
  (because $f_0(x)=\tfrac12 f(x_1,\ldots,x_{n-1},1)+\tfrac12
  f(x_1,\ldots,x_{n-1},-1)$ and $f_1(x)=\tfrac12
  f(x_1,\ldots,x_{n-1},1)-\tfrac12 f(x_1,\ldots,x_{n-1},-1)$).
  Furthermore, the degrees of $f_0$, $f_1$, $g_0$, and $g_1$ are at most
  $d$, $d-1$, $e$, and $e-1$, respectively.

  Since $\E x_n  = \E x_n^3 = 0$, if we expand $\E f^2g^2 = \E (f_0 +
  x_n f_1)^2(g_0 + x_n g_1)^2$ then the terms where $x_n$ appears in an odd
  power vanish, and we obtain
\[
\E f^2g^2 = \E f_0^2g_0^2 + f_1^2g_1^2 + f_0^2g_1^2 + f_1^2g_0^2 + 4f_0f_1g_0g_1
\]
By expanding the square expression $2\E (f_0f_1-g_0g_1)^2$, we get $4\E f_0f_1g_0g_1  \sle 2\E f_0^2g_1^2 + f_1^2g_0^2$ and thus
\begin{equation}
\E f^2g^2 \sle \E f_0^2g_0^2 + \E f_1^2g_1^2 + 3 \E f_0^2g_1^2 + 3 \E f_1^2g_0^2    \mper \label{eq:hyper-cont-poly1a}
\end{equation}
Applying the induction hypothesis to all four terms on the right-hand side
of \pref{eq:hyper-cont-poly1a} (using for the last two terms that the
degree of $f_1$ and $g_1$ is at most $d-1$ and $e-1$),
\begin{align*}
  \E f^2g^2 %
  & \sle 9^{\frac{d+e}2} \left(\E f_0^2 \right)\left( \E g_0^2 \right) %
  +  9^{\frac{d+e}2}\left(\E f_1^2 \right)\left(\E g_1^2\right)\\
  &\qquad + 3\cdot 9^{\frac{d+e}2 -1/2} \Paren{\E f_0^2}\Paren{ \E g_1^2 }
  +
  3\cdot  9^{\frac{d+e}2-1/2} \left(\E f_1^2\right)\left( \E g_0^2 \right)\\
  &= 9^{\frac{d+e}2}\left( \E f_0^2 + \E f_1^2 \right) \left( \E g_0^2 + \E
    g_1^2 \right) \mper
\end{align*}
Since  $\E f_0^2 + \E f_1^2=\E (f_0+x_n f_1)^2=\E f^2$ (using $\E x_n=0$)
and similarly $\E g_0^2+\E g_1^2=\E g^2$, we derive the desired relation
$\E f^2g^2 \sle 9^{\frac{d+e}2}\Paren{\E f^2}\Paren{\E g^2}$.
\end{proof}

\section{SoS succeeds on Unique Games integrality gaps}

%
%

In this section we prove Theorem~\ref{thm:i-gaps}, showing that 8 rounds of the SoS hierarchy can beat the Basic SDP program on the canonical integrality gaps considered in the literature.

\begin{theorem}[Theorem~\ref{thm:i-gaps}, restated] \label{thm:gaps}   For sufficiently small $\e$ and large $k$, and every $n\in \N$, let $\cW$ be an $n$-variable $k$-alphabet \uniquegames instance of the type considered in~\cite{RaghavendraS09c,KhotS09, KhotPS10} obtained by composing the ``quotient noisy cube'' instance of \cite{KhotV05} with the long-code alphabet reduction of~\cite{KhotKMO04} so that the best assignment to $\cW$'s variable satisfies at most an $\e$ fraction of the constraints.  Then, on input $\cW$, eight rounds of the SoS relaxation outputs at most $1/100$.
\end{theorem}

\subsection{Proof sketch of Theorem~\ref{thm:gaps}}

The proof is very technical, as it is obtained by taking the already rather technical proofs of soundness for these instances, and ``lifting'' each step into the SoS hierarchy, a procedure that causes additional difficulties. \Bnote{removed: Thus in this extended abstract we only provide a proof sketch, leaving the details to Appendix~\ref{sec:refut-inst-uniq}.}  The high level structure of all integrality gap instances constructed in the literature was the following: Start with a basic integrality gap instance of \uniquegames where the Basic SDP outputs $1-o(1)$ but the true optimum is $o(1)$, the \emph{alphabet size} of $\cU$ is (necessarily) $R=\omega(1)$. Then, apply an \emph{alphabet-reduction gadget} (such as the long code, or in the recent work~\cite{BarakGHMRS11} the so called ``short code'') to transform $\cU$ into an instance $\cW$ with some constant alphabet size $k$. The soundness proof of the gadget guarantees that the  true optimum of $\cU$ is small, while the analysis of previous works managed to ``lift'' the completeness proofs, and argue that the instance $\cU$ survives a number of rounds that tends to infinity as  $\e$  tends to zero, where $(1-\e)$ is the completeness value in the gap constructions, and exact tradeoff between number of rounds and $\e$ depends on the paper and hierarchy.

The fact that the basic instance $\cU$ has small integral value can be shown by appealing to hypercontractivity of low-degree polynomials, and hence can be ``lifted'' to the SoS world via Lemma~\ref{lem:hypercontractivity1}. The bulk of the technical work is in lifting the soundness proof of the gadget. On a high level this proof involves the following components:
\textbf{(1)} The \emph{invariance principle} of~\cite{MosselOO05}, saying that low influence functions cannot distinguish between the cube and the sphere; this allows us to argue that functions that perform well on the gadget must have an influential variable, and \textbf{(2)} the \emph{influence decoding} procedure of~\cite{KhotKMO04} that maps these influential functions on each local gadget into a good global assignment for the original instance $\cU$.

The invariance principle poses a special challenge, since the proof of~\cite{MosselOO05} uses so called ``bump'' functions which are not at all low-degree polynomials.\footnote{A similar, though not identical, challenge arises in~\cite{BarakGHMRS11} where they need to extend the invariance principle to the ``short code'' setting. However, their solution does not seem to apply in our case, and we use a different approach.} We use a weaker invariance principle, only showing that the $4$ norm of a low influence function remains the same between two probability spaces that agree on the first $2$ moments.  Unlike the usual invariance principle, we do not move between Bernoulli variables and Gaussian space, but rather between two different distributions on the discrete cube. It turns out that for the purposes of these \uniquegames integrality gaps, the above suffices. The lifted invariance principle is proven via a ``hybrid'' argument similar to the argument of~\cite{MosselOO05}, where hypercontractivity of low-degree polynomials again plays an important role.

The soundness analysis of~\cite{KhotKMO04} is obtained by replacing each local function with an average over its neighbors, and then choosing a random influential coordinate from the new local function as an assignment for the original uniquegames instance. We follow the same approach, though even simple tasks such as independent randomized rounding turn out to be much subtler in the lifted setting. However, it turns out that by making appropriate modification to the analysis, it can be lifted to complete the proof of Theorem~\ref{thm:i-gaps}.

In the following, we give a more technical description of the proof.
%
Let $T_{1-\eta}$ be the $\eta$-noise graph on $\sbits^R$.
Khot and Vishnoi \cite{KhotV05} constructed a unique game $\cU$ with
label-extended graph $T_{1-\eta}$.
A solution to the level-$4$ SoS relaxation of $\cU$ is $4$-f.r.v. $h$ over $\L2(\sbits^R)$.
This variable satisfies $h(x)^2\equiv_h h(x)$ for all $x\in \sbits^R$ and
also $\pE_h (\E h)^2\le 1/R^2$.
(The variable $h$ encodes a $0/1$ assignment to the vertices of the
label-extended graph.
A proper assignment assigns $1$ only to a $1/R$ fraction of these
vertices.)
\pref{lem:expansion-bound} allows us to bound the objective value of the
solution $h$ in terms of the fourth moment $\pE_h \E (P_{>\lambda}h)^4$,
where $P_{>\lambda}$ is the projector into the span of the eigenfunctions
of $T_{1-\eta}$ with eigenvalue larger than $\lambda\approx 1/R^{\eta}$.
(Note that $\E(P_{>\lambda} h)^4$ is a degree-$4$ polynomial in $h$.)
For the graph $T_{1-\eta}$, we can bound the degree of $P_{>\lambda} h$ as
a Fourier polynomial (by about $\log(R)$).
Hence, the hypercontractivity bound (\pref{lem:hypercontractivity1}) allows
us to bound the fourth moment $\pE_h\E(P_{>\lambda} h)^4\le \pE_h(\E
h^2)^2$.
\Dnote{is some factor like $R^{0.1}$ missing in this bound?}
By our assumptions on $h$, we have $\pE_h(\E h^2)^2=\pE_h(\E h)^2\le
1/R^2$.
Plugging these bounds into the bound of \pref{lem:expansion-bound}
demonstrates that the objective value of $h$ is bounded by
$1/R^{\Omega(\eta)}$ (see \pref{thm:certify-sse}).

Next, we consider a unique game $\cW$ obtained by composing the unique game
$\cU$ with the alphabet reduction of \cite{KhotKMO04}.
Suppose that $\cW$ has alphabet $\Omega=\set{0,\ldots,k-1}$.
The vertex set of $\cW$ is $V\times \Omega^R$ (with $V$ being the
vertex set of $\cU$).
Let $f=\set{f_u}_{u\in V}$ be a solution to the level-$8$ SoS
relaxation of $\cW$.
To bound the objective value of $f$, we derive from it a level-$4$ random
variable $h$ over $L_2(V\times [R])$.
(Encoding a function on the label-extended graph of the unique game $\cU$.)
We define $h(u,r)=\super \Inf{\le \ell}_r \bar f_u$, where $\ell\approx
\log k$ and $\bar f_u$ is a variable of $L_2(\Omega^R)$ obtained by
averaging certain values of $f_u$ (``folding'').
It is easy to show that $h^2\sle h$ (using \pref{lem:boundedness-relation})
and $\pE_h (\E h)^2\le \ell/R$ (bound on the total influence of low-degree
Fourier polynomials).
\pref{thm:ug-value-bound} (influence decoding) allows us to bound the
objective value of $f$ in terms of the correlation of $h$ with the
label-extended graph of $\cU$ (in our case, $T_{1-\eta}$).
Here, we can use again \pref{thm:certify-sse} to show that the correlation
of $h$ with the graph $T_{1-\eta}$ is very small.
(An additional challenge arises because $h$ does not satisfy $h^2\equiv_h h$, but only the weaker condition $h^2\sle h$.
\pref{cor:independent-rounding} fixes this issue by simulating independent
rounding for fictitious random variables.)
To prove \pref{thm:ug-value-bound} (influence decoding), we analyze the
behavior of fictitious random variables on the alphabet-reduction gadget of
\cite{KhotKMO04}.
This alphabet-reduction gadget essentially corresponds to the $\e$-noise
graph $T_{1-\e}$ on $\Omega^R$.
Suppose $g$ is a fictitious random variables over $L_2(\Omega^R)$ satisfying $g^2\sle g$.
By \pref{lem:expansion-bound}, we can bound the correlation of $g$ with the graph $T_{1-\e}$ in terms of the fourth moment of  $P_{>\lambda} g$.
At this point, the hypercontractivity bound
(\pref{lem:hypercontractivity1}) is too weak to be helpful.
Instead we show an ``invariance principle'' result (\pref{thm:invariance-fourth-moment}),
which allows us to relate the fourth moment of $P_{>\lambda} g$ to the
fourth moment of a nicer random variable and the influences of $g$.


\medskip

\paragraph{Organization of the proof}
We now turn to the actual proof of Theorem~\ref{thm:gaps}. The proof consists of lifting to the SoS hierarchy all the steps used in the analysis of previous
integrality gaps, which themselves arise from hardness reductions.  We start in Section~\ref{sec:invariance-principle-variant} by showing a sum-of-squares
proof for a weaker version of  \cite{MosselOO05}'s invariance principle.  Then in Section~\ref{sec:interl-indep-round} we show how one can perform independent
rounding in the SoS world (this is a trivial step in proofs involving true random variables, but becomes much more subtle when dealing with SoS solutions). In
Sections~\ref{sec:dict-test-sse} and~\ref{sec:ug-dictatorship-test} we lift variants of the~\cite{KhotKMO04} dictatorship test. The proof uses a SoS variant of
influence decoding, which is covered in Section~\ref{sec:influence-decoding}.  Together all these sections establish SoS analogs of the soundness properties of
the hardness reduction used in the previous results. Then, in Section~\ref{sec:small-set-expansion} we show that analysis of the basic instance has a sum of
squares proof (since it is based on hypercontractivity of low-degree polynomials).  Finally in Section~\ref{sec:putt-things-togeth} we combine all these tools
to conclude the proof. In Section~\ref{sec:short-code} we discuss why this proof applies (with some modifications) also to the  ``short-code'' based instances
of \cite{BarakGHMRS11}.

\subsection{Invariance Principle for Fourth Moment}
\label{sec:invariance-principle-variant}

In this section, we will give a sum-of-squares proof for a variant of the
invariance principle of \cite{MosselOO05}.
Instead of for general smooth functionals (usually constructed from ``bump
functions''), we show invariance only for the fourth moment.
It turns out that invariance of the fourth moment is enough for our
applications.

Let $k=2^t$ for $t\in\N$ and let $\cX=(\cX_1,\ldots,\cX_R)$ be an
independent sequence\footnote{An \emph{orthonormal ensemble} is a
  collection of orthonormal real-valued random variables, one being the
  constant $1$.
  A sequence of such ensembles is \emph{independent} if each ensemble is
  defined over an independent probability space. (See \cite{MosselOO05} for
  details.)} of orthonormal ensembles $\cX_r=(X_{r,0},\ldots,X_{r,k-1})$.
Concretely, we choose $X_{r,i}=\chi_i(x_r)$, where
${\chi_0,\ldots,\chi_{k-1}}$ is the set of characters of $\GF 2^t$ and $x$
is sampled uniformly from $(\GF 2 ^t)^R$.
\Dnote{maybe define $\chi_i$ explicity as the character corresponding to
  $i$ as a bitstring}
Every random variable over $(\GF 2^t)^R$ can be expressed as a multilinear
polynomial over the sequence $\cX$.
In this sense, $\cX$ is maximally dependent.
On the other hand, let $\cY=(\cY_1,\ldots,\cY_R)$ be a sequence of
ensembles $\cY_r=(Y_{r,0},\ldots,Y_{r,k-1})$, where $Y_{r,0}\equiv 1$ and
$Y_{r,j}$ are independent, unbiased $\sbits$ Bernoulli variables.
The sequence $\cY$ is maximally independent since it consists of
completely independent random variables.

Let $f$ be a $4$-\frv over the space of multilinear polynomials with degree
at most $\ell$ and monomials indexed by $[k]^R$.
Suppose $\pE_f \norm{f}^4\le 1$.
(In the following, we mildly overload notation and use $[k]$ to denote the set
$\set{0,\ldots,k-1}$.)
Concretely, we can specify $f$ by the set of monomial coefficients
$\set{\hat f_\alpha}_{\alpha \in [k]^R,~\card{\alpha}\le \ell}$, where
$\card{\alpha}$ is the number of non-zero entries in $\alpha$.
As usual, we define $\Inf_r f=\sum_{\alpha\in[k]^R,~\alpha_r\neq 0} \hat
f_\alpha^2$.
Note that $\Inf_r f$ is a degree-2 polynomial in $f$.
(Hence, the pseudo-expectation of $(\Inf_r f)^2$ is defined.)

\Dnote{add some sentence introducing the theorem}

\begin{theorem}[Invariance Principle for Fourth Moment]
\label{thm:invariance-fourth-moment}
  For $\tau=\pE_f \sum_r (\Inf_r f)^2$,
  \begin{displaymath}
    \pE_f \E_\cX f^4 = \pE_f \E_\cY f^4 \pm k^{O(\ell)}\sqrt\tau \mper
  \end{displaymath}
\end{theorem}
\noindent (Since the expressions $\E_\cX f^4$ and $\E_\cY f^4$ are degree-$4$ polynomials
in $f$, their pseudo-expectations are defined.)

Using the SoS proof for hypercontractivity of low-degree polynomials
(over the ensemble $\cY$), the fourth moment $\pE_f \E_\cY f^4$ is
bounded in terms of the second moment $\pE_f\E_\cY f^2$. 
Since the first two moments of the ensembles $\cX$ and $\cY$ match, we
have $\pE_f\E_\cY f^2=\pE_f\E_\cX f^2$.
Hence, we can bound the fourth moment of $f$ over $\cX$ in terms of
the its second moment and $\tau$.

\begin{corollary}
\label{cor:invariance}
  \begin{displaymath}
    \pE_f \E_\cX f^4 = 2^{O(\ell)}\pE_f (\E_\cX f^2)^2 
    \pm k^{O(\ell)}\sqrt\tau \mper    
  \end{displaymath}
\end{corollary}

(The corollary shows that for small enough $\tau$, the $4$-norm and
$2$-norm of $f$ are within a factor of $2^{O(\ell)}$. 
This bound is useful because  the worst-case ratio of these norms is
$k^{O(\ell)}\gg 2^{O(\ell)}$.)

\begin{proof}[Proof of \pref{thm:invariance-fourth-moment}]
  We consider the following intermediate sequences of ensembles $\super \cZ
  r = (\cX_1,\ldots,\cX_r,\cY_{r+1},\ldots,\cY_R)$.
  Note that $\super \cZ 0 = \cY$ and $\super \cZ R = \cX$.
  For $r\in\N$, we write $f=E_r f + D_r f$, where $E_r f$ is the part of
  $f$ that does not dependent on coordinate $r$ and $D_r f=f-E_r f$.
  For all $r\in\N$, the following identities (between polynomials in $f$)
  hold
  \begin{align*}    
    \E_{\super \cZ r} f^4 - \E_{\super \cZ {r-1}} f^4 
    &=     \E_{\super \cZ r} (E_r f+D_r f)^4 - \E_{\super \cZ {r-1}} (E_r
    f+D_r f)^4 \\
    & = \E_{\super \cZ r} 4(E_r f)(D_r f)^3 + (D_r f)^4
    - \E_{\super \cZ {r-1}} 4(E_r f)(D_r f)^3 + (D_r f)^4
   \mper
  \end{align*}
  The last step uses that the first two moments of the ensembles $\cX_r$
  and $\cY_r$ match and that $E_r f$ does not dependent on coordinate $r$.
  
  Hence, 
  \begin{displaymath}
    \E_\cX f^4 - \E_\cY f^4 
     = \sum_r \E_{\super \cZ r} 4(E_r f)(D_r f)^3 + (D_r f)^4
    - \E_{\super \cZ {r-1}} 4(E_r f)(D_r f)^3 + (D_r f)^4
  \end{displaymath}
  
  It remains to bound the pseudo-expectation of the right-hand side.
  First, we consider the term $\sum_r \E_{\super \cZ r}(D_r f)^4$.
  The expression $\E_{\super \cZ r}(D_r f)^4$ is the fourth moment of a
  Fourier-polynomial with degree at most $t\cdot \ell$.
  (Here, we use that the ensembles in the sequence $\cY$ consist of
  characters of $\GF 2^t$, which are Fourier polynomials of degree at most $t$.)
  Furthermore, $\Inf_r f=\E_{\super \cZ r}(D_r f)^2$ is the second moment of the
  this Fourier-polynomial.
  \Dnote{explain this better. which Fourier polynomial?}
  Hence, by hypercontractivity of low-degree Fourier-polynomials,
  \Dnote{add reference to the right theorem} $\sum_r \E_{\super \cZ r}(D_r
  f)^4\sle \sum_r 2^{O(t\cdot \ell)}(\Inf_r f)^2$.
  Thus, the pseudo-expectation is at most $\pE_f \sum_r\E_{\super \cZ
    r}(D_r f)^4 \le 2^{O(t\cdot \ell)}\tau = k^{O(\ell)}\tau$.

  Next, we consider the term $\sum_r \E_{\super \cZ r}(E_r f)(D_r f)^3$.
  (The remaining two terms are analogous.)
  To bound its pseudo-expectation, we apply Cauchy-Schwarz, \Dnote{add
    reference to the right version}
  \begin{equation}
    \label{eq:third-order-error}
    {\pE_f \sum_r \E_{\super \cZ r}(E_r f)(D_r f)^3}
    \le \Paren{\pE_f \sum_r \E_{\super \cZ r}(E_r f)^2(D_r f)^2}^{1/2}\cdot 
    \Paren{\pE_f \sum_r \E_{\super \cZ r}(D_r f)^4}^{1/2}
  \end{equation}
  Using hypercontractivity of low-degree Fourier-polynomials, we can bound
  the second factor of \pref{eq:third-order-error} by $\pE_f \sum_r
  \E_{\super \cZ r}(D_r f)^4=k^{O(\ell)} \tau$.
  It remains to bound the first factor of \pref{eq:third-order-error}.  
  Again by hypercontractivity, $\E_{\super \cZ r}(E_r f)^2(D_r f)^2\sle
  k^{O(\ell)}\cdot \snorm{E_r f}\cdot \snorm{D_r f}\sle k^{O(\ell)}
  \snorm{f}\cdot \snorm{D_r f}$.
  By the total influence bound for low-degree polynomials, we have $\sum_r
  \snorm{D_r f}\sle \ell \snorm{f}$.
  Thus $\sum _r \E_{\super \cZ r}(E_r f)^2(D_r f)^2\sle k^{O(\ell)} \norm{f}^4$.
  Using the assumption $\pE_f \norm{f}^4\le 1$, we can bound the first
  factor of \pref{eq:third-order-error} by $k^{O(\ell)}$.

  We conclude as desired that 
  \begin{displaymath}
    \Abs{\pE_f \E_\cX f^4 - \E_\cY f^4}\le k^{O(\ell)}\sqrt\tau\mper
  \end{displaymath}
\end{proof}

\subsection{Interlude: Independent Rounding}
\label{sec:interl-indep-round}

In this section, we will show how to convert variables that satisfy
$f^2\sle f$ to variables $\bar f$ satisfying $\bar f^2=\bar f$.
The derived variables $\bar f$ will inherit several properties of the
original variables $f$ (in particular, multilinear expectations).
This construction corresponds to the standard independent rounding for
variables with values between $0$ and $1$.
The main challenge is that our random variables are fictitious.

\Dnote{here we show that independent rounding works within SoS}

Let $f$ be a \lrvar4 over $\R^n$.
%
%
Suppose $f_i^2\sle f_i$ (in terms of an unspecified jointly-distributed
\lrvar4).
\Dnote{explain what is meant here somewhere}
Note that for real numbers $x$, the condition $x^2\le x$ is equivalent to
$x\in[0,1]$.

\begin{lemma}
\label{lem:independent-rounding}
  Let $f$ be a \lrvar4 over $\R^n$ and let $i\in [n]$ such that $f_i^2\sle f_i$.
  Then, there exists an \lrvar4 $(f,\bar f_i)$ over $\R^{n+1}$ such that
  $\pE_{f,\bar f_i} (\bar f_i^2-\bar f_i)^2 =0$ and for every polynomial
  $P$ which is linear in $\bar f_i$ and has degree at most $4$,
  \begin{displaymath}
    \pE_{f,\bar f_i} P(f,\bar f_i) = \pE_{f} P(f,f_i)\mper
  \end{displaymath}
\end{lemma}

\begin{proof}
  We define the pseudo-expectation functional $\pE_{f,\bar f_i}$ as follows:
  For every polynomial $P$ in $(f,\bar f_i)$ of degree at most $4$, let
  $P'$ be the polynomial obtained by replacing $\bar f_i^2$ by $\bar f_i$
  until $P'$ is (at most) linear in $\bar f_i$.
  (In other words, we reduce $P$ modulo the relation $\bar f_i^2=\bar
  f_i$.)
  We define $\pE_{f,\bar f_i} P(f,\bar f_i) = \pE_{f} P'(f,f_i)$.
  With this definition, $\pE_{f,\bar f_i} (\bar f_i^2-\bar f_i)^2 =0$.
  The operator $\pE_{f,\bar f_i}$ is clearly linear (since $(P+Q)'=P'+Q'$).
  It remains to verify positivity.
  Let $P$ be a polynomial of degree at most $4$.
  We will show $\E_{f,\bar f_i} P^2(f,\bar f_i)\ge 0$.
  Without loss of generality $P$ is linear in $\bar f_i$.
  We express $P = Q + \bar f_i R$, where $Q$ and $R$ are polynomials in
  $f$.
  Then, $(P^2)' = Q^2 + 2 \bar f_i Q R + \bar f_i R^2 $.
  Using our assumption $f_i^2\sle f_i$, we get $(P^2)'(f,f_i) = Q^2 + 2 f_i
  Q R + f_i R^2 \sge Q^2 + 2 f_i Q R + f_i^2 R^2= P^2(f,f_i)$.
  It follows as desired that
  \begin{displaymath}
    \pE_{f,\bar f_i} P^2 = \pE_{f} (P^2)'(f,f_i) \ge \pE_f P^2(f,f_i)\ge 0\mper
  \end{displaymath}
\end{proof}

\begin{corollary}
\label{cor:independent-rounding}
  Let $f$ be a \lrvar4 over $\R^n$ and let $I\sse [n]$ such that $f_i^2\sle
  f_i$ for all $i\in I$.
  Then, there exists an \lrvar4 $(f,\bar f_I)$ over $\R^{n+\card I}$ such
  that $\pE_{f,\bar f_I} (\bar f_i^2-\bar f_i)^2 =0$ for all $i\in I$ and
  for every polynomial $P$ which is multilinear in the variables $\set{\bar
    f_i}_{i\in I}$ and has degree at most $4$,
  \begin{displaymath}
    \pE_{f,\bar f_I} P(f,\bar f_I) = \pE_{f} P(f,f_I)\mper
  \end{displaymath}
\end{corollary}

\subsection{Dictatorship Test for Small-Set Expansion}
\label{sec:dict-test-sse}

\Dnote{cite KKMO here somewhere. point out the difference to the KKMO analysis}

Let $\Omega=\set{0,\ldots,k-1}$ and let $T_{1-\e}$ be the noise graph on
$\Omega^R$ with second largest eigenvalue $1-\e$.
Let $f$ be a \lrvar4 over $L_2(\Omega^R)$.
Suppose $f^2\sle f$ (in terms of an unspecified jointly-distributed
\lrvar4).
\Dnote{explain what is meant here somewhere}
Note that for real numbers $x$, the condition $x^2\le x$ is equivalent to
$x\in[0,1]$.

The following theorem is an analog of the ``Majority is Stablest'' result
\cite{MosselOO05}.
\begin{theorem}
  \label{thm:sse-dictatorship-test}
  Suppose 
  $\pE_f (\E f)^2\le \delta^2$.
  Let $\tau=\pE_f \sum_r (\super \Inf{\le \ell}_rf)^2$ for
  $\ell=\Omega(\log(1/\delta))$.
  Then,
  \begin{displaymath}
    \pE_f \iprod{f,T_{1-\e} f} \le \delta^{1+\Omega(\e)} +
    k^{O(\log(1/\delta))} \cdot \tau^{1/8}\mper
  \end{displaymath}
  (Here, we assume that $\e$, $\delta$ and $\tau$ are sufficiently small.)
\end{theorem}

The previous theorem is about graph expansion (measured by the quadratic
form $\iprod{f,T_{1-\e} f}$).
The following lemma allows us to relate graph expansion to the $4$-norm of
the projection of $f$ into the span of the eigenfunctions of $T_{1-\e}$
with significant eigenvalue.
We will be able to bound this $4$-norm in terms of the influences of $f$
(using the invariance principle in the previous section).

\begin{lemma}
  \label{lem:expansion-bound}
  Let $f$ be a \lrvar4 over $L_2(\Omega^R)$.
  Suppose $f^2\sle f$ (in terms of unspecified jointly-distributed
  \lrvar4s).
  Then for all $\lambda>0$,
  \begin{displaymath}
    \pE_f \iprod{f,T_{1-\e} f}
    \le (\pE_f \E f)^{3/4} (\pE_f \E (P_{>\lambda}f)^4)^{1/4}
    + \lambda \pE_f \E f
    \mper
  \end{displaymath}
  Here, $P_{>\lambda}$ is the projector into the span of the eigenfunctions of
  $T_{1-\e}$ with eigenvalue larger than $\lambda$.
\end{lemma}

\begin{proof}
  The following relation between polynomials holds
  \begin{displaymath}
    \iprod{f,T_{1-\e} f} \sle \E f\cdot (P_{>\lambda} f) + \lambda \E f^2\mper
  \end{displaymath}
  By \pref{cor:independent-rounding}, there exists a \lrvar4 $(f,\bar f)$
  over $L_2(\Omega^R)\times L_2(\Omega^R)$ such that $\bar f^2 \equiv_{\bar f}
  \bar f$.
  Then,
  \begin{align*}
    \pE_{f} \E f\cdot (P_{>\lambda} f)
    & = \pE_{f,\,\bar f} ~\E \bar f\cdot (P_{>\lambda} f)
    \quad\using{linearity in $\bar f$} \\
    & = \pE_{f,\,\bar f} ~\E \bar f ^3 \cdot (P_{>\lambda} f)
    \quad\using{$\bar f^2 \equiv_{\bar f} \bar f$}\\
    & \le \Paren{\pE\nolimits_{\bar f} ~\E \bar f ^4}^{3/4}
    \cdot \Paren{\pE\nolimits_{f}\E (P_{>\lambda} f)^4}^{1/4}
    \quad\using{\pref{lem:holders} (\Holder)} \\
    & =   \Paren{\pE\nolimits_{\bar f} ~\E \bar f }^{3/4}
    \cdot \Paren{\pE\nolimits_{f}\E (P_{>\lambda} f)^4}^{1/4}
    \quad\using{$\bar f^2 \equiv_{\bar f} \bar f$}\\
    & =   \Paren{\pE\nolimits_{f} ~\E  f }^{3/4}
    \cdot \Paren{\pE\nolimits_{f}\E (P_{>\lambda} f)^4}^{1/4}
    \quad\using{linearity in $\bar f$}
    \qedhere
  \end{align*}
\end{proof}

\begin{proof}[Proof of \pref{thm:sse-dictatorship-test}]
  By \pref{lem:expansion-bound},
  \begin{displaymath}
    \pE_f \iprod{f,T_{1-\e} f}
    \le (\pE_f \E f)^{3/4} (\pE_f \E (P_{>\lambda}f)^4)^{1/4}
    + \lambda \pE_f \E f^2
    \mper
  \end{displaymath}
  Using \pref{cor:invariance},
  \begin{displaymath}
    \pE_f \iprod{f,T_{1-\e} f}
    \le 2^{O(\ell)}\cdot (\pE_f \E f)^{3/4} (\pE_f (\E f^2)^2 +\sqrt{\tau}\cdot k^{O(\ell)} )^{1/4}
    + \lambda \pE_f \E f^2
    \mper
  \end{displaymath}
  Here, $\ell=\log(1/\lambda)/\e$.
  Using the relation $f^2\sle f$ and our assumption $\pE_f (\E f)^2\le
  \delta^2$, we get $\pE_f \E f^2 \le \pE_f \E f \le (\pE_f (\E f)^2)^{1/2}
  \le \delta$ (by Cauchy--Schwarz).
  Hence,
  \begin{align*}
    \pE_f \iprod{f,T_{1-\e} f}& \le (1/\lambda)^{O(1/\e)} \delta^{3/4}(\delta^2
    + \sqrt\tau \cdot (1/\lambda)^{O(\log k)/\e})^{1/4} + \lambda \delta\\
    & \le (1/\lambda)^{O(1/\e)} \delta^{5/4} + (1/\lambda)^{O(\log k)/\e}
    \delta^{3/4} \tau^{1/8} + \lambda \cdot \delta \mper
  \end{align*}
  To balance the terms $(1/\lambda)^{O(1/\e)}\delta^{5/4}$ and $\lambda\delta$,
  we choose $\lambda=\delta^{\Omega(\e)}$.
  \Dnote{check. making some mild assumptions about the range of $\e$ and
    $\tau$ here.}
  We conclude  the desired bound,
  \begin{displaymath}
    \pE_f \iprod{f,T_{1-\e} f} \le \delta^{1+\Omega(\e)} +
    k^{O(\log(1/\delta))} \cdot \tau^{1/8}\mper
    \qedhere
  \end{displaymath}
\end{proof}

\subsection{Dictatorship Test for Unique Games}
\label{sec:ug-dictatorship-test}

Let $\Omega=\Z_k$ (cyclic group of order $k$) and let $f$ be a \lrvar4 over
$L_2(\Omega\times\Omega^R)$.
Here, $f(a,x)$ is intended to be $0/1$ variable indicating whether symbol
$a$ is assigned to the point $x$.
%

\Dnote{say that this graph corresponds to the usual 2-query test for unique games}

The following graph $T'_{1-\e}$ on $\Omega\times \Omega^R$ corresponds to
the $2$-query dictatorship test for \uniquegames \cite{KhotKMO04},
\begin{displaymath}
  T_{1-\e}' f(a,x) = \E_{c\in \Omega} \E_{y\sim_{1-\e} x} f(a+c,y-c\cdot \Ind)\mper
\end{displaymath}
Here, $y\sim_{1-\e} x$ means that $y$ is a random neighbor of $x$ in the
graph $T_{1-\e}$ (the $\e$-noise graph on $\Omega^R$).

We define $\bar f(x) \seteq \E_{c\in \Omega} f(c,x-c\cdot \Ind)$.
(We think of $\bar f$ as a variable over $L_2(\Omega^R)$.)
Then, the following polynomial identity (in $f$) holds \Dnote{check this
  identity}
\begin{displaymath}
  \iprod{f,T'_{1-\e} f} =  \iprod{\bar f,T_{1-\e} \bar f}.
\end{displaymath}

\begin{theorem}
  \label{thm:ug-dictatorship-test}
  Suppose $f^2\sle f$ and $\pE_f (\E f)^2\le \delta^2$.
  Let $\tau=\pE_f \sum_{r} (\super \Inf{\le \ell}_r\bar f)^2$ for
  $\ell=\Omega(\log(1/\delta))$.
  Then,
  \begin{displaymath}
    \pE_f \iprod{f,T'_{1-\e} f} \le \delta^{1+\Omega(\e)} +
    k^{O(\log(1/\delta))} \cdot \tau^{1/8}\mper
  \end{displaymath}
  (Here, we assume that $\e$, $\delta$ and $\tau$ are sufficiently small.)
\end{theorem}

\begin{proof}

  Apply \pref{thm:sse-dictatorship-test} to bound $\pE_f
  \iprod{\bar f,T_{1-\e} \bar f}$.
  Use that fact that $\E f= \E \bar f$ (as polynomials in $f$).
\end{proof}

\subsection{Influence Decoding}
\label{sec:influence-decoding}

Let $\cU$ be a unique game with vertex set $V$ and alphabet $[R]$.
Recall that we represent $\cU$ as a distribution over triples $(u,v,\pi)$
where $u,v\in V$ and $\pi$ is a permutation of $[R]$.
The triples encode the constraints of $\cU$.
We assume that the unique game $\cU$ is regular in the same that every
vertex participates in the same fraction of constraints.

Let $\Omega=\Z_k$ (cyclic group of order $k$).
We reduce $\cU$ to a unique game $\cW=\cW_{\e,k}(\cU)$ with vertex set
$V\times \Omega^R$ and alphabet~$\Omega$.
Let $f=\set{f_u}_{u\in V}$ be a variable over $L_2(\Omega\times \Omega^R)^V$.
The unique game $\cW$ corresponds to the following quadratic form in $f$,
\begin{displaymath}
  \iprod{f,\cW f}
  \seteq \E_{u\in V}
  \E_{\substack{(u,v,\pi)\sim \cU\mid u \\ (u,v',\pi')\sim\cU\mid u}}
  \iprod{\super f \pi _v, T'_{1-\e} \super f {\pi'}_{v'}}\mper
\end{displaymath}
Here, $(u,v,\pi)\sim \cU\mid u$ denotes a random constraint of $\cU$
incident to vertex $u$, the graph $T'_{1-\e}$ corresponds to the
dictatorship test of \uniquegames defined in
\pref{sec:ug-dictatorship-test}, and $\super f {\pi}_v(a,x)=f_v(a,\pi.x)$
is the function obtained by permuting the last $R$ coordinates according to
$\pi$ (where $\pi.x(i) = x_{\pi(i)}$).
\Dnote{say more precisely what is meant by permuting the coordinates
  according to $\pi$.}

We define $g_u = \E_{(u,v,\pi)\sim \cU\mid u} \super f \pi_v$.
Then,
\begin{equation}
  \label{eq:ug-form-symmetric}
  \iprod{f,\cW f} = \E_{u\in V} \iprod{g_u,T'_{1-\e} g_u}\mper
\end{equation}

\paragraph{Bounding the value of  SoS solutions}

Let $f=\set{f_u}_{u\in V}$ be a solution to the level-$d$ SoS
relaxation for the unique game $\cW$.
In particular, $f$ is a $d$-f.r.v. over $L_2(\Omega\times
\Omega^R)^V$.
Furthermore, $\pE_f (\E f_u)^2 \le 1/k^2$ for all vertices $u\in V$.

By applying \pref{thm:ug-dictatorship-test} to \pref{eq:ug-form-symmetric},
we can bound the objective value of $f$
\begin{align*}
  \pE_f\iprod{f,\cW f} &\le 1/k^{1+\Omega(\e)}
  + k^{O(\log k)}\Paren{\pE_f \E_{u\in V} \tau_u}^{1/8}\mcom
\end{align*}
where $\tau_u = \sum_{r} (\super \Inf {\le \ell}_r \bar g_{u})^2$,
$\bar g_{u}(x) = \E_{(u,v,\pi)\sim \cU\mid u} \super{\bar f} \pi _v$, and
$\bar f_v (x) = \E_{c\in \Omega} f_v(c,x-c \cdot \Ind)$.
\Dnote{justify that conditions of \pref{thm:ug-dictatorship-test} were
  satisfied}

Since $\super \Inf {\le \ell }_r$ is a positive semidefinite form,
\Dnote{check $\pi$ vs $\inv \pi$ issue}
\begin{displaymath}
  \tau_u
  \sle \sum_r \Paren{\E_{(u,v,\pi)\sim \cU\mid u}
    \super \Inf {\le \ell}_r \super {\bar f_{v}} \pi}^2
  = \sum_r \Paren{\E_{(u,v,\pi)\sim \cU\mid u}
    \super \Inf {\le  \ell}_{\pi(r)} {\bar f_{v}} }^2.
\end{displaymath}
Let $h$ be the level-$d/2$ fictitious random variable over $\L2(V\times [R])$ with
$h(u,r)=\super \Inf {\le \ell} _r \bar f_u$.
Let $G_\cU$ be the label-extended graph of the unique game $\cU$.
\Dnote{should define label-extended graph}
Then, the previous bound on $\tau_u$ shows that
\begin{math}
  \E_{u\in V}\tau_u \sle R \cdot \snorm{G_\cU h}\mper
\end{math}
\pref{lem:boundedness-relation} shows that $h^2\sle h$.
\Dnote{maybe explain a bit more}
On the other hand, $\sum_r h(u,r) \sle \ell\snorm{\bar f_u}\sle \ell
\snorm{f_u}$ (bound on the total influence of low-degree Fourier
polynomials).
In particular, $\E h \sle \ell \E_{u\in V} \snorm{f_u}/R$.
Since $f$ is a valid SoS solution for the unique game $\cW$, we have
$\pE_f \norm{f_u}^d\le 1/k^{d/2}$ for all $u\in V$.
(Here, we assume that $d$ is even.)
It follows that $\pE_h (\E h)^{d/2}\le (\tfrac{\ell}{k\cdot R})^{d/2}$.

The arguments in this subsection imply the following theorem.

\begin{theorem}
  \label{thm:ug-value-bound}
  The optimal value of the level-$d$ SoS relaxation for the unique
  game $\cW=\cW_{\e,k}(\cU)$ is bounded from above by
  \begin{displaymath}
    1/k^{\Omega(\e)} + k^{O(\log k)}
    \Paren{R\cdot \max_h \pE_h \snorm{G_\cU h}}^{1/8}\mcom
  \end{displaymath}
  where the maximum is over all level-$d/2$ fictitious random variables $h$ over
  $L_2(V\times [R])$ satisfying $h^2\sle h$ and $\pE_h (\E h)^{d/2}\le
  \ell/R^{d/2}$.
\end{theorem}

\Dnote{remark not so important}

\begin{remark}
  Since the quadratic form $\snorm{G_\cU h}$ has only nonnegative
  coefficients (in the standard basis), we can use
  \pref{cor:independent-rounding} to ensure that the level-$d/2$ random
  variable $h$ satisfies in addition $h^2\equiv_h h$.
\end{remark}

\Dnote{maybe say that one can extract more information about the variable
  $h$. For example, it cannot put too much mass on a single cloud. in this
  sense, $h$ is very close to a SoS solution for the unique game $\cU$.}

\subsection{Certifying Small-Set Expansion}
\label{sec:small-set-expansion}

Let $T_{1-\e}$ be a the noise graph on $\sbits^R$ with second largest
eigenvalue $1-\e$.
\begin{theorem}
  \label{thm:certify-sse}
  Let $f$ be level-$4$ fictitious random variables over $L_2(\sbits^R)$.
  Suppose that $f^2\sle f$ (in terms of unspecified jointly-distributed
  level-$4$ fictitious random variables) and that $\pE_f (\E f)^2\le \delta^2$.
  Then,
  \begin{displaymath}
    \pE_f\iprod{f,T_{1-\e} f} \le \delta^{1+\Omega(\e)}\mper
  \end{displaymath}
\end{theorem}

\begin{proof}
  By \pref{lem:expansion-bound} (applying it for the case
  $\Omega=\set{0,1}$), for every $\lambda>0$,
  \begin{displaymath}
    \pE_f\iprod{f,T_{1-\e} f}
    \le (\pE_f \E f)^{3/4} (\pE_f \E (P_{>\lambda}f)^4)^{1/4}
    + \lambda \pE_f \E f\mper
  \end{displaymath}
  For the graph $T_{1-\e}$, the eigenfunctions with eigenvalue larger than
  $\lambda$ are characters with degree at most $\log(1/\lambda)/\e$.
  Hence, \pref{lem:hypercontractivity1} implies $\E(P_{>\lambda} f)^4\sle
  (1/\lambda)^{O(1/\e)} \norm{f}^4$.
  Since $f^2\sle f$, we have $\norm{f}^4\sle (\E f)^2$.
  Hence, $\pE_f \E(P_{>\lambda} f)^4\le (1/\lambda)^{O(1/\e)}\delta^2$.
  Plugging in, we get
  \begin{displaymath}
    \pE_f\iprod{f,T_{1-\e} f}\le (1/\lambda)^{O(1/\e)} \delta^{5/4} +
    \lambda\cdot \delta
    \mper
  \end{displaymath}
  To balance the terms, we choose $\lambda=\delta^{\Omega(\e)}$, which gives the
  desired bound.
\end{proof}

\subsection{Putting Things Together}
\label{sec:putt-things-togeth}

Let $T_{1-\eta}$ be a the noise graph on $\sbits^R$ with second largest
eigenvalue $1-\eta$.
Let $\cU=\cU_{\eta,R}$ be an instance of \uniquegames with label-extended graph
$G_\cU=T_{1-\eta}$ (e.g., the construction in \cite{KhotV05}).

Combining \pref{thm:ug-value-bound} (with $d = 4$) and \pref{thm:certify-sse} gives the
following result.

\begin{theorem}
\label{thm:ug-main}
  The optimal value of the level-$8$ SoS relaxation for the unique
  game $\cW=\cW_{\e,k}(\cU_{\eta,R})$ is bounded from above by
  \begin{displaymath}
    1/k^{\Omega(\e)} + k^{O(\log k)} \cdot R^{-\Omega(\eta)}\mper
  \end{displaymath}
  In particular, the optimal value of the relaxation is close to
  $1/k^{\Omega(\e)}$ if $\log R \gg (\log k)^2 /\eta$.
\end{theorem}

\subsection{Refuting Instances based on Short Code} \label{sec:short-code}

\Dnote{TODO: describe the kind of things that change (and don't change) if
  we apply everything to the short code constructions as opposed to the
  long code constructions.}

Let $\cU'=\cU'_{\eta,R}$ be an instance of \uniquegames according to the
basic construction in \cite{BarakGHMRS11}.
(The label-extended graph of $\cU$ will be a subgraph of $T_{1-\e}$ induced
by the subset of $\sbits^{R}$ corresponding to a Reed--Muller code, that
is, evaluations of low-degree $\GF2$-polynomials.)

Let $\cW'=\cW'_{\e,k}(\cU'_{\eta,R})$ be the unique game obtained by
applying the short-code alphabet reduction of \cite{BarakGHMRS11}.

The following analog of \pref{thm:ug-main} holds.

\begin{theorem}
  \label{thm:ug-main-short}
  The optimal value of the level-$8$ SoS relaxation for the unique
  game $\cW'=\cW'_{\e,k}(\cU'_{\eta,R})$ is bounded from above by
  \begin{displaymath}
    1/k^{\Omega(\e)} + k^{O(\log k)} \cdot R^{-\Omega(\eta)}\mper
  \end{displaymath}
  In particular, the optimal value of the relaxation is close to
  $1/k^{\Omega(\e)}$ if $\log R \gg (\log k)^2 /\eta$.
\end{theorem}

The proof of \pref{thm:ug-main-short} is almost literally the same as the
proof of \pref{thm:ug-main}.
In the following, we sketch the main arguments why the proof doesn't have
to change.
First, several of the results of the previous sections apply to general
graphs and instances of \uniquegames.
In particular, \pref{lem:expansion-bound} applies to general graphs and
\pref{thm:ug-value-bound} applies to general gadget-composed instances of
unique games assuming a ``Majority is Stablest'' result for the gadget.
In fact, the only parts that require further justification are the
invariance principle (\pref{thm:invariance-fourth-moment}) and hypercontractivity bound
(\pref{lem:hypercontractivity1}).
Both the invariance principle and the hypercontractivity bound are about
the fourth moment of a low-degree Fourier polynomial (whose coefficients
are fictitious random variables).
For the construction of \cite{BarakGHMRS11}, we need to argue about the
fourth moment with respect to a different distribution over inputs.
(Instead of the uniform distribution, \cite{BarakGHMRS11} considers a
distribution over inputs related to the Reed--Muller code.)
However, this distribution happens to be $k$-wise independent for $k/4$
larger than the degree of our Fourier polynomial.
Hence, as a degree-$4$ polynomial in Fourier coefficients, the fourth
moment with respect to the \cite{BarakGHMRS11}-input distribution is the
same as with respect to the uniform distribution, which considered here.


\section{Hypercontractivity of random operators} \label{sec:random}

We already saw that the \tensorsdp algorithm provides non-trivial guarantees on the $2\to 4$ norms of the projector to low-degree polynomials. In this section we show that it also works for a natural but very different class of instances, namely random linear operators.
Let 
\be A=\sum_{i=1}^m \sum_{j=1}^n \frac{a_{i,j}}{\sqrt{n}} e_i e_j^T,
\label{eq:RM-def}\ee
 where $e_i$ is the vector with a 1
in the $i^{\text{th}}$ position, and each $a_{i,j}$ is chosen i.i.d. from
a distribution $\cD$ on $\R$.  
We will show that \tensorsdp returns an
answer close to the correct value under fairly general assumptions on
$\cD$.   Specifically, we will assume that 
\begsub{D-ass}
 \E[a_{i,j}]&=0 \label{eq:Ea}\\
\E[a_{i,j}^2]&=1 \label{eq:Ea2}\\
\E[a_{i,j}^4] & =:\mu_4 \label{eq:Ea4}\\
\E[\exp((a_{i,j}/\psi)^2)]&\leq 2, \label{eq:subgaussian}
\endsub
for some constant $\psi>0$.  Two examples of distributions that meet
these criteria are the uniform distribution over $\{-1,1\}$ and the
standard (mean-zero, unit variance) Gaussian distribution.

The following theorem is the main result of this section, and shows
that the approximation ratio of \tensorsdp approaches 1 as $m,n\ra
\infty$ and $n^2/m \ra 0$.
\begin{theorem}\label{thm:random}
There exist constants $0< c_1 < c_2$ such that
\begin{align} \max(\mu_4,3 + c_1 \frac{n^2}{m})(1- o(1))
&\leq \|A\|_{2\ra 4}^4
\leq \tensorsdp(A)\\
&\leq \max(\mu_4,3) + c_2 \max\Paren{\frac{n}{\sqrt{m}},\frac{n^2}{m}}
\end{align}
with high probability (i.e. probability $1-o(1)$) over random matrices $A$ distributed according
to \eq{RM-def} with $\cD$ satisfying \eq{D-ass}.   Here $o(1)$ refers
to quantities that approach zero when both $m$ and $n$ approach $\infty$.
\end{theorem}

To get some intuition for these terms, note that $\E[Z^4]=3$ if $Z$ is
a standard Gaussian random variable.  Similarly for any fixed vector
$x\in \R^n$, if $A$ has standard Gaussian entries, then $Ax$ will have
Gaussian entries with mean zero and variance $\norm{x}^2$.  Even if
$A$ has general entries (with variance 1) this will be approximately
true because of the central limit theorem.  This accounts for the 3
term that dominates when $m\gg n^2$. On the other hand, the
lower bound of $n^2/m$ holds because we can choose $x$ {\em after} $A$
is chosen.  We will see that it emerges from choosing the signs of $x$
to match those of any row of $A$.  These lower bounds hold under
rather general assumptions and variants of them apply even without any
randomness, as we will briefly explore in \lemref{24-LB}.

The upper bound on the value of \tensorsdp is rather more
complicated.   It will be seen to follow from a
concentration-of-measure bound for matrices.  However, the argument
does not simply involve bounding the top eigenvalue of the appropriate
random matrix.  We will need to further make use of the symmetry
properties of \tensorsdp; indeed it appears crucial that \tensorsdp
uses semidefinite programming rather than simply an eigenvalue
calculation.  See Remark~\ref{remark:complex} for more on this point
and for a comparison with the case of  complex matrices.

\subsection{Orlicz norms and background results}
Before proving \thmref{random}, we introduce some results from the
literature.  First, we discuss the implications of \eq{subgaussian}.
\begin{lemma}[Lemma 5.5 and Remark 5.6 of \cite{Vershynin12}]\label{lem:subgaussian}
Let $Z$ be a real-valued random variable with $\E[Z]=0$.  The
following are equivalent with parameters $C_i>0$ differing by at most
constant factors:
\benum
\item Sub-gaussian moment: $\E[\exp((Z/C_1)^2)] \leq 2$.
\item Moments: $\E[|Z|^k] \leq (C_2 \sqrt{k})^k$ for all nonnegative
  integers $k$.
\item Tails: $\Pr[|Z|\geq t] \leq \exp(1 - t^2/C_3^3)$ for all $t\geq
  0$.
\item Sub-exponential moments: $\E[\exp(tZ)] \leq \exp(t^2 C_4^2)$ for
  all $t\in\R$.
\eenum
\end{lemma}
The largest $\psi$ for which \eq{subgaussian} holds is called the
$\psi_2$ norm of a distribution, where the $2$ refers to the fact that
we have $Z^2$ in the exponent.  In general, the $\psi_p$ norm of a
distribution $\cD$ refers
to the smallest $\psi>0$ such that $\E_{Z\sim
  \cD}[\exp(|Z/\psi|^p)]\leq 2$. 

In what follows, we will also need to define the $\psi_1$ norm of a
vector.  If $\cD$ is now a distribution on $\R^N$, define the $\psi_p$
norm $\norm{\cD}_{\psi_p}$ to be the smallest $\psi>0$ such that
 \be
\max_{x\in S(\R^N)} \E_{v\sim \cD} \exp\left(\psi^{-p}|\iprod{x,v}|^p
  N^{p/2}\right) \leq 2 ,
\label{eq:psi-p-norm}\ee
or $\infty$ if no finite such $\psi$ exists.  
We depart from the normal convention by including a factor of
$N^{p/2}$ in the definition, so that our expectation-norm convention
will be consistent with the results in \cite{AdamczakLPT11},
and specifically \lemref{ALPT} below.  This definition is also
consistent with the term ``subgaussian moment'' since a vector of
i.i.d.~Gaussians with mean zero and variance one will have $\psi_2$
equal to a constant (in fact $\sqrt{8/3}$).
We will abuse notation 
and write $\|X\|_{\psi_p}$ to refer to the $\psi_p$ norm of the
distribution $\cD$ associated with a random variable $X$.

{\em Relations between norms.} One can verify that if a vector has
i.i.d. entries with $O(1)$ $\psi_2$ norm, then the distribution over
vectors also has $\|\cD\|_{\psi_2}\leq O(1)$. However, the relation
between $\psi_p$ norms for different values of $p$ is less clear in
the vector case.   If $Z$ is a real-valued random variable, then
$\|X\|^p_{\psi_p}=\|X^p\|_{\psi_1}$, but there is no simple analogue
of this for vectors.  We will see an example below where $\|v\ot
v\|_{\psi_1}$ can be larger than $\|v\|_{\psi_2}^2$ by a
dimension-dependent factor; also we will see why bounding the
$\psi_1$ norm is important.

We will require a bound from \cite{AdamczakLPT09,AdamczakLPT11} about the
convergence of sums of i.i.d rank-one matrices.
\begin{lemma}[\cite{AdamczakLPT11}]\label{lem:ALPT}
Let $b_1,\ldots,b_m$ be independent random vectors in $\R^N$ with
$\|b_i\|_{\psi_1} \leq \varphi$ and satisfying
\be \Pr_{b_1,\ldots,b_m}\left[\max_{i\in [m]}\normt{b_i} > 
K \max(1,(m/N)^{1/4})\right] \leq e^{-\sqrt{N}}. 
\label{eq:boundedness}\ee 
Then with probability $\geq 1 - 2\exp(-c\sqrt{N})$, we
have
\be \left\| \frac{1}{m} \sum_{i=1}^m (b_ib_i^T - \E[b_ib_i^T])\right\| \leq \eps
\label{eq:concentration}\ee
where $\eps= C(\varphi+K)^2\max(N/m,\sqrt{N/m})$ with $c,C>0$ universal constants.
\end{lemma}
The $N\leq m$ case (when the $\sqrt{N/m}$ term is applicable) was
proven in Theorem 1 of \cite{AdamczakLPT11}, and the $N>m$ case
(i.e.~when the max is achieved by $N/m$) is discussed in Remark 1.2 of
\cite{AdamczakLPT11} (see also Theorem 3.13 of \cite{AdamczakLPT09}).

We will also make use of the Hanson-Wright inequality, in the form
proved by Rudelson and Vershynin in \cite{RudelsonV13}.
\begin{lemma}[Theorem 1.1 of \cite{RudelsonV13}]\label{lem:H-W}
Let $a\in\R^N$ be a random vector with independent components
satisfying $\E[a_i]=0$ and $\|a_i\|_{\psi_2} \leq \psi$.  Let $Y$ be an
$n\times n$ matrix.  For $t\geq 0$,
\be \Pr\left[ |a^TYa-\E[a^TYa]| > t\right]
\leq 2 \exp\left[ -c\min\left(\frac{t^2}{\psi^4\|Y\|_F^2},
    \frac{t}{\psi^2\|Y\|_{2\ra 2}}\right)\right].\label{eq:H-W}\ee
Here $\|Y\|_F$ is the Frobenius norm, defined as $\|Y\|_F :=\sqrt{\tr Y^TY}$.
\end{lemma}

\subsection{Proof of \thmref{random}}
\begin{proof}
First we prove that, for some constant $c_1>0$,
\begsub{A24-lb}
(3 + c_1 \frac{n^2}{m})(1- o(1))&\leq \|A\|_{2\ra 4}^4
\label{eq:A24-lb-1}\\
\mu_4(1- o(1))&\leq \|A\|_{2\ra 4}^4.
\label{eq:A24-lb-2}
\endsub
with high probability (i.e. probability $1-o(1)$).
 Define a vector $x\in\R^n$ by $x_j =
\sign(a_{1,j})$ (or arbitrarily if $a_{1,j}=0$).  Then $\|x\|_2=1$ and 
\ban 
 \|Ax\|_4^4 &= \frac{1}{n^2}\E_{i\in [m]} \left(\sum_{j\in [n]} a_{i,j} x_j\right)^4
= \frac{1}{mn^2} \left(\sum_{j\in [n]} |a_{1,j}|\right)^4 
+ \frac{m-1}{mn^2} \E_{i\neq 1}\left(\sum_{j\in [n]} a_{i,j}x_j\right)^4 \\
& =
\frac{1}{mn^2} 
\!\!\!\sum_{\substack{j_1,j_2,j_3,j_4\\\in [n]}}\!\!\!
|a_{1,j_1} a_{1,j_2} a_{1,j_3} a_{1,j_4}|
+ \frac{m-1}{mn^2} \E_{i\neq 1}
\!\!\!\sum_{\substack{j_1,j_2,j_3,j_4\\\in [n]}}\!\!\!
a_{i,j_1} a_{i,j_2} a_{i,j_3} a_{i,j_4} 
x_{j_1}x_{j_2}x_{j_3}x_{j_4}
\ean 
Next, we will average over the choice of $A$.  The first term is
proportional to 
$\E[|a_{1,j_1} a_{1,j_2} a_{1,j_3} a_{1,j_4}|] = \E[|a_{1,1}|]^4$, if
we drop the $O(1/n)$ fraction of terms where the $j_1,j_2,j_3,j_4$ are
not all distinct. 

For the second
term, recall that $\E[a_{i,j}]=0$.  Since  $x_j$ and $a_{i,j}$ are
independent for $i\neq 1$ we also have $\E[a_{i,j}x_j]=0$.
  Thus the only terms that survive have $j_1,j_2,j_3,j_4$ paired off in
  one of three ways: either $j_1=j_2,j_3=j_4$ or $j_1=j_3, j_2=j_4$ or
  $j_1=j_4,j_2=j_3$.  Each of these contributes an identical
  $\E[a_{1,1}^2]^2$.  We can neglect the overcounting from the
  $j_1=j_2=j_3=j_4$ terms because it accounts for an $O(1/n)$ fraction of
  the terms and by \lemref{subgaussian}, $\E[a_{1,1}^4]=\mu_4 \leq O(1)$.  We
  conclude that 
\be
\E[ \|Ax\|_4^4] =
\frac{n^2}{m}(1\pm o(1)) \E[|a_{1,1}|]^4
+ 3(1\pm o(1))\E[a_{1,1}^2]^2
\label{eq:EAx44}\ee
By assumption $\E[a_{1,1}^2]=1$.  From \lemref{subgaussian} we have
that $\E[a_{1,1}^4] \leq O(1)$.  Recall the Berger-H\"older
inequality $\E[|X|] \geq \frac{\E[X^2]^{3/2}}{\E[X^4]^{1/2}}$ from
\cite{Berger}.  Combining these facts we have that $\E[|a_{1,1}|]
\geq \mu_4^{-1/2} =: c_1^{1/4}$.  This proves that $\E[\|Ax\|_4^4]
\geq (3+c_1n^2/m)(1-o(1))$.  

To show that this inequality holds with
high probability we will argue that $\Var(\|Ax\|_4^4) \leq
o(\E[\|Ax\|_4^4]^2)$.   We calculate
\be \E[\|Ax\|_4^8] = 
\frac{1}{n^4} \E_{i,i'\in [m]}
\!\!\!\!\!\!\sum_{\substack{j_1,j_2,j_3,j_4\in [n]\\
j_1',j_2',j_3',j_4'\in [n]}}\!\!\!\!\!\!
a_{i,j_1}a_{i,j_2}a_{i,j_3}a_{i,j_4}
a_{i',j_1'}a_{i',j_2'}a_{i',j_3'}a_{i',j_4'} 
x_{j_1}x_{j_2}x_{j_3}x_{j_4}x_{j_1'}x_{j_2'}x_{j_3'}x_{j_4'}
\ee
The dominant terms here all correspond (up to $1-o(1)$ factors) to
terms in the expansion of $\E[\|Ax\|_4^4]^2$.  These terms correspond
to $i=1$, $j_1,j_2,j_3,j_4$ all distinct or $i\neq 1$,
$j_1,j_2,j_3,j_4$ comprising two elements each repeated twice; and
similarly one of those two possibilities for
$i',j_1',j_2',j_3',j_4'$.  In the case when $i,i' \neq 1$, the
dominant contribution comes from $i\neq i'$ for which the only
pairings we count involve matchings within the $\{j_1,j_2,j_3,j_4\}$
and $\{j_1',j_2',j_3',j_4'\}$, but not between these two sets.
The terms we neglect in this way are smaller by a 
$O(n^{-1}+m^{-1})$ factor.  We conclude that \eq{A24-lb-1} holds with high probability.

In case $\mu_4$ is large, this bound may not be optimal.  In that
case, we choose $x=\sqrt n e_j$ for some $j\in [n]$.  Then $Ax=\sum_i
a_{i,j}e_i$ and $\|Ax\|_4^4 = \frac{1}{m} \sum_{i\in [m]} a_{i,j}^4$.
Thus $\E[\|Ax\|_4^4] = \mu_4$.  For any fixed $j$, we have
$\E[\|Ax\|_4^8] = \mu_4^2 + m^{-1}(\E[a_{1,1}^8]-\mu_4^2)$, implying
that $\Var[\|Ax\|_4^4] = O(1/m)$.  This implies that
 \eq{A24-lb-2} holds with high probability, and concludes the proof of
 the lower bound on $\|A\|_{2\ra 4}^4$.

The more interesting half of the proof is to show that
\be  \tensorsdp(A)
\leq \max(\mu_4,3) + c_2 \max\Paren{\frac{n}{\sqrt{m}},\frac{n^2}{m}}
\label{eq:tsdp-ub}\ee
Let $a_i := \sum_{j\in [n]} a_{i,j} e_j$ so that $A = \sum_{i=1}^m e_ia_i^T /
\sqrt{n}$.  Define $A_{2,2} = \frac{1}{m}\sum_{i=1}^m a_i a_i^T \ot a_i a_i^T$.
For $n^2\times n^2$ real matrices $X,Y$, define $\iprod{X,Y} := \tr
X^TY/n^2 = \E_{i,j\in [n]} X_{i,j}Y_{i,j}$.  Additionally define the convex set
$\cX$ to be the set of $n^2\times n^2$ real matrices
$X=(X_{(i_1,i_2),(i_3,i_4)})_{i_1,i_2,i_3,i_4\in [n]}$ with $X\succeq
0$, $\E_{i,j\in [n]}
X_{(i,j),(i,j)}= 1$ and $X_{(i_1,i_2),(i_3,i_4)} =
X_{(i_{\pi(1)},i_{\pi(2)}), (i_{\pi(3)},i_{\pi(4)})}$ for any permutation
$\pi\in\cS_4$.   Finally, let $h_{\cX}(Y) := \max_{X\in \cX}\iprod{X,Y}$.
It is straightforward to show (c.f. \lemref{equiv}) that
\be \tensorsdp(A) = h_{\cX}(A_{2,2})
= \max_{X\in\cX} \iprod{X, A_{2,2}}.\ee
We note that if $\cX$ were defined without the symmetry constraint, it
would simply
be the convex hull of $xx^T$ for unit vectors $x\in\bbR^{n^2}$ and
$\tensorsdp(A)$ would simply be the largest eigenvalue of $A_{2,2}$.
However, we will later see that the symmetry constraint is crucial to
$\tensorsdp(A)$ being $O(1)$.

Our strategy will be to analyze $A_{2,2}$ by applying \lemref{ALPT} to
show that $A_{2,2}$ is close to $\Sigma:=\E[A_{2,2}] = \E[a_ia_i^T \ot
a_i a_i^T]$. 
First we calculate $\Sigma$.  Following the paragraph above \eq{EAx44}
we find that 
\ba \Sigma_{j_1,j_2; j_3,j_4} &=
 \delta_{j_1,j_2}\delta_{j_3,j_4} + 
 \delta_{j_1,j_3}\delta_{j_2,j_4} + 
 \delta_{j_1,j_4}\delta_{j_2,j_3} 
+ (\mu_4-3) \delta_{j_1,j_2}\delta_{j_2,j_3}\delta_{j_3,j_4}
\label{eq:Sigma-elements}
\ea
We can write this more concisely in terms of operators.
Define
\ba
\Phi & := \sum_{i\in[n]} e_i \ot e_i  
& F := \sum_{i,j \in [n]} e_i e_j^T \ot e_j  e_i^T\\
\Delta &:=\sum_{i\in[n]} e_ie_i^T \ot e_ie_i^T 
\ea 
Then we can rewrite \eq{Sigma-elements} as
\be \Sigma = I + F + \Phi \Phi^T + (\mu_4-3) \Delta.
\label{eq:Sigma-operator}\ee
This will help us compute the spectrum of $\Sigma$.  First observe
that $F(x\ot y) = y\ot x$ for any $x,y\in \R^n$ and thus has all 
eigenvalues $\pm 1$.  The spectrum of $\Delta$ is similarly bounded, since
it is a projector onto the $n$-dimensional subspace spanned by $\{e_i
\ot e_i : i\in [n]\}$. On the other hand, $\Phi\Phi^T$ has  a single
eigenvalue equal to $n$, and the rest equal to zero.
Putting this together, $\Sigma$ has one eigenvalue equal to
$n+\mu_4-1$, $n-1$ eigenvalues equal to $\mu_4-1$,
$\frac{n(n-1)}{2}$ eigenvalues equal to $2$ and 
$\frac{n(n-1)}{2}$ eigenvalues equal to 0.   

We would like to show that $A_{2,2}$ converges to $\Sigma$ using
\lemref{ALPT}.  To this end, define $\Sigma_0 := I + F + \Phi\Phi^T$
and define $b_i := \Sigma_0^{-1/2}(a_i \ot a_i)$.  Here
$\Sigma_0^{-1/2}$ refers to the square root of the pseudo-inverse of
$\Sigma_0$.  We can compute this by considering $\Sigma_0$ to be an
operator on $\vee^2\R^n$, which is defined to be the symmetric
subspace of $\R^n \ot \R^n$, i.e.~the set of $+1$-eigenvectors of
$F$.  Observe that $\|\Sigma_0^{-1/2}\|\leq 1/2$.

We will now show that $b_i$ meets the conditions of \lemref{ALPT}.
First we establish the
boundedness condition of \eq{boundedness} using standard arguments.  Observe that 
$N=\binom{n+1}{2} \leq n^2$ and that $\|b_i\|_2 \leq \|a_i \ot a_i\|_2 = \|a_i\|_2^2
 = n^{-1} \sum_j a_{i,j}^2$.   \Anote{maybe condense this by replacing
   with a suitable reference to part 3 of \lemref{subgaussian}.}
Let $\lambda,K>0$ be parameters we
 will choose later, and abbreviate $\tilde K = K\max(1,(m/n^2)^{1/4})$.  Then
\begsub{bd-calc}
 \Pr[\max_{i\in[m]} \|b_i\|_2 \geq \tilde K]
&=\Pr[\max_{i\in[m]} \exp(\|b_i\|_2) \geq \exp(\tilde K)]\\
&\leq\Pr\left[\sum_{i\in[m]} \exp(\|b_i\|_2) \geq \exp(\tilde K)\right]\\
&\leq\exp(-\tilde K\lambda)\E\left[\sum_{i\in[m]} \exp(\lambda\|b_i\|_2)\right]\\
&= \exp(-\tilde K\lambda)\sum_{i\in[m]} \E\left[\exp(\lambda n^{-1} 
\sum_{j=1}^n a_{i,j}^2)\right]\\
&= \exp(-\tilde K\lambda)m \left(\E\left[\exp(\lambda n^{-1} a_{1,1}^2)\right]\right)^n.
\endsub
Now set $\lambda = n/\psi^2$ so that this probability is $\leq
m (2e^{-\tilde K/\psi^2})^n$.  If $m\leq
n^2$ then $\tilde K = K$ and we obtain the desired bound by taking $K
= O(\psi^2)$.  If $m\geq n^2$, then $\tilde K = K
m^{1/4}/n^{1/2}$ and the probability we wish to bound is $m
2^ne^{-m^{1/4}n^{1/2}K/\psi^2}$, which again is $\leq e^{-n}$
for $K$ a sufficiently large constant and $n$ sufficiently large.

Next we would like to bound the $\psi_1$ norm of $b_i$, corresponding
to the fact that this distribution is not too ``pointy.''  To this end we will compute
\ba
\max_{x \in S(\vee^2\R^n)}\E_b \exp(|\iprod{x,b}|n/\varphi) & = 
\max_{x \in S(\vee^2\R^n)}\E_a \exp(|\iprod{x,\Sigma_0^{-1/2}(a \ot a)}|n/\varphi) \\
& = \max_{\substack{y \in \vee^2\R^n \\ \|\Sigma_0^{1/2}y\|_2\leq 1}}
\E_a \exp(|\iprod{y,(a \ot a)}|n/\varphi),
\ea
where $\varphi>0$ will be chosen later, we have assumed WLOG in the first equation that $x\in \vee^2\R^n$
and we have defined $y=\Sigma^{-1/2}x$ in the
last equation.
Let $Y$ be the $n\times n$ matrix  with $Y_{i,j} = y_{i,j}/n$.  Thus
$\iprod{y, a \ot a} = a^T Y a / n$.  To interpret the condition
$\|\Sigma_0^{1/2}y\|_2\leq 1$, observe that
\be \|\Sigma_0^{1/2}y\|_2^2 = \iprod{y, \Sigma_0 y} = \tr[\Sigma_0(Y
\ot Y)]
 = \tr(Y)^2  + 2\tr(Y^2) 
\label{eq:Y-norm}.\ee
Then the $\psi_1$ norm of $b$
is the smallest positive $\varphi$ for which
\be
\max_{\substack{Y=Y^T \\
\tr(Y)^2 + 2\tr(Y^2) 
\leq 1}}
\E_a \exp[a^T Y a / \varphi]\label{eq:Y-max}\ee
is $\leq 2$.  
Choose a $Y$ achieving the maximum in \eq{Y-max}.  Observe that
\be \|Y\|_{2\ra 2} \leq \|Y\|_F = \sqrt{\tr(Y^2)} \leq 1.
\label{eq:Y-norms-bounded}\ee

From \eq{Ea} and \eq{Ea2} we have $\E[a^T Ya] = \tr Y \leq 1$.  Applying
\lemref{H-W} and using \eq{Y-norms-bounded} we have that
\be \Pr[a^TYa \geq \tr Y + s \psi^2]\leq 
2 e^{-c\min(s,s^2)} \label{eq:quad-tail-bound}\ee
Note that for a random variable $X$, $\E[e^X] = \int_0^s ds\,
\Pr[e^X\geq s] = \int_{-\infty}^\infty dt\, e^t \Pr[X\geq t]$.
Combining this with \eq{quad-tail-bound} we can upper bound \eq{Y-max} as
\be \E_a\exp\left (\frac{a^T Y a}{\varphi}\right)
\leq \exp\Paren{\frac{\tr Y}{\varphi}}
\int_{-\infty}^\infty ds \exp\Paren{\frac{s(1+\psi^2)}{\varphi} - c\max(0,\min(s,s^2))}.
\label{eq:exp-moment-b}\ee
Using $\tr Y \leq 1$, it follows that we can take $\varphi =
O(\psi^2)$ to bound \eq{exp-moment-b} $\leq 2$.  We conclude that
the $\psi_1$ norm of the $b_i$ vectors are $\leq O(\psi^2)$.

We can now apply \lemref{ALPT}  to find that with
high probability
\be \Sigma_0^{-1/2} (A_{2,2} - \Sigma)\Sigma_0^{-1/2}
\preceq \eps I,\label{eq:op-convergence}\ee
where $\eps= C\psi^4\max(n^2/m,\sqrt{n^2/m})$ for some $C>0$.
Rearranging we find that
\be A_{2,2} \leq \Sigma + \eps \Sigma_0
 = (1+\eps)(I + F + \Phi \Phi^T) + (\mu_4-3) \Delta.\ee

To translate this operator inequality into a statement about $h_{\cX}$
we observe that 
\bit \item $h_\cX(M_1) \leq h_\cX(M_2)$ whenever $M_1\preceq M_2$; and
\item $h_{\cX}(M_1 + M_2) \leq h_{\cX}(M_1) + h_{\cX}(M_2)$.
\eit
Using these bounds in turn we obtain
\be  h_{\cX}(A_{2,2}) \leq h_{\cX}(\Sigma + \eps \Sigma_0)
\leq (1+\eps)(h_\cX(I) + h_\cX(F) + h_\cX(\Phi\Phi^T))
 + (\mu_4-3)^+h_\cX(\Delta).
\ee
Here we define $(\mu_4-3)^+ := \max(0, \mu_4-3)$.
Observe that $I$, $F$ and $\Delta$ each have largest eigenvalue equal to 1, and
so $h_\cX(I), h_\cX(F), h_\cX(\Delta) \leq 1$.  (In fact, these are each
equalities.)  

However, the single nonzero eigenvalue of $\Phi\Phi^T$
is equal to $n$.  Here we will need to use the symmetry constraint on
$\cX$.  Let $X^\Gamma$ be the matrix with entries
$X^\Gamma_{(i_1,i_2),(i_3,i_4)} := X_{(i_1,i_4), (i_3,i_2)}$.  If
$X\in \cX$ then $X=X^\Gamma$.   Additionally, $\iprod{X,Y} =
\iprod{X^\Gamma,Y^\Gamma}$.  Thus 
$$h_\cX(\Phi\Phi^T) = h_\cX((\Phi\Phi^T)^\Gamma) \leq
\|(\Phi\Phi^T)^\Gamma\|_{2\ra 2} = 1.$$
This last equality follows from the fact that $(\Phi\Phi^T)^\Gamma =
F$.

Putting together these ingredients, we establish \eq{tsdp-ub}, which
concludes the proof of the
theorem.  
\end{proof}

\subsection{Discussion}
It may seem surprising that the factor of $3^{1/4}$ emerges even for
matrices with, say, $\pm 1$ entries.  An intuitive justification for
this is that even if the columns of $A$ are not Gaussian vectors, most
linear combinations of them resemble Gaussians.  The following Lemma
shows that this behavior begins as soon as $n$ is $\omega(1)$.

\begin{lemma}\label{lem:24-LB}
Let $A = \sum_{i=1}^m e_i a_i^T / \sqrt{n}$ with $\E_i \|a_i\|_2^4\geq 1$.
Then $\tfnorm{A} \geq (3/(1+2/n))^{1/4}.$
\end{lemma}
To see that the denominator cannot be improved in general, observe
that when $n=1$ a random sign matrix will have $2\ra 4$ norm equal to
1.

\begin{proof}
Choose $x\in\R^n$ to be a random Gaussian vector such that $\E_x
\|x\|_2^2=1$.  Then
\be \E_x \|Ax\|_4^4 = \E_i \E_x n^{-2} (a_i^Tx)^4
= n^2 \E_i \E_x \iprod{a_i, x}^4 = 3 \E_i \|a_i\|_2^4 \geq 3.
\label{eq:EAx4-large}\ee
The last equality comes from the fact that $\iprod{a_i,x}$ is a
Gaussian random variable with mean zero and variance $\|a_i\|_2^2/n$.
On the other hand, $\E_x \|x\|_2^4 = 1 + 2/n$.  Thus, there must exist
an $x$ for which $\|Ax\|_4^4 / \|x\|_2^4 \geq 3 / (1+2/n)$.
\end{proof}

\begin{remark}\label{remark:complex}
It is instructive to consider a variant of the above argument.  A
simpler upper bound on the value of $\tensorsdp(A)$ is given simply by
$\|A_{2,2}\|$.  However, the presence of the $\Phi\Phi^T$ term means
that this bound will be off by an $n$-dependent factor.  Thus we
observe that the symmetry constraints of $\tensorsdp^{(4)}$ provide a
crucial advantage over the simpler bound using eigenvalues.  In the
language of quantum information (see \secref{BCY-app}), this means
that the PPT constraint is necessary for the approximation to
succeed.  See \secref{gap-DPS} for an example of this that applies to
higher levels of the hierarchy as well.

There is a similar reason why we cannot directly apply \lemref{ALPT}
to the vectors $a_i \ot a_i$.  In computing the $\psi_1$ norm for $a_i
\ot a_i$, take $x = \sqrt{n}\Phi$ in \eq{psi-p-norm}.  Then we find
that $\iprod{\sqrt{n}\Phi,a_i \ot a_i}n/\varphi =
\|a_i\|_2^2\sqrt{n}/\varphi$ and thus $\psi_1 \sim \sqrt{n}$.

On the other hand, when the $a_i$ are chosen to be random {\em
  complex} Gaussian vectors, we simply have $\E a_i a_i^* \ot a_i
a_i^* = I + F$.  In this case, the upper bound $\tensorsdp(A)\leq
\|A_{2,2}\|$ is already sufficient.  Thus, only real random
vectors demonstrate a separation between these two bounds.
\end{remark}

\begin{remark}
Our results can be seen as proving that $h_{\PPT}(M)$ is close to
$h_{\Sep}(M)$ when $M$ is of the form $\frac{1}{m} \sum_{i=1}^m a_i
a_i^T \ot a_i a_i^T$ and $a_1,\ldots,a_m$ are random vectors.  In the
case when $M$ is instead a randomly chosen projector,
Montanaro~\cite{Montanaro11} proved a 
similar bound using much more sophisticated techniques.  His methods
do not apply to our problem since our choices of $M$ are very
different than random projectors.
\end{remark}

\section{The 2-to-q norm and small-set expansion} \label{sec:sse}

In this section we show that a graph is a \emph{small-set expander} if and only if the projector to the subspace of its adjacency matrix's top eigenvalues has a bounded $2\to q$ norm for even $q \geq 4$. While the ``if'' part was known before, the ``only if'' part is novel. This characterization of small-set expanders is of general interest, and also leads to a reduction from the \smallsetexpansion problem considered in~\cite{RaghavendraS10} to the problem of obtaining a good approximation for the $2\to q$ norms.

\paragraph{Notation} For a regular graph $G=(V,E)$ and a subset $S \subseteq V$, we define the \emph{measure} of $S$ to be $\mu(S)=|S|/|V|$
and we define $G(S)$ to be the distribution obtained by picking a random
$x\in S$ and then outputting a random neighbor $y$ of $x$. We define the
\emph{expansion} of $S$, to be  $\bd_G(S)=\Pr_{y \in G(S)}[ y\not\in S]$.
 For $\delta \in (0,1)$, we define
$\bd_G(\delta)=\min_{S\subseteq V: \mu(S)\leq \delta} \bd_G(S)$. We often
drop the subscript $G$ from $\bd_G$ when it is clear from context. We
identify $G$ with its normalized adjacency (i.e., random walk) matrix. For
every $\lambda \in [-1,1]$, we denote by $V_{\ge \lambda}(G)$ the subspace
spanned by the eigenvectors of $G$ with eigenvalue at least $\lambda$.
The projector into this subspace is denoted $P_{\ge \lambda}(G)$.
%
%
For a distribution $D$, we let $\cp(D)$ denote the collision probability of
$D$ (the probability that two independent samples from $D$ are identical).

Our main theorem of this section is the following:
%

\restatetheorem{thm:sse-hyper}



One corollary of Theorem~\ref{thm:sse-hyper} is that a good approximation to the $2\to q$ norm implies an approximation of $\bd_{\delta}(G)$ 

\begin{corollary} \label{thm:ssetohyper} If there is a polynomial-time computable relaxation $\cR$ yielding good approximation for the $2\to q$, 
then the \emph{Small-Set Expansion Hypothesis} of \cite{RaghavendraS10} is false.
\end{corollary}
\begin{proof} Using~\cite{RaghavendraST10}, to refute the small-set
  expansion hypothesis it is enough to come up with an efficient algorithm that given an input graph $G$ and sufficiently small $\delta>0$, can distinguish between the \emph{Yes} case: $\bd_G(\delta) < 0.1$ and the \emph{No} case $\bd_G(\delta') > 1-2^{-c\log(1/\delta')}$ for any $\delta'\geq \delta$ and some constant $c$. In particular for all $\eta>0$, if $\delta$ is small enough then in the \emph{No} case $\bd_G(\delta^{0.4}) > 1-\eta$. 
  
  Using Theorem~\ref{thm:sse-hyper}, in the \emph{Yes} case we know $\Vert P_{\geq 1/2} \Vert_{2 \rightarrow q} \geq  1/(10\delta^{(q-2)/2q})$, while in the \emph{No} case, if we choose $\delta$ sufficiently small so that $\eta$ is smaller than $c_1(1/2)^q 2^{-c_2 q}$, then we know that $ \Vert P_{\geq 1/2} \Vert_{2 \rightarrow q}   \leq 2/\sqrt{\delta^{0.2}}$. Clearly, if we have a good approximation for the $2\to q$ norm then, for sufficiently small $\delta$, we can distinguish between these two cases.
\end{proof}

The first part of Theorem~\ref{thm:sse-hyper} follows from previous work (e.g., see~\cite{KhotV05}). For completeness, we include a proof in Appendix~\ref{app:hyper-imp-sse}.  The second part will follow from the following lemma:

\begin{lemma}\label{lem:ssetonorm} Set $e = e(\lambda, q) := c_1 2^{c_2 q} /\lambda^q$, with universal constants $c_1, c_2 >0$. Then for every $\lambda>0$ and $1 \geq \delta \geq 0$, if $G$ is a graph that satisfies $\cp(G(S)) \leq 1/(e|S|)$ for all $S$ with $\mu(S)\leq \delta$, then $\norm{f}_q \leq 2\norm{f}_2/\sqrt{\delta}$ for all $f\in V_{\geq \lambda}(G)$.
\end{lemma}

\paragraph{Proving the second part of Theorem~\ref{thm:sse-hyper} from Lemma~\ref{lem:ssetonorm}} We use the variant of the local Cheeger bound obtained in~\cite[Theorem 2.1]{Steurer10c},
stating that if $\bd_{G}(\delta) \geq 1 - \eta$ then for every $f\in \L2(V)$ satisfying $\norm{f}_1^2 \leq \delta \norm{f}_2^2$, $\norm{Gf}_2^2 \leq c\sqrt{\eta}\norm{f}_2^2$.
The proof follows by noting that for every set $S$, if $f$ is the characteristic function of $S$ then $\norm{f}_1 = \norm{f}_2^2 = \mu(S)$, and $\cp(G(S)) = \norm{Gf}_2^2 / (\mu (S) |S|)$. \qed

\begin{proof}[Proof of Lemma~\ref{lem:ssetonorm}]  Fix $\lambda>0$. We assume that the graph satisfies the condition of the Lemma with $c_1 2^{c_2 q} /\lambda^q$, for constants $c_1, c_2$ that we will set later. Let $G=(V,E)$ be such a graph, and $f$ be function in $V_{\geq \lambda}(G)$ with $\norm{f}_2=1$ that maximizes $\norm{f}_q$. We write $f = \sum_{i=1}^m \alpha_i \chi_i$ where $\chi_1,\ldots,\chi_m$ denote the eigenfunctions of $G$ with values $\lambda_1,\ldots,\lambda_m$ that are at least $\lambda$.  Assume towards a contradiction that $\norm{f}_q>2/\sqrt{\delta}$. We'll prove that  $g= \sum_{i=1}^m (\alpha_i/\lambda_i)\chi_i$ satisfies $\norm{g}_q \geq 5\norm{f}_q/\lambda$. This is a contradiction since (using $\lambda_i \in [\lambda,1]$) $\norm{g}_2 \leq \norm{f}_2/\lambda$, and we assumed $f$ is a function in $V_{ \geq \lambda}(G)$ with a maximal ratio of $\norm{f}_q/\norm{f}_2$.

Let $U \subseteq V$ be the set of vertices such that $|f(x)| \geq 1/\sqrt{\delta}$ for all $x\in U$. Using Markov inequality and the fact that $\E_{x\in V}[ f(x)^2 ] = 1$, we know that $\mu(U)=|U|/|V| \leq \delta$, meaning that under our assumptions any subset $S\subseteq U$ satisfies $\cp(G(S))\leq 1/(e|S|)$. On the other hand, because $\norm{f}_q^q \geq 2^q/\delta^{q/2}$, we know that $U$ contributes at least half of the term $\norm{f}_q^q = \E_{x\in V} f(x)^q$. That is, if we define $\alpha$ to be $\mu(U)\E_{x\in U} f(x)^q$ then $\alpha \geq \norm{f}_q^q/2$. We will prove the lemma by showing that $\norm{g}_q^q \geq \left(10\lambda^{-1} \right)^{q} \alpha$.

Let $c$ be a sufficiently large constant ($c=100$ will do). We define $U_i$ to be the set $\{x \in U : f(x) \in [c^i/\sqrt{\delta},c^{i+1}/\sqrt{\delta}) \}$, and let $I$ be the maximal $i$ such that $U_i$ is non-empty. Thus, the sets $U_0,\ldots,U_I$ form a partition of $U$ (where some of these sets may be empty). We let $\alpha_i$ be the contribution of $U_i$ to $\alpha$. That is, $\alpha_i = \mu_i\E_{x\in U_i} f(x)^q$, where $\mu_i=\mu(U_i)$. Note that $\alpha = \alpha_0 + \cdots + \alpha_I$. We'll show that there are some indices $i_1,\ldots,i_J$ such that:

\begin{description}

\item[(i)] $\alpha_{i_1} + \cdots + \alpha_{i_J} \geq \alpha/(2c^{10})$.

\item[(ii)] For all $j\in [J]$, there is a non-negative function $g_j:V\to\R$ such that $\E_{x\in V}g_j(x)^q \geq e\alpha_{i_j}/(10c^2)^{q/2}$.

\item[(iii)] For every $x\in V$, $g_1(x)  + \cdots + g_J(x) \leq |g(x)|$.
\end{description}

Showing these will complete the proof, since it is easy to see that for two non-negative functions $g',g''$ and even integer $q$, $\E(g'(x)+g''(x))^q \geq \E g'(x)^q + \E g''(x)^q$, and hence \textbf{(ii)} and \textbf{(iii)} imply that
\begin{equation}
\norm{g}_q^q= \E g(x)^q \geq (e/(10c^2)^{q/2}) \sum_j \alpha_{i_j} \;. \label{eq:ssetonorm-finalbound}
\end{equation}
Using \textbf{(i)} we conclude that for $e \geq  2c^{10} 10^q (10c^2)^{q/2}/\lambda^q$, the right-hand side of (\ref{eq:ssetonorm-finalbound}) will be larger than $(10/\lambda)^q\alpha$.

We find the indices $i_1,\ldots,i_J$ iteratively. We let $\cI$ be initially the set $\{0..I\}$ of all indices. For $j=1,2,...$ we do the following as long as $\cI$ is not empty:
\begin{enumerate}
\item Let $i_j$ be the largest index in $\cI$.

\item Remove from $\cI$ every index $i$ such that $\alpha_i \leq c^{10}\alpha_{i_j}/2^{i-i_j}$.

\end{enumerate}

We let $J$ denote the step when we stop. Note that our indices $i_1,\ldots,i_J$ are sorted in descending order. For every step $j$, the total of the $\alpha_i$'s for all indices we removed is less than $c^{10}\alpha_{i_j}$ and hence we satisfy \textbf{(i)}. The crux of our argument will be to show \textbf{(ii)} and \textbf{(iii)}. They will follow from the following claim:

\begin{claim}\label{clm:sse-norm}  Let $S\subseteq V$ and $\beta>0$ be such that $|S| \leq \delta$ and $|f(x)| \geq \beta$ for all $x\in S$. Then there is a set $T$ of size at least $e|S|$ such that $\E_{x\in T} g(x)^2 \geq \beta^2/4$.
\end{claim}
The claim will follow from the following lemma:

\begin{lemma} \label{lem:sse-norm} Let $D$ be a distribution with $\cp(D) \leq 1/N$ and $g$ be some function. Then there is a set $T$ of size $N$ such that $\E_{x\in T} g(x)^2 \geq (\E  g(D))^2/4$.
\end{lemma}
\begin{proof} Identify the support of $D$ with the set $[M]$ for some $M$, we let $p_i$ denote the probability that $D$ outputs $i$, and sort the $p_i$'s such that $p_1\geq p_2 \cdots p_M$. We let $\beta'$ denote $\E g(D)$; that is, $\beta'=\sum_{i=1}^M p_ig(i)$. We separate to two cases. If $\sum_{i>N} p_ig(i) \geq \beta'/2$, we define the distribution $D'$ as follows: we set $\Pr [ D' = i]$ to be $p_i$ for $i>N$, and we let all $i\leq N$ be equiprobable (that is be output with probability $(\sum_{i=1}^N p_i)/N$). Clearly, $\E  |g(D')|  \geq \sum_{i>N} p_ig(i) \geq \beta'/2$, but on the other hand, since the maximum probability of any element in $D'$ is at most $1/N$, it can be expressed as a convex combination of flat distributions over sets of size $N$, implying that one of these sets $T$ satisfies $\E_{x\in T} |g(x)|  \geq \beta'/2$, and hence $\E_{x\in T} g(x)^2  \geq \beta'^2/4$.

The other case is that $\sum_{i=1}^N p_ig(i) \geq \beta'/2$. In this case we use Cauchy-Schwarz and argue that
\begin{equation}
\beta'^2/4 \leq \left(\sum_{i=1}^N p_i^2\right)\left(\sum_{i=1}^N g(i)^2 \right) \;. \label{eq:cs-sse-norm}
\end{equation}
But using our bound on the collision probability, the right-hand side of (\ref{eq:cs-sse-norm}) is upper bounded by  $\tfrac{1}{N}\sum_{i=1}^N g(i)^2 = \E_{x\in[N]} g(x)^2$.
\end{proof}
\begin{proof}[Proof of Claim~\ref{clm:sse-norm} from Lemma~\ref{lem:sse-norm}] By construction $f=Gg$, and hence we know that for every $x$, $f(x)=\E_{y\sim x} g(y)$. This means that if we let $D$ be the distribution $G(S)$ then
\[
\E |g(D)| = \E_{x\in S} \E_{y\sim x} |g(y)|  \geq  \E_{x\in S} |\E_{y\sim x}  g(y) | = \E_{x\in S} |f(x)| \geq \beta \;.
\]
By the expansion property of $G$, $\cp(D) \leq 1/(e|S|)$ and thus by Lemma~\ref{lem:sse-norm} there is a set $T$ of size $e|S|$ satisfying  $\E_{x\in T} g(x)^2 \geq \beta^2/4$.
\end{proof}

We will construct the functions $g_1,\ldots,g_J$ by applying iteratively Claim~\ref{clm:sse-norm}. We do the following for $j=1,\ldots,J$:

\begin{enumerate}

\item Let $T_j$ be the set of size $e|U_{i_j}|$ that is obtained by applying Claim~\ref{clm:sse-norm} to the function $f$ and the set $U_{i_j}$. Note that $\E_{x \in T_j} g(x)^2 \geq \beta_{i_j}^2/4$, where we let $\beta_i = c^i/\sqrt{\delta}$ (and hence for every $x\in U_i$, $\beta_i \leq |f(x)| \leq c\beta_i$).

\item Let $g'_j$ be the function on input $x$ that outputs $\gamma\cdot |g(x)|$ if $x\in T_j$ and $0$ otherwise, where $\gamma \leq 1$ is a scaling factor that ensures that $\E_{x \in T_j} g'(x)^2 $  equals exactly $\beta_{i_j}^2/4$.

\item We define $g_j(x) = \max \{ 0 , g'_j(x) - \sum_{k<j} g_k(x) \}$.

\end{enumerate}

Note that the second step ensures that $g'_j(x) \leq |g(x)|$, while the third step ensures that $g_1(x) + \cdots + g_j(x) \leq g'_j(x)$ for all $j$, and in particular $g_1(x) + \cdots + g_J(x) \leq |g(x)|$.  Hence the only thing left to prove is the following:

\begin{claim} $\E_{x\in V} g_j(x)^q \geq e\alpha_{i_j}/(10c)^{q/2}$
\end{claim}

\begin{proof} Recall that for every $i$, $\alpha_i = \mu_i \E_{x\in U_i} f(x)^q$, and hence  (using $f(x) \in [\beta_i,c\beta_i)$ for $x\in U_i$):
\begin{equation}
 \mu_i\beta_i^q \leq \alpha_i \leq \mu_ic^q\beta_i^q \;.  \label{eq:alpha-i-bound}
 \end{equation}

Now fix $T=T_j$. Since $\E_{x\in V} g_j(x)^q$ is at least (in fact equal) $\mu(T)\E_{x\in T} g_j(x)^q$ and $\mu(T) = e\mu(U_{i_j})$, we can use (\ref{eq:alpha-i-bound}) and $\E_{x\in T} g_j(x)^q \geq (E_{x\in T} g_j(x)^2)^{q/2}$, to reduce proving the claim to showing the following:
\begin{equation}
\E_{x\in T} g_j(x)^2  \geq (c\beta_{i_j})^2/(10c^2) = \beta_{i_j}^2/10 \;. \label{eq:lower-bound-gj}
\end{equation}

We know that  $\E_{x\in T} g'_j(x)^2  = \beta_{i_j}^2/4$. We claim  that (\ref{eq:lower-bound-gj}) will follow by showing that for every $k<j$,
\begin{equation}
\E_{x\in T} g'_k(x)^2  \leq 100^{-i'}\cdot \beta_{i_j}^2/4 \;, \label{eq:upper-bound-gk}
\end{equation}
where $i' = i_k - i_j$. (Note that $i'>0$ since in our construction the indices $i_1,\ldots,i_J$ are sorted in descending order.)

Indeed, (\ref{eq:upper-bound-gk}) means that if we let momentarily $\norm{g_j}$ denote $\sqrt{\E_{x\in T} g_j(x)^2}$ then
\begin{equation}
  \norm{g_j} \geq \norm{g'_j} - \norm{\tsum_{k<j} g_k} \geq
  \norm{g'_j}-\sum_{k<j}\norm{g_k}\ge  \norm{g'_j}(1-  \sum_{i'=1}^{\infty} 10^{-i'}) \geq  0.8\norm{g'_j} \;. \label{eq:norm-subtract}
\end{equation}
The first inequality holds because we can write $g_j$ as $g'_j - h_j$,
where $h_j=\min\set{g'_j,\sum_{k<j}g_k}$.
Then, on the one hand, $\norm{g_j}\ge \norm{g'_j}-\norm{h_j}$, and on the
other hand, $\norm{h_j}\le \norm{\sum_{k<j}g_k}$
since $g'_j\ge 0$.
The second inequality holds because $\norm{g_k}\le \norm{g'_k}$.
%
By squaring (\ref{eq:norm-subtract}) and plugging in the value of $\norm{g'_j}^2$ we get (\ref{eq:lower-bound-gj}).

\paragraph{Proof of (\ref{eq:upper-bound-gk})} By our construction, it must hold that
\begin{equation}
c^{10}\alpha_{i_k}/2^{i'} \leq \alpha_{i_j}  \;, \label{eq:sse-norm-relate-levelsets}
\end{equation}
since otherwise the index $i_j$ would have been removed from the $\cI$ at the $k^{th}$ step. Since $\beta_{i_k} = \beta_{i_j}c^{i'}$, we can plug  (\ref{eq:alpha-i-bound}) in (\ref{eq:sse-norm-relate-levelsets}) to get
\[
\mu_{i_k}c^{10+4i'}/2^{i'} \leq c^4\mu_{i_j}
\]
or
\[
\mu_{i_k} \leq \mu_{i_j}(2/c)^{4i'}c^{-6} \;.
\]

Since $|T_i| = e|U_i|$ for all $i$, it follows that $|T_{k}|/|T| \leq (2/c)^{4i'}c^{-6}$. On the other hand, we know that $\E_{x\in T_{k}} g'_k(x)^2 = \beta_{i_k}^2/4 = c^{2i'}\beta^2_{i_j}/4$. Thus,
\[
\E_{x\in T} g'_k(x)^2  \leq 2^{4i'}c^{2i'-4i'-6}\beta_{i_j}^2/4 \leq (2^4/c^2)^{i'}\beta_{i_j}^2/4 \;,
\] and now we just choose $c$ sufficiently large so that $c^2/2^4 > 100$.
\end{proof}
\end{proof}

\section{Relating the 2-to-4 norm and the injective tensor norm} \label{sec:ITN}

In this section, we present several equivalent formulations of the 2-to-4 norm: 1) as the injective tensor norm of a 4-tensor, 2) as the injective tensor norm
of a 3-tensor, and 3) as the maximum of a linear function over a convex set, albeit a set where the weak membership problem is hard.  Additionally, we can
consider maximizations over real or complex vectors.  These equivalent formulations are discussed in \secref{inj-equiv}.

We use this to show hardness of approximation (\thmref{hardness}) for the 2-to-4 norm in \secref{hardness}, and then show positive algorithmic results
(\thmref{BCY}) in \secref{BCY-app}.  Somewhat surprisingly, many of the key arguments in these sections are imported from the quantum information literature,
even though no quantum algorithms are involved.  It is an interesting question to find a more elementary proof of the result in \secref{BCY-app}.

\ifcount{In this section, it will be convenient to sometimes work with
  the counting norms $\cnorm{.}$, which we recall are defined as
  $\cnorm{x}_p := (\sum_i |x_i|^p)^{1/p}$, and the counting inner
  product, defined by $\ciprod{x,y} := x^*y$, where $^*$ denotes the
  conjugate transpose.}

\subsection{Equivalent maximizations}\label{sec:inj-equiv}


\subsubsection{Injective tensor norm and separable states}

Recall from the introduction the definition of the injective tensor norm: if $V_1,\ldots,V_r$ are vector spaces with $T\in V_1 \ot \cdots \ot V_r$, then
$\cnorm{T}_{\inj} = \max \{ | \ciprod{T,(x_1\ot\cdots\ot x_r)}| : x_1\in \bS(V_1),\ldots,x_r\in \bS(V_r)\}$, where $\bS(V)$ denotes the $L_2$-unit vectors in a
vector space $V$. In this paper we use the term ``injective tensor norm'' to mean the injective tensor norm of $\ell_2$ spaces, and we caution the reader that
in other contexts it has a more general meaning.  These norms were introduced by Grothendieck, and they are further discussed in~\cite{Ryan02}.

We will also need the definition of separable states from quantum information.  For a vector space $V$, define $L(V)$ to be the linear operators on $V$, and define $\cD(V) := \{\rho \in L(V): \rho\succeq 0, \tr\rho=1\} = \conv \{
vv^* : v\in \bS(V)\}$ to be the {\em density operators} on $V$.  \ifexp{Here we use $^*$ to denote the conjugate transpose, and define the trace to be the normalized trace; that is, the average of the diagonal elements of a matrix.}  The trace induces an inner product on operators: $\ciprod{X,Y} :=\tr X^*Y$.  An important class of density operators are the {\em separable density operators}.  For vector spaces $V_1,\ldots,V_r$, these are
$$ \Sep(V_1,\ldots,V_r)  :=
\conv\left\{ v_1v_1^* \ot \cdots \ot v_r v_r^* : \forall i, v_i\in \bS(V_i)\right\}.$$
If $V=V_1=\cdots=V_r$, then let $\Sep^r(V)$ denote $\Sep(V_1,\ldots,V_r)$.  Physically, density operators are the quantum analogues of probability distributions, and separable density operators describe unentangled quantum states; conversely, entangled states are defined to be the set of density operators that are not separable.  For readers familiar with quantum information, we point out that our treatment differs principally in its use of the expectation for norms and inner products, rather than the sum.

For any bounded convex set $K$, define the {\em support function} of $K$
to be
$$\bh_K(x) := \max_{y\in K} |\ciprod{x,y}|.$$
Define $e_i\in \F^n$ to be the vector with 1 in the $i^{\text{th}}$ position\ifexp{; beware that $\|e_i\|_2=1/\sqrt{n}$}.
Now we can give the convex-optimization formulation of the injective tensor norm.
\begin{lemma}\label{lem:h-sep}
Let $V_1,\ldots,V_r$ be vector spaces with  $n_i := \dim V_i$, and $T\in V_1\ot \cdots \ot V_r$.
\ifcount{Choose an orthonormal basis $e_1,\ldots,e_{n_r}$ for $V_r$.}
\ifexp{Choose an orthogonal basis $e_1,\ldots,e_{n_r}$ for $V_r$, normalized so that $\|e_i\|_2=1/\sqrt{n_r}$.}
Define $T_1,\ldots,T_{n_r} \in V_1\ot \cdots \ot V_{r_1}$ by $T = \sum_{i=1}^{n_r} T_i \ot e_i$ and define $M \in L(V_1 \ot \cdots \ot V_{r-1})$ by $M = \ifexp{\frac{1}{n_1n_2\cdots n_r}} \sum_{i=1}^{n_r} T_i T_i^*$.  Then
\be \cnorm{T}_{\inj}^2 = \bh_{\Sep(V_1,\ldots,V_{r-1})}(M).\ee

\end{lemma}

Observe that any $M\succeq 0$ can be expressed in this form, possibly by padding $n_r$ to be at least $\rank M$.  Thus calculating $\cnorm{\cdot}_{\rm inj}$ for $r$-tensors is equivalent in difficulty to computing $\bh_{\Sep^{r-1}}$ for p.s.d.~arguments.  This argument appeared before in \cite{HarrowM10}, where it was explained using quantum information terminology.

It is instructive to consider the $r=2$ case.   In this case, $T$ is equivalent to a matrix $\hat T$ and $\ifexp{\sqrt{n_1n_2}}\cnorm{T}_{\rm inj} = \cnorm{\hat T}_{2\ra 2}$.
Moreover $\Sep^1(\F^{n_1})= \cD(\F^{n_1})$ is simply the convex hull of $vv^*$ for unit vectors $v$.  Thus $\bh_{\Sep^1(\F^{n_1})}(M)$ is simply the maximum eigenvalue of $M=TT^*\ifexp{/n_1n_2}$.  In this case,  \lemref{h-sep} merely states that the square of the largest singular value of $\hat T$ is the largest eigenvalue of $\hat T\hat T^*$.  The general proof follows this framework.

\begin{proof}[Proof of \lemref{h-sep}]
\begin{align}
\cnorm{T}_{{\rm inj}} &= \max_{x_1\in \bS(V_1),\ldots,x_r\in \bS(V_r)}
\left|\cbigiprod{T,x_1\ot\cdots\ot x_r}\right|
\\ &= \max_{x_1\in \bS(V_1),\ldots,x_{r-1}\in \bS(V_{r-1})}\max_{x_r\in \bS(V_r)}
 \left|\sum_{i=1}^n \cbigiprod{T_i,x_1\ot\cdots\ot
  x_{r-1}} \cdot \ciprod{e_i,x_r}\right|
\\ &= \max_{x_1\in \bS(V_1),\ldots,x_{r-1}\in \bS(V_{r-1})}
\cBignorm{\sum_{i=1}^n \cbigiprod{T_i,x_1\ot\cdots\ot
  x_{r-1}} e_i}_2
\end{align}

Therefore
\begin{align}
\cnorm{T}_{{\rm inj}}^2 &= \max_{x_1\in \bS(V_1),\ldots,x_{r-1}\in \bS(V_{r-1})}
\cnorm{ \sum_{i=1}^{n_r} \ciprod{T_i,x_1\ot\cdots\ot
  x_{r-1}} e_i }_2^2
\\ &= \max_{x_1\in \bS(V_1),\ldots,x_{r-1}\in \bS(V_{r-1})}
\ifexp{\frac{1}{n_r}} \sum_{i=1}^{n_r} \left|\cIprod{T_i,x_1\ot\cdots\ot
  x_{r-1}}  \right|_2^2
\\ & = \max_{x_1\in \bS(V_1),\ldots,x_{r-1}\in \bS(V_{r-1})}
\cIprod{ \ifexp{\frac{1}{n_r}}\sum_{i=1}^{n_r}  T_iT_i^*,
x_1x_1^* \ot \cdots \ot x_r x_r^*}
\\ & =
\bh_{\Sep(V_1,\ldots,V_{r_1})} \left(\ifexp{\frac{1}{n_1\cdots n_r}}\sum_{i=1}^{n_r}  T_iT_i^*\right)
\end{align}
\end{proof}

In what follows, we will also need to make use of some properties of symmetric tensors.  Define $\cS_k$ to be the group of permutations of $[k]$ and define $P_n(\pi)\in L((\F^n)^{\ot k})$ to be the operator that permutes $k$ tensor copies of $\F^n$ according to $\pi$.  Formally,
\ba P_n(\pi) := \sum_{i_1,\ldots,i_r \in [d]} \bigotimes_{k=1}^r e_{i_k} e_{i_{\pi(k)}}^T.
\label{eq:perm-def}\ea
Then define $\vee^k\F^n$ to be the subspace of vectors in $(\F^n)^{\ot r}$ that are unchanged by each $P_n(\pi)$.  This space is called the {\em symmetric subspace.}  A classic result in symmetric polynomials states that $\vee^r\F^n$ is spanned by the vectors $\{v^{\ot r}: v\in \F^n\}$.\footnote{For the proof, observe that $v^{\ot r}\in\vee^r\F^n$ for any $v\in \F^n$.  To construct a basis for $\vee^r\F^n$ out of linear combinations of different $v^{\ot r}$, let $z_1,\ldots,z_n$ be indeterminates and evaluate the $r$-fold derivatives of $(z_1e_1 + \cdots + z_ne_n)^{\ot r}$ at $z_1=\cdots=z_n=0$.}

One important fact about symmetric tensors is that for injective tensor norm, the vectors in the maximization can be taken to be equal.  Formally,
\begin{fact}\label{fact:sym-equal}
If $T\in \vee^r \F^n$ then
\be \cnorm{T}_{\inj} = \max_{x\in \bS(\F^n)} |\ciprod{T,x^{\ot r}}|.\ee
\end{fact}
This has been proven in several different works; see the paragraph above Eq.~(3.1) of \cite{CobosKP00} for references.


\subsubsection{Connection to the 2-to-4 norm}

Let $A = \sum_{i=1}^m e_i a_i^T$, so that $a_1,\ldots,a_m\in \R^n$ are the rows of $A$.
Define
\begin{align}
A_4 &= \ifexp{\frac{n^4}{m}}
 \sum_{i=1}^m a_i^{\ot 4} &\in (\R^n)^{\ot 4}
\label{eq:A4-def}\\
A_3 &= \ifexp{n^2}\sum_{i=1}^m a_i \ot a_i \ot e_i &\in \R^n \ot \R^n \ot \R^m  \\
A_{2,2} &=\ifexp{\frac{n^2}{m}} \sum_{i=1}^m a_ia_i^T \ot a_ia_i^T &\in L((\R^n)^{\ot 2})
\label{eq:A22-def}
\end{align}
The subscripts indicate that that $A_r$ is an $r$-tensor, and $A_{r,s}$ is a map from $r$-tensors to $s$-tensors.

Further, for a real tensor $T \in (\bbR^n)^{\ot r}$, define $\cnorm{T}_{\inj[\C]}$ to be the injective tensor norm that results from treating $T$ as a complex tensor; that is, $\max\{|\ciprod{T, x_1\ot\cdots\ot x_r}|: x_1,\ldots,x_r\in \bS(\C^n)\}$.  For $r\geq 3$, $\cnorm{T}_{\inj[\C]}$ can be larger than $\cnorm{T}_{\inj}$ by as much as $\sqrt{2}$~\cite{CobosKP00}.

Our main result on equivalent forms of the $2\ra 4$ norm is the following.
\begin{lemma}\label{lem:equiv}
$$\cnorm{A}_{2\ra 4}^4 = \cnorm{A_4}_{\inj} = \cnorm{A_3}_{\rm inj}^2
= \cnorm{A_4}_{\inj[\C]} =\cnorm{A_3}_{{\rm inj}[\C]}^2
= \bh_{\Sep^2(\R^n)}(A_{2,2})
= \bh_{\Sep^2(\C^n)}(A_{2,2})$$
\end{lemma}

\begin{proof}
\begin{align}
\cnorm{A}_{2\ra 4}^4
&= \max_{x\in \bS(\R^n)}\ifexp{\frac{n^4}{m}}
\sum_{i=1}^m \cIprod{a_i, x}^4
\\&= \max_{x\in \bS(\R^n)} \cIprod{A_4, x^{\ot 4}}
\\& = \max_{x_1,x_2,x_3,x_4\in \bS(\R^n)}  \left|\cIprod{A_4, x_1\ot x_2\ot x_3 \ot x_4}\right|
\label{eq:A4-sym}
\\ &= \cnorm{A_4}_{\inj} \label{eq:A4-inj-equiv}
\end{align}
Here \eq{A4-sym} follows from Fact~\ref{fact:sym-equal}.

Next one can verify with direct calculation (and using $\max_{z\in \bS(\R^n)} \ciprod{v,z} = \cnorm{v}_2$) that
\be
\max_{x\in \bS(\R^n)} \cIprod{A_4, x^{\ot 4}}
 = \max_{x\in \bS(\R^n)} \cIprod{ A_{2,2}, xx^T \ot xx^T}
= \max_{x\in \bS(\R^n)} \max_{z\in \bS(\R^m)} \cIprod{A_3, x\ot x\ot z}^2.\ee
Now define $z(i):=\ciprod{e_i,z}$ and continue.
\begin{align}
 \max_{x\in \bS(\R^n)} \max_{z\in \bS(\R^m)}\left|\cIprod{A_3, x\ot x\ot z}\right|
&= \max_{x\in \bS(\R^n)}\max_{z\in \bS(\R^m)} \text{Re}\sum_{i=1}^m z(i) \ciprod{a_i,x}^2
\\&= \max_{x\in \bS(\R^n)}\max_{z\in \bS(\C^m)} \text{Re}\sum_{i=1}^m z(i) \ciprod{a_i,x}^2
\\&= \ifexp{\frac{1}{n^2}} \max_{z\in \bS(\C^m)}
\cNorm{\sum_{i=1}^m z(i) a_ia_i^T }_{2\ra 2}
\label{eq:z-op-norm}
\\&= \max_{z\in \bS(\C^m)} \max_{x,y\in \bS(\C^n)}
\text{Re} \sum_{i=1}^m z(i) \ciprod{x^*,a_i}\ciprod{a_i,y}
\label{eq:xy-complex}
\\ &= \cnorm{A_3}_{\inj[\C]} = \cnorm{A_3}_{\inj}
\end{align}
From \lemref{h-sep}, we thus have $\cnorm{A}_{2\ra 4}^4 = \bh_{\Sep^2(\R^n)}(A_{2,2})
= \bh_{\Sep^2(\C^n)}(A_{2,2})$.

To justify \eq{xy-complex}, we argue that the maximum in \eq{z-op-norm} is achieved by taking all the $z(i)$ real (and indeed nonnegative).  The resulting matrix $\sum_i z(i) a_ia_i^T$ is real and symmetric, so its operator norm is achieved by taking $x=y$ to be real vectors.  Thus, the maximum in $\cnorm{A_3}_{\inj[\C]}$ is achieved for real $x,y,z$ and as a result $\cnorm{A_3}_{\inj[\C]}=\cnorm{A_3}_{\inj}$.

Having now made the bridge to complex vectors, we can work backwards to establish the last equivalence: $\|A_4\|_{\inj[\C]}$.  Repeating the argument that led to \eq{A4-inj-equiv}  will establish that $\cnorm{A_4}_{\inj[\C]} = \max_{x\in \bS(\C^n)}\max_{z\in \bS(\C^m)} |\cIprod{A_3, x\ot x\ot z}|^2 = \cnorm{A_3}_{\inj[\C]}^2$.
\end{proof}

\subsection{Hardness of approximation for the 2-to-4 norm}\label{sec:hardness}

This section is devoted to the proof of \thmref{hardness}, establishing hardness of approximation for the 2-to-4 norm.

First, we restate \thmref{hardness} more precisely.  We omit the
reduction to when $A$ is a projector, deferring this argument to
\corref{proj-hard}, where we will further use a randomized reduction.
\begin{theorem}\label{thm:hardness-formal} (restatement of \thmref{hardness})
Let $\phi$ be a 3-SAT instance with $n$ variables and $O(n)$ clauses.
Determining whether $\phi$ is satisfiable can be reduced in polynomial
time to determining whether $\|A\|_{2\ra 4} \geq C$ or $\|A\|_{2\ra
  4}\leq c$ where $0\leq c < C$ and $A$ is an $m\times m$ matrix.  This is possible for two choices of parameters:
\benum
\item $m=\poly(n)$, and $C/c > 1 + 1/n\poly\log(n)$; or,
\item $m=\exp(\sqrt{n}\poly\log(n)\log(C/c))$.
\eenum
\end{theorem}

The key challenge is establishing the following reduction.
\begin{lemma}\label{lem:approx-preserving}
Let $M\in L(\C^n \ot \C^n)$ satisfy $0\leq M \leq I$.  Assume that either (case Y) $\bh_{\Sep(n,n)}(M)=1$ or (case N) $\bh_{\Sep(n,n)}(M)\leq 1-\delta$.  Let $k$ be a positive integer.
Then there exists a matrix $A$ of size $n^{4k}\times n^{2k}$ such that in case Y, $\|A\|_{2\ra 4}=1$, and in case N, $\|A\|_{2\ra 4}=(1-\delta/2)^k$.  Moreover, $A$ can be constructed efficiently from $M$.
\end{lemma}

\begin{proof}[Proof of \thmref{hardness-formal}]
Once \lemref{approx-preserving} is proved, \thmref{hardness} follows from previously known results about the hardness of approximating $\bh_{\Sep})$.  Let $\phi$ be a 3-SAT instance with $n$ variables and $O(n)$ clauses.  In Theorem 4 of \cite{GallNN11} (improving on earlier work of \cite{Gurvits03}), it was proved that $\phi$ can be reduced to determining whether $\bh_{\Sep(n^c,n^c)}(M)$ is equal to 1 (``case Y'') or $\leq 1-1/n\log^c(n)$ (``case N''), where $c>0$ is a universal constant, and $M$ is an efficiently constructible matrix with $0\leq M \leq I$.  Now we apply \lemref{approx-preserving} with $k=1$ to find that exists a matrix $A$ of dimension $\poly(n)$ such that in case Y, $\tfcnorm{A}=1$, and in case N, $\tfcnorm{A}\leq 1- 1/2n\log^c(n)$.  Thus, distinguishing these cases would determine whether $\phi$ is satisfiable.  This establishes part (1) of \thmref{hardness}.

For part (2), we start with Corollary 14 of \cite{HarrowM10}, which gives a reduction from determining the satisfiability of $\phi$ to distinguishing between (``case Y'') $\bh_{\Sep(m,m)}(M)=1$ and (``case N'') $\bh_{\Sep(m,m)}(M)\leq 1/2$.  Again $0\leq M \leq I$, and $M$ can be constructed in time $\poly(m)$ from $\phi$, but this time $m = \exp(\sqrt{n}\poly\log(n))$.  Applying \lemref{approx-preserving} in a similar fashion completes the proof.
\end{proof}

\begin{proof}[Proof of \lemref{approx-preserving}]
  The previous section shows that computing $\tfcnorm{A}$ is equivalent to computing $\bh_{\Sep(n,n)}(A_{2,2})$, for $A_{2,2}$ defined as in \eq{A22-def}.  However, the hardness results of \cite{Gurvits03,GallNN11,HarrowM10} produce matrices $M$ that are not in the form of $A_{2,2}$.  The reduction of \cite{HarrowM10} comes closest, by producing a matrix that is a sum of terms of the form $xx^* \ot yy^*$. However, we need a sum of terms of the form $xx^*\ot xx^*$.  This will be achieved by a variant of the protocol used in \cite{HarrowM10}.

 Let $M_0 \in L(\C^n \ot \C^n)$ satisfy $0\leq M \leq I$.  Consider the promise problem of distinguishing the cases $\bh_{\Sep(n,n)}(M_0) = 1$ (called ``case Y'') from $\bh_{\Sep(n,n)}(M_0)\leq 1/2$ (called ``case N'').  We show that this reduces to
finding a multiplicative approximation for $\cnorm{A}_{2\ra 4}$ for some real $A$ of dimension $n^{\alpha}$ for a constant $\alpha>0$.
Combined with known hardness-of-approximation results (Corollary 15 of \cite{HarrowM10}), this will imply \thmref{hardness}.

Define $P$ to be the projector onto the subspace of $(\C^n)^{\ot 4}$ that is invariant under $P_n((1,3))$ and $P_n((2,4))$ (see \secref{inj-equiv} for definitions).  This can be obtained by applying $P_n((2,3))$ to $\vee^2\C^n \ot \vee^2\C^n$, where we recall that $\vee^2\C^n$ is the symmetric subspace of $(\C^n)^{\ot 2}$.  Since $P$ projects onto the vectors invariant under the 4-element group generated by $P_n((1,3))$ and $P_n((2,4))$, we can write it as
\be P = \frac{I + P_n((1,3))}{2} \cdot \frac{I + P_n((2,4))}{2}.\label{eq:P-sym-def}\ee
An alternate definition of $P$ is due to Wick's theorem:
\be P = \E_{a,b} [aa^* \ot bb^* \hat{\ot}\, aa^* \ot bb^*],
\label{eq:prod-test-def}\ee
where the expectation is taken over complex-Gaussian-distributed vectors $a,b\in \C^n$ normalized so that $\E \|a\|_2^2 = \E \|b\|_2^2 = \ifexp{1}\ifcount{n}/\sqrt{2}$.   Here we use the notation $\hat\ot$ to mark the separation between systems that we will use to define the separable states $\Sep(n^2,n^2)$.  We could equivalently write $P = \E_{a,b} [(aa^* \ot bb^*)^{\hat\ot 2}]$.
We will find that \eq{P-sym-def} is more useful for doing calculations, while \eq{prod-test-def} is helpful for converting $M_0$ into a form that resembles $A_{2,2}$ for some matrix $A$.

Define $M_1 = (\sqrt{M_0}\, \hat\ot \sqrt{M_0}) P\, (\sqrt{M_0}\, \hat\ot \sqrt{M_0})$, where $\sqrt{M_0}$ is taken to be the unique positive-semidefinite square root of $M_0$.  Observe that
\be M_1 = \E_{a,b} [v_{a,b} v_{a,b}^* \hat\ot\, v_{a,b} v_{a,b}^*]
= \E_{a,b} [V_{a,b}^{\hat\ot 2}],
\label{eq:M1-def}\ee
where 
we define $v_{a,b} := \sqrt{M_0}(a \ot b)$ and $V_{a,b} := v_{a,b}v_{a,b}^*$.  We claim that $\bh_{\Sep}(M_1)$ gives a reasonable proxy for $\bh_{\Sep}(M_0)$ in the following sense.
\begin{lemma}\label{lem:M1-like-M0}
\be \bh_{\Sep(n^2,n^2)}(M_1)  \begin{cases}
=1 & \text{in case Y} \\
\leq 1-\delta/2 & \text{in case N}.
\end{cases}
\label{eq:M1-claim}\ee
\end{lemma}
The proof of \lemref{M1-like-M0} is deferred to the end of this section.  The analysis is very similar to Theorem 13 of \cite{HarrowM10}, but the analysis here is much simpler because $M_0$ acts on only two systems.  However, it is strictly speaking not a consequence of the results in  \cite{HarrowM10}, because that paper considered a slightly different choice of $M_1$.

The advantage of replacing $M_0$ with $M_1$ is that (thanks to \eq{prod-test-def}) we now have a matrix with the same form as $A_{2,2}$ in \eq{A22-def}, allowing us to make use of \lemref{equiv}.  However, we first need to amplify the separation between cases Y and N.  This is achieved by the matrix $M_2 := M_1^{\ot k}$.  This tensor product is {\em not} across the cut we use to define separable states; in other words:
\be M_2 = \E_{\substack{a_1,\ldots,a_k\\b_1,\ldots,b_k}}
[(V_{a_1,b_1} \ot \cdots\ot V_{a_k,b_k})^{\hat \ot 2}].\ee
Now Lemma 12 from \cite{HarrowM10} implies that $\bh_{\Sep(n^{2k},n^{2k})}(M_2) = \bh_{\Sep(n^2,n^2)}(M_1)^k$.  This is either 1 or $\leq (3/4)^k$, depending on whether we have case Y or N.

Finally, we would like to relate this to the $2\ra 4$ norm of a matrix.  It will be more convenient to work with $M_1$, and then take tensor powers of the corresponding matrix.
Naively applying \lemref{equiv} would relate $\bh_{\Sep}(M_1)$ to $\tfcnorm{A}$ for an infinite-dimensional $A$.  Instead, we first replace the
continuous distribution on $a$ (resp. $b$) with a
finitely-supported distribution in a way that does not change
$\E_a aa^* \ot aa^*$ (resp. $\E_b bb^* \ot bb^*$).
Such distributions are called complex-projective
(2,2)-designs or quantum (state) 2-designs, and can be constructed
from spherical 4-designs on $\R^{2n}$~\cite{AmbainisE07}.  Finding these designs is challenging when each vector needs to have the same weight, but for our purposes we can use Carath\'eodory's theorem to show that there exist vectors $z_1,\ldots,z_m$  with $m = n^2$ 
such that
\ba \E_a [aa^* \ot aa^*] = \sum_{i\in[m]} z_i z_i^* \ot z_i z_i^*.
\label{eq:2-design}\ea
In what follows, assume that the average over $a,b$ used in the definitions of $P,M_1,M_2$ is replaced by the sum over $z_1,\ldots,z_m$.  By \eq{2-design} this change does not affect the values of $P, M_1, M_2$.

For $i,j\in [m]$, define $w_{i,j} := \sqrt{M_0}(z_i \ot z_j)$, and let $e_{i,j} := e_i \ot e_j$.
 Now we can apply \lemref{equiv} to find that $\bh_{\Sep}(M_1)=\tfcnorm{A_1}^4$, where
$$A_1 = \ifexp{\frac{1}{n}}\sum_{i,j\in [m]} e_{i,j} w_{i,j}^*.$$
The amplified matrix $M_2$ similarly satisfies $\bh_{\Sep(n^{2k}, n^{2k})}(M_2) = \tfcnorm{A_2}^4$, where
$$A_2 := A_1^{\ot k}
= \ifexp{\frac{1}{n^k}}
\sum_{i_1,\ldots,i_k,j_1,\ldots,j_k\in [m]} (e_{i_1,j_1} \ot e_{i_k, j_k}) ( w_{i_1.j_1} \ot \cdots \ot w_{i_k,j_k})^*.$$

The last step is to relate the complex matrix $A_2$ to a real matrix $A_3$ with the same $2\ra 4$ norm once we restrict to real inputs.  This can be achieved by replacing a single complex entry $\alpha+i\beta$ with the $6\times 2$ real matrix
\[
\frac{1}{\sqrt{2}}
\begin{pmatrix}
1 & 1 \\
1 & -1 \\
2^{1/4} & 0 \\
2^{1/4} & 0 \\
0 & 2^{1/4} \\
0 & 2^{1/4}
\end{pmatrix}
\cdot \begin{pmatrix}
\alpha & -\beta \\ \beta & \alpha
\end{pmatrix}
\]
A complex input $x+iy$ is represented by the column vector $\begin{pmatrix} x \\ y\end{pmatrix}$.  The initial $2\times 2$ matrix maps this to the real representation of $(\alpha + i \beta)(x+iy)$, and then the fixed $6\times 2$ matrix maps this to a vector whose 4-norm equals $|(\alpha + i \beta)(x+iy)|^4$.

\end{proof}

We conclude with the proof of \lemref{M1-like-M0}, mostly following \cite{HarrowM10}.
\begin{proof}
Case Y is simplest, and also provides intuition for the choices of the $M_1$ construction.  Since the extreme points of $\Sep(n,n)$ are of the form $xx^*\ot yy^*$ for $x,y\in \bS(\C^n)$, it follows that there exists $x,y \in \bS(\C^n)$ with $\ciprod{x\ot y, M(x\ot y)}=1$.  Since $M\leq I$, this implies that $M(x\ot y) = (x\ot y)$.  Thus $\sqrt{M_0}(x\ot y)=(x\ot y)$.  Let
\[ z = x \ot y \ot x \ot y. \]
Then $z$ is an eigenvector of both $\sqrt{M_0} \ot \sqrt{M_0}$ and $P$, with eigenvalue $1$ in each case.  To see this for $P$, we use the definition in \eq{P-sym-def}.  Thus $\ciprod{z, M_1z}=1$, and it follows that $\bh_{\Sep(n^2,n^2)}(M_1)\geq 1$.  On the other hand, $M_1\leq I$, implying that $\bh_{\Sep(n^2,n^2)}(M_1)\leq 1$.   This establishes case Y.

For case N, we assume that $\bh_{\Sep(n,n)}(M_0) \leq 1-\delta$ for
any $x,y\in \bS(\C^n)$.  The idea of the proof is that for any $x,y\in
\bS(\C^{n^2})$, we must either have $x,y$ close to a product state, in
which case the $\sqrt{M_0}$ step will shrink the vector, or if they
are far from a product state and preserved by $\sqrt{M_0} \ot
\sqrt{M_0}$, then the $P$ step will shrink the vector.  In either
case, the length will be reduced by a dimension-independent factor.

We now spell this argument out in detail.  Choose $x,y\in \bS(\C^{n^2})$ to achieve
\be s:= \ciprod{x\ot y, M_1(x \ot y)} = \bh_{\Sep(n^2,n^2)}(M_1).\ee
Let $X,Y\in L(\C^n)$ be defined by
\be \sqrt{M_0}x =: \sum_{i,j\in [n]} X_{i,j} e_i \ot e_j
\qquad\text{and}\qquad
 \sqrt{M_0}y =: \sum_{i,j\in [n]} Y_{i,j} e_i \ot e_j \ee
Note that $\ciprod{X,X} = \ciprod{x, M_0 x}\leq 1$ and similarly for $\ciprod{Y,Y}$.
We wish to estimate
\be s = \sum_{i,j,k,l,i',j',k',l'\in [n]}
\bar{X}_{i',j'}\bar{Y}_{k',l'}X_{i,j}Y_{k,l}
\ciprod{e_{i'}\ot e_{j'} \ot e_{k'} \ot e_{l'}, P (e_i \ot e_j  \ot e_k \ot e_l)}\ee
Using \eq{P-sym-def} we see that the expression inside the $\ciprod{\cdot}$ is
\be \ifexp{\frac{1}{n^4}}
\frac{ \delta_{i,i'}\delta_{j,j'}\delta_{k,k'}\delta_{l,l'} +
\delta_{i,k'}\delta_{j,j'}\delta_{k,i'}\delta_{l,l'} +
\delta_{i,i'}\delta_{j,l'}\delta_{k,k'}\delta_{l,j'} +
\delta_{i,k'}\delta_{j,l'}\delta_{k,i'}\delta_{l,j'}
}{4}\ee

Rearranging, we find
\be s=\frac{\ciprod{X,X}\ciprod{Y,Y} + \ciprod{X,Y}\ciprod{X,Y} + \ciprod{YY^*,XX^*} + \ciprod{Y^*Y,X^*X}}{4} .\ee
Using the AM-GM inequality we see that the maximum of this expression is achieved when $X=Y$, in which case we have
\be s = \frac{\ciprod{X,X}^2 + \ciprod{X^*X, X^*X}}{2}
\leq \frac{1 + \ciprod{X^*X, X^*X}}{2}.\ee
Let the singular values of $X$ be $\sigma_1 \geq \cdots \geq \sigma_n$.  Observe that $\cnorm{\sigma}_2^2 = \ciprod{X,X}\leq 1$, and thus $\cnorm{\sigma}_4^4 = \ciprod{X^*X,X^*X} \leq \sigma_1^2$.
On the other hand,
\begin{align}\sigma_1^2 &= \max_{a,b\in\bS(\C^n)} |\ciprod{a, Xb}|^2 \\
& = \max_{a,b\in\bS(\C^n)} |\ciprod{a\ot b, \sqrt{M_0}x}|^2 \\
& = \max_{a,b\in\bS(\C^n)} |\ciprod{\sqrt{M_0}(a\ot b), x}|^2 \\
& = \max_{a,b\in\bS(\C^n)} \ciprod{\sqrt{M_0}(a\ot b), \sqrt{M_0}(a\ot b)} \\
& = \max_{a,b\in\bS(\C^n)} \ciprod{a\ot b, M_0(a\ot b)} \\
& = \bh_{\Sep(n,n)}(M_0) \leq 1-\delta\end{align}
\end{proof}

{\bf Remark:} It is possible to extend \lemref{approx-preserving} to the situation when case Y has $\bh_{Sep}(M)>1-\delta'$ for some constant $\delta'<\delta$.  Since the details are somewhat tedious, and repeat arguments in \cite{HarrowM10}, we omit them here.

\subsubsection{Hardness of approximation for projectors}
Can \thmref{hardness} give any super-polynomial lower bound for the
SSE problem if we assume the Exponential-Time Hypothesis for 3-SAT?
To resolve this question using our techniques, we would like to reduce
3-SAT to estimating the $2\ra 4$ norm of the projector onto the
eigenvectors of a graph that have large eigenvalue.  We do not know
how to do this.  However, instead, we show that the matrix $A$
constructed in \thmref{hardness} can be taken to be a projector.  This
is almost WLOG, except that the resulting $2\ra 4$ norm will be at
least $3^{1/4}$.

\begin{lemma}\label{lem:almost-proj} Let $A$ be a linear map from $\R^k$ to $\R^n$ and $0<c<C$ , $\e>0$ some numbers. Then there is $m={O}(n^2/\eps^2)$ and a map $A'$ from $\R^k$ to $\R^m$ such that $\sigma_{\min}(A') \geq 1-\e$ and \textbf{(i)} if $\tfnorm{A} \leq c$ then $\tfnorm{A'} \leq 3^{1/4}+\e$, \textbf{(ii)} $\tfnorm{A} \geq C$ then $\tfnorm{A'} \geq \Omega(\e C/c)$.
\end{lemma}
\begin{proof} We let $B$ be a random map from $\R^k$ to
  $\R^{{O(n^2/\delta^2)}}$ with entries that are
  i.i.d. Gaussians with mean zero and variance $1/\sqrt k$.  Then Dvoretzky's
 theorem~\cite{Pisier99} implies that for every $f\in \R^k$, $\norm{Bf}_4
  \in 3^{1/4}(1\pm \delta)\norm{f}_2$.
 Consider the operator $A' = \begin{pmatrix}A/c \\ B \end{pmatrix}$ that
 maps $f$ into the concatenation of $Af$ and $Bf$. Moreover we take
 multiple copies of each coordinate so that the measure of output
 coordinates of $A'$ corresponding to $A$ is $\delta$,
 while the measure of coordinates corresponding to $B$ is $1-\delta$.

Now for every function $f$, we get that $\norm{A'f}_4^4  =
\frac{\delta}{c^4} \norm{Af}_4^4 + (1-\delta)\norm{Bf}_4^4$. In
particular, since $\norm{Bf}_4^4 \in 3(1\pm \delta)\norm{f}_2^4$, we
get that if $f$ is a unit vector and $\norm{Af}_4^4 \leq c^4$ then
$\norm{A'f}_4^4 \leq  \delta + 3(1 + \delta)$, while if
$\norm{Af}_4^4 \geq C^4$, then $\norm{A'f}_4^4 \geq \delta (C/c)^4$. 

Also note that the random operator $B$ will satisfy that for every function $f$, $\norm{Bf}_2 \geq (1-\delta)\norm{f}_2$, and hence $\norm{A'f} \geq (1-\delta)^2\norm{f}$. Choosing $\delta=\eps/2$ concludes the proof.
\end{proof}

It turns out that for the purposes of hardness of good approximation,
the case that $A$ is a projector is almost without loss of generality.

\begin{lemma}\label{lem:exact-proj} Suppose that for some $\e >0, C >
  1+\eps$ there is a $\poly(n)$ algorithm that on input a subspace
  $V\subseteq \R^n$ can distinguish between the case \textbf{(Y)}
  $\tfnorm{\Pi_V} \geq C$ and the case \textbf{(N)} $\tfnorm{\Pi_V} \leq
  3^{1/4}+\e$, where $\Pi_V$ denotes the projector onto $V$. Then there is $\delta = \Omega(\e)$ and a $\poly(n)$
  algorithm that on input an operator $A:\R^k\to\R^n$ with
  $\sigma_{\min}(A) \geq 1-\delta$ can distinguish between the case
 \textbf{(Y)} $\tfnorm{A} \geq C(3^{1/4}+\delta)$ and \textbf{(N)}
 $\tfnorm{A} \leq 3^{1/4}+\delta$. 
\end{lemma}
\begin{proof} First we can assume without loss of generality that $\norm{A}_{2\to 2} = \sigma_{\max}(A) \leq 3^{1/4}+\delta$, since otherwise we could rule out case \textbf{(N)}. Now we let $V$ be the image of $A$. In the case \textbf{(N)} we get that that for every $f\in \R^k$
\[
\norm{Af}_4 \leq (3^{1/4}+\delta)\norm{f}_2 \leq (3^{1/4}
+\delta)\norm{Af}_2/\sigma_{\min}(A) \leq (3^{1/4}
+O(\delta))\norm{Af}_2 \;, 
\]
implying $\|\Pi_V\|_{2\ra 4}\leq 3^{1/4} + O(\delta)$.
In the case \textbf{(Y)} we get that there is some $f$ such that
$\norm{Af}_4 \geq C(3^{1/4}+\delta)\norm{f}_2$, but since $\norm{Af}_2 \leq
\sigma_{\max}(A)\norm{f}_2 \leq (3^{1/4} + \delta)\norm{f}_2$, we get
that $\norm{Af}_4 \geq C \norm{Af}_2$,
implying $\|\Pi_V\|_{2\ra 4}\geq C$.
\end{proof}

Together these two lemmas effectively extend \thmref{hardness} to the case when $A$ is a projector.
We focus on the hardness of approximating to within a constant factor.
\begin{corollary}\label{cor:proj-hard}
For any $\ell,\eps>0$, if $\phi$ is a 3-SAT instance with $n$ variables and $O(n)$ clauses, then determining satisfiability of $\phi$ can be reduced to distinguishing between the cases $\tfnorm{A}\leq 3^{1/4}+\eps$ and $\tfnorm{A} \geq \ell)$, where $A$ is a projector acting on $m = \exp(\sqrt{n}\poly\log(n)\log(\ell/\eps))$ dimensions.
\end{corollary}
\begin{proof}
Start as in the proof of \thmref{hardness}, but in the application of  \lemref{approx-preserving}, take $k= O(\log(\ell/\eps))$.  This will allow us to take $C/c = \Omega(\ell/\eps)$ in \lemref{almost-proj}.  Translating into a projector with \lemref{exact-proj}, we obtain the desired result.
\end{proof}

\subsection{Algorithmic applications of equivalent formulations}\label{sec:BCY-app}

In this section we discuss the positive algorithmic results that come from the equivalences in \secref{inj-equiv}.
Since entanglement plays such a central role in quantum mechanics, the set $\Sep^2(\C^n)$ has been extensively studied.  However, because its hardness has long been informally recognized (and more recently has been explicitly established~\cite{Gurvits03,Liu07,HarrowM10,GallNN11}), various relaxations have been proposed for the set.  These relaxations are generally efficiently computable, but also have limited accuracy; see \cite{BeigiS10} for a review.

Two of the most important relaxations are the PPT condition and $k$-extendability.   For an operator $X\in L((\C^n)^{\ot r})$ and a set $S\subseteq [r]$, define the {\em partial transpose} $X^{T_S}$ to be the result of applying the transpose map to the systems $S$.  Formally, we define
\begin{align*} (X_1 \ot \cdots \ot X_r)^{T_S}
&:= \bigotimes_{k=1}^r f_k(X_k) \\
f_k(M) &:= \begin{cases} M & \text{ if }k\not\in S \\
M^T & \text{ if }k\in S
\end{cases}\end{align*}
and extend $T_S$ linearly to all of $L((\C^n)^{\ot r})$.
One can verify that if $X\in \Sep^r(\C^n)$ then $X^{T_S}\succeq 0$ for all $S\subseteq [r]$.  In this case we say that $X$ is PPT, meaning that it has Positive Partial Transposes.  However, the converse is not always true.  If $n>2$ or $r>2$, then there are states which are PPT but not in $\Sep$~\cite{HorodeckiHH96}.

The second important relaxation of $\Sep$ is called $r$-extendability.  To define this, we need to introduce the partial trace.  For $S\subseteq [r]$, we define $\tr_S$ to be the map from $L((\C^n)^{\ot r})$ to
$L((\C^n)^{\ot r-|S|})$ that results from applying $\tr$ to the systems in $S$. Formally
$$\tr_S \bigotimes_{k=1}^r X_k =
\prod_{k\in S} \tr X_k \, \bigotimes_{k\not\in S} X_k,$$
and $\tr_S$ extends by linearity to all of $L((\C^n)^{\ot r})$.

To obtain our relaxation of $\Sep$, we say that $\rho\in \cD(\C^n \ot \C^n)$ is $r$-extendable if there exists a {\em symmetric extension}
$\sigma\in \cD(\C^n \ot \vee^r\C^n)$ such that $\tr_{\{3,\ldots,r+1\}}\sigma = \rho$.
Observe that if $\rho\in \Sep^2(\C^n)$, then we can write $\rho = \sum_i x_i x_i^* \ot y_iy_i^*$, and so $\sigma = \sum_i x_i x_i^* \ot (y_iy_i^*)^{\ot r}$ is a valid symmetric extension.  Thus the set of $k$-extendable states contains the set of separable states, but again the inclusion is strict.  Indeed, increasing $k$ gives an infinite hierarchy of strictly tighter approximations of $\Sep^2(\C^n)$.  This hierarchy ultimately converges~\cite{DohertyPS04}, although not always at a useful rate (see Example IV.1 of \cite{ChristandlKMR07}).  Interestingly this relaxation is known to completely fail as a method of approximating $\Sep^2(\R^n)$~\cite{CavesFS02}, but our \lemref{equiv} is evidence that those difficulties do not arise in the $2\!\ra\! 4$-norm problem.

These two relaxations can be combined to optimize over symmetric extensions that have positive partial transposes~\cite{DohertyPS04}.  Call this the {\em level-$r$ DPS relaxation}.  It is known to converge in some cases more rapidly than $r$-extendability alone~\cite{NavascuesOP09}, but also is never exact for any finite $r$~\cite{DohertyPS04}.  Like SoS, this relaxation is an SDP with size $n^{O(r)}$.    In fact, for the case of the $2\ra 4$ norm, the relaxations are equivalent.
\Anote{It is possible that a more elegant proof of the following lemma can be obtained somehow from the DPS paper.}

\begin{lemma}\label{lem:equiv-DPS}
When the level-$r$ DPS relaxation is applied to $A_{2,2}$, the resulting approximation is equivalent to $\tensorsdp^{(2r+2)}$
\end{lemma}

\begin{proof}
Suppose we are given an optimal solution to the level-$r$ DPS relaxation.  This can be thought of as a density operator $\sigma\in \cD(\C^n \ot \vee^r\C^n)$ whose objective value is $\lambda:= \ciprod{A_{2,2}, \tr_{\{3,\ldots,r+1\}}\sigma}  = \ciprod{A_{2,2}\ot I_n^{\ot r-1}, \sigma}$.  Let $\Pi_{\rm sym}^{(2)} := (I + P_n((1,2)))/2$ be the orthogonal projector onto $\vee^2\C^n$.  Then $A_{2,2} = \Pi_{\rm sym}^{(2)} A_{2,2} \Pi_{\rm sym}^{(2)}$.  Thus, we can replace $\sigma$ by $\sigma':= (\Pi_{\rm sym}^{(2)} \ot I_n^{\ot r-1}) \sigma (\Pi_{\rm sym}^{(2)} \ot I_n^{\ot r-1})$ without changing the objective function.  However,  unless $\sigma'=\sigma$, we will have $\tr\sigma'<1$.  In this case, either $\sigma'=0$ and $\lambda=0$, or $\sigma'/\tr\sigma'$ is a solution of the DPS relaxation with a higher objective value.  In either case, this contradicts the assumption that $\lambda$ is the optimal value.  Thus, we must have $\sigma=\sigma'$, and in particular $\supp\sigma \subseteq \vee^2\C^n \ot (\C^n)^{\ot r-1}$.  Since we had $\supp\sigma \subseteq \C^n \ot\vee^r\C^n$ by assumption, it follows that
\[\supp\sigma \subseteq (\vee^2\C^n \ot (\C^n)^{\ot r-1})\cap(\C^n \ot\vee^r\C^n)
= \vee^{r+1}\C^n\]

Observe next that $\sigma^T$ is also a valid and optimal solution to the DPS relaxation, and so $\sigma' = (\sigma + \sigma^T)/2$ is as well.  Since $\sigma'$ is both symmetric and Hermitian, it must be a real matrix.  Replacing $\sigma$ with $\sigma'$, we see that we can assume WLOG that $\sigma$ is real.

Similarly, the PPT condition implies that $\sigma^{T_A}\geq 0$. (Recall that the first system is $A$ and the rest are $B_1,\ldots,B_r$.)
 Since the partial transpose doesn't change the objective function, $\sigma' = (\sigma + \sigma^{T_A})/2$ is also an optimal solution.  Replacing $\sigma$ with $\sigma'$, we see that we can assume WLOG that $\sigma = \sigma^{T_A}$.  Let $\vec\sigma\in(\R^n)^{\ot 2r+2}$ denote the flattening of $\sigma$; i.e. $\ciprod{x\ot y,\vec\sigma} = \ciprod{x,\sigma y}$ for all $x,y\in(\R^n)^{r+1}$.  Then the fact that $\sigma = \sigma^{T_A}$ means that $\vec\sigma$ is invariant under the action of $P_n((1,r+1))$.  Similarly, the fact that $\supp\sigma\subseteq\vee^{r+1}\R^n$ implies that $\vec\sigma\in\vee^{r+1}\R^n\ot\vee^{r+1}\R^n$.  Combining these two facts we find that $\vec\sigma\in\vee^{2r+2}\R^n$.

 Now that $\vec\sigma$ is fully symmetric under exchange of all $2r+2$ indices, we can interpret it as a real-valued pseudo-expectation $\pE_\sigma$ for polynomials of degree $2r+2$.  More precisely, we can define the linear map $\coeff$ that sends homogeneous degree-$2r+2$ polynomials to $\vee^{2r+2}\R^n$ by its action on monomials:
 \be
\coeff( f_1^{\alpha_1}\cdots f_n^{\alpha_n}) :=
\Pi_{\rm sym}^{(2r+2)}(e_1^{\ot \alpha_1}\ot\cdots\ot e_n^{\ot \alpha_n}),
\label{eq:first-coeff-def}\ee
where $\Pi_{\rm sym}^{(2r+2)} := \frac{1}{2r+2!}\sum_{\pi\in\cS_{2r+2}}P_n(\pi)$.
For a homogenous polynomial $Q(f)$ of even degree $2r'\leq 2r+2$ we define $\coeff$ by
\[ \coeff(Q(f)) := \coeff(Q(f) \cdot \cnorm{f}_2^{2r+2-2r'}). \]
For a homogenous polynomial $Q(f)$ of odd degree, we set $\coeff(Q):=0$.  Then we can extend $\coeff$ by linearity to all polynomials of degree $\leq 2r+2$.  Now define
\[\pE_\sigma[Q] := \ciprod{\coeff(Q), \vec\sigma}.\]

We claim that this is a valid pseudo-expectation.  For normalization, observe that $\pE[1] = \ciprod{\coeff(\cnorm{f}_2^{2r+2}),\vec\sigma} = \tr\sigma = 1$.  Similarly, the \tensorsdp constraint of $\pE[(\cnorm{f}_2^2-1)^2]=0$ is satisfied by our definition of $\coeff$.  Linearity follows from the linearity of $\coeff$ and the inner product.
  For positivity, consider a polynomial $Q(f)$ of degree $\leq r+1$.
  Write $Q=Q_o + Q_e$, where $Q_o$ collects all monomials of odd
  degree and $Q_e$ collects all monomials of even degree
  (i.e. $Q_e,Q_o = (Q(f) \pm Q(-f))/2$).
  Then $\pE[Q^2] = \pE[Q_o^2]  + \pE[Q_e^2]$, using the property that
  the pseudo-expectation of a monomial of odd degree is zero.

Consider first $\pE[Q_e^2]$.  Let $r' = 2\lfloor \frac{r+1}{2}\rfloor$ (i.e. $r'$ is $r+1$ rounded down to the nearest even number), so that
$ Q_e = \sum_{i=0}^{r'/2} Q_{2i},$ where $Q_{2i}$ is homogenous of degree $2i$.  Define
$ Q_e' := \sum_{i=0}^{r'/2} Q_{2i} \|f\|_2^{r'-2i}$.
Observe that $Q_e'$ is homogenous of degree $r'\leq r+1$, and that $\pE[Q_e^2] = \pE[(Q_e')^2]$.  Next, define $\coeff'$ to map homogenous polynomials of degree $r'$ into $\vee^{r'}\R^n$ by replacing $2r+2$ in \eq{first-coeff-def} with $r'$.
If $r'=r+1$ then define $\sigma'=\sigma$, or if $r'=r$ then define $\sigma' = \tr_{A}\sigma$.  Thus $\sigma' $ acts on $r'$ systems. Define $\vec\sigma' \in \vee^{2r'}\R^n$ to be the flattened version of $\sigma'$.
Finally we can calculate
\[ \pE[Q_e^2] = \pE[(Q_e')^2] = \ciprod{\coeff'(Q_e')\ot\coeff'(Q_e'), \vec\sigma'} =
\ciprod{\coeff'{Q_e'},\sigma' \coeff'{Q_e'}} \geq 0.\]
A similar argument establishes that $\pE[Q_o^2]\geq 0$ as well.
This establishes that any optimal solution to the DPS relaxation translates into a solution of the \tensorsdp relaxation.

To translate a \tensorsdp solution into a DPS solution, we run this construction in reverse.  The arguments are essentially the same, except that we no longer need to establish symmetry across all $2r+2$ indices.
\end{proof}

\subsubsection{Approximation guarantees and the proof of \thmref{BCY}}
\label{sec:BCY-proof}
Many approximation guarantees for the $k$-extendable relaxation (with
or without the additional PPT constraints) required that $k$ be $\poly(n)$, and thus do not lead to useful algorithms.  Recently, \cite{BrandaoCY11} showed that in some cases it sufficed to take $k=O(\log n)$, leading to quasi-polynomial algorithms.  It is far from obvious that their proof translates into our sum-of-squares framework, but nevertheless \lemref{equiv-DPS} implies that \tensorsdp can take advantage of their analysis.

To apply the algorithm of \cite{BrandaoCY11}, we need to upper-bound
$A_{2,2}$ by an 1-LOCC measurement operator.  That is, a quantum
measurement that can be implemented by one-way Local Operations and
Classical Communication (LOCC).  Such a measurement should have a
decomposition of the form $\sum_i V_i \ot W_i$ where each $V_i,W_i
\succeq 0$, $\sum_i V_i \preceq I_n$ and each $W_i \preceq I_n$.
Thus, for complex vectors $v_1,\ldots,v_m, w_1,\ldots,w_m$ satisfying $\sum_i v_i v_i^* \preceq I_n$ and $\forall i,\, w_iw_i^* \preceq I_n$, the operator $\sum_i v_iv_i^* \ot w_i w_i^*$ is a 1-LOCC measurement.

To upper-bound $A_{2,2}$ by a 1-LOCC measurement, we note that
 $a_i a_i^T \preceq \ifexp{n}\cnorm{a_i}_2^2 I_n$.   Thus, if we
 define $Z:=\ifexp{\frac{n^3}{m}}\cnorm{\sum_i a_i a_i^T}_{2\to 2}
 \max_i \cnorm{a_i}^2$, then $A_{2,2}/Z$ is a 1-LOCC measurement.  Note that this is a stricter requirement than merely requiring $A_{2,2}/Z \preceq  I_{n^2}$.  On the other hand, in some cases (e.g. $a_i$ all orthogonal), it may be too pessimistic.

In terms of the original matrix $A=\sum_i e_i a_i^T$, we have $\max_i \cnorm{a_i}_2=\ifexp{\frac{1}{n}}\cnorm{A}_{2\ra \infty}$.  Also $\cnorm{\sum_i a_i a_i^T}_{2\to 2} = \cnorm{A^TA}_{2\ra 2} = \ifexp{\frac{m}{n}}\cnorm{A}_{2\ra 2}^2$.  Thus
$$Z = \cnorm{A}_{2\to 2}^2 \cnorm{A}_{2\ra \infty}^2.$$
Recall from the introduction that $Z$ is an upper bound on $\cnorm{A}_{2\ra 4}^4$, based on the fact that $\cnorm{x}_4\leq \sqrt{\cnorm{x}_2\cnorm{x}_{\infty}}$ for any $x$.
(This bound also arises from using interpolation of norms~\cite{Stein56}.)

We can now apply the argument of \cite{BrandaoCY11} and show that optimizing over $O(r)$-extendable states will approximate $\cnorm{A}_{2\ra 4}^4$ up to additive error $\sqrt{\frac{\log(n)}{r}}Z$.  Equivalently, we can obtain additive error $\eps Z$ using $O(\log(n)/\eps^2)$-round \tensorsdp.  Whether the relaxation used is the DPS relaxation or our SoS-based \tensorsdp algorithm, the resulting runtime is $\exp(O(\log^2(n)/\eps^2))$.

\subsubsection{Gap instances}\label{sec:gap-DPS}
 Since \tensorsdp is equivalent than the DPS
relaxation for separable states, any gap instance for \tensorsdp would
translate into a gap instance for the DPS relaxation.  This would mean
the existence of a state that passes the $k$-extendability and PPT
test, but nevertheless is far from separable, with $A_{2,2}$ serving
as the entanglement witness demonstrating this.  While such states are
already known~\cite{DohertyPS04,BeigiS10}, it would be of interest to
find new such families of states, possibly with different scaling of
$r$ and $n$.

Our results, though, can be used to give an asymptotic separation of
the DPS hierarchy from the $r$-extendability hierarchy.  (As a
reminder, the DPS hierarchy demands that a state not only have an
extension to $r+1$ parties, but also that the extension be PPT across
any cut.)  To state this more precisely, we introduce some notation.
Define $\DPS_r$ to be the set of states $\rho^{AB}$ for which there
exists an extension $\tilde\rho^{AB_1\cdots B_r}$ with support in
$\C^n \ot \vee^r\C^n$ (i.e. a symmetric extension) such that $\tilde\rho$ is invariant under taking
the partial transpose of any system.  Define $\Ext_r$ to be the set of
states on $AB$ with symmetric extensions to $AB_1\ldots B_r$ but
without any requirement about the partial transpose.  Both
$\bh_{\DPS_r}$ and $\bh_{\Ext_r}$ can be computed in time $n^{O(r)}$,
although in practice $\bh_{\Ext_r}(M)$ is easier to work with, since it
only requires computing the top eigenvalue of $M\ot I_n^{\ot r-1}$
restricted to $\C^n \ot \vee^r\C^n$ and does not require solving an
SDP.

Many of the results about the convergence of $\DPS_r$ to $\Sep$ (such
as \cite{DohertyPS04, ChristandlKMR07, KoenigM09, BrandaoCY11}) use only the
fact that $\DPS_r \subset \Ext_r$.  A rare exception is
\cite{NavascuesOP09}, which shows that $\DPS_r$ is at least quadratically
closer to $\Sep$ than $\Ext_r$ is, in the regime where $r \gg n$.
Another simple example comes from $M=\Phi\Phi^*$, where $\Phi$ is the
maximally entangled state $n^{-1/2}\sum_{i=1}^n e_i \ot e_i$.  Then
one can readily compute that $\bh_{\Sep}(M) = \bh_{\DPS_1}(M) = 1/n$, 
while the $r$-extendible state
\be \tilde\rho^{AB_1\ldots B_r} = \frac{1}{r} \sum_{i=1}^r
(\Phi\Phi^*)^{AB_i} \ot \bigotimes_{j \in [r]\backslash\{i\}}
\left(\frac{I}{n}\right)^{B_j}
\label{eq:Phi-extension}\ee
achieves $\bh_{\Ext_r}(M) \geq 1/r$.  (In words, \eq{Phi-extension}
describes a state where $A$ and a randomly chosen $B_i$ share the
state $\Phi\Phi^*$, while the other $B_j$ systems are described by 
maximally mixed states.)
This proves that the $r$-extendable
hierarchy cannot achieve a good multiplicative approximation of
$\bh_{\Sep}(M)$ for all $M$ without taking $r\geq \Omega(n)$.

Can we improve this when $M$ is in a restricted class, such as 1-LOCC?  Here
\cite{BuhrmanRSW11} show that the Khot-Vishnoi integrality
construction can yield an $n^2$-dimensional $M$ for which $h_{\Sep}(M)
\leq O(1/n)$, but $\tr M \Phi \geq \Omega(1/\log^2(n))$.  Combined
with \eq{Phi-extension} this implies that $\bh_{\Ext_r}(M) \geq
\Omega(1/r\log^2(n))$.   On the other hand, \thmref{ug-main} and
\lemref{equiv-DPS} implies that $\bh_{\DPS_3}(M) \leq O(1/n)$.
Additionally, the $M$ from Ref.~\cite{BuhrmanRSW11} belongs to the
class BELL, a subset of 
1-LOCC, given by measurements of the form $\sum_{i,j} p_{i,j} A_i \ot
B_j$, with $0\leq p_{i,j}\leq 1$ and $\sum_i A_i = \sum_j B_j = I$.
As a result, we obtain the following corollary.
\begin{corollary}
There exists an $n^2$ dimensional $M\in\text{BELL}$ such that
$$\frac{h_{\Ext_r}(M)}{h_{\DPS_3}(M)} \leq O\left( \frac{r\log^2(n)}{n}\right)$$
\end{corollary}

\section{Subexponential algorithm for the 2-to-q norm} \label{sec:subexp}

In this section we prove Theorem~\ref{thm:subexp}:

\restatetheorem{thm:subexp}

and obtain as a corollary a subexponential algorithm for
\smallsetexpansion . The algorithm roughly matches the performance of
\cite{AroraBS10}'s for the same problem, and in fact is a very close
variant of it. The proof is obtained by simply noticing that a
subspace $V$ cannot have too large of a dimension without containing a
vector $v$ (that can be easily found) such that $\norm{v}_q \gg
\norm{v}_2$, while of course it is always possible to find such a vector (if it exists) in time exponential in $\dim(V)$. The key observation is the following basic fact (whose proof we include here for completeness):

\begin{lemma} \label{lem:two-to-infty}For every subspace $V \subseteq \R^n$, $\norm{V}_{2\to \infty} \geq \sqrt{\dim(V)}$.
\end{lemma}

\begin{proof} Let $f^1,\ldots,f^d$ be an orthonormal basis for $V$, where $d=\dim(V)$. For every $i\in [n]$, let $g^i$ be the function $\sum_{j=1}^d f^j_i f^i$. Note that the $i^{th}$ coordinate of $g^i$ is equal to $\sum_{j=1}^d (f^j_i)^2$ \textbf{(*)}  which also equals $\norm{g^i}_2^2$ since the $f^j$'s are an orthonormal basis. Also the expectation of \textbf{(*)} over $i$ is $\sum_{j=1}^d \E_{i\in [n]} (f^j_i)^2 = \sum_{j=1}^d \norm{f^j}_2^2 = d$ since these are unit vectors. Thus we get that $\E_i \norm{g^i}_{\infty} \geq \E_i g^i_i = d = \E_i \norm{g}_2^2$. We claim that one of the $g^i$'s must satisfy $\norm{g^i}_{\infty} \geq \sqrt{d}\norm{g^i}_2$. Indeed, suppose otherwise, then we'd get that
\[
d = \E_i \norm{g^i}_2^2 > E_i \norm{g^i}_{\infty}^2/d
\]
meaning $E_i \norm{g^i}_{\infty}^2 < d^2$, but  $E_i \norm{g^i}_{\infty}^2 \geq \left( \E_i \norm{g^i}_{\infty} \right)^2 = d^2$--- a contradiction.
\end{proof}

\begin{corollary} \label{cor:dim-bound} For every subspace $V \subseteq \R^n$, $\norm{V}_{2\to q} \geq \sqrt{\dim(V)}/n^{1/q}$
\end{corollary}
\begin{proof} By looking at the contribution to the $q^{th}$-norm of just one coordinate one can see that for every function $f$, $\norm{f}_q \geq (\norm{f}_{\infty}^q/n)^{1/q} = \norm{f}_{\infty}/n^{1/q}$.
\end{proof}

\paragraph{Proof of Theorem~\ref{thm:subexp} from
  Corollary~\ref{cor:dim-bound}}  Let $A:\R^m\to\R^n$ be an operator,
and let $1<c<C$ be some constants and $\sigma = \sigma_{\min}(A)$ be
such that $\norm{Af}_2 \geq \sigma \norm{f}_2$ for every $f$
orthogonal to the kernel of $A$. We want to distinguish between the
case that $\norm{A}_{2\to q} \leq c$ and the case that $\norm{A}_{2\to
  q} \geq C$. If $\sigma>c$ then clearly we are not in the first case,
and so we are done. Let $V$ be the image of $A$. If $\dim(V) \leq
C^2n^{2/q}$ then we can use brute force enumeration to find out if
such $v$ exists in the space. \Bnote{add a bit of detail} Otherwise,
by Corollary~\ref{cor:dim-bound} we must be in the second case. \qed

Note that by applying Theorem~\ref{thm:BCY} we can replace the brute
force enumeration step by the SoS hierarchy, since $\norm{V}_{2\to
  2}\leq 1$ automatically, and unless $\norm{V}_{2\to\infty} \leq
Cn^{1/q}$ we will be in the second case.

A corollary of Theorem~\ref{thm:subexp} is a subexponential algorithm for  \smallsetexpansion

\begin{corollary}
For every $0.4 > \nu > 0$ there is an $\exp( n^{1/ O(\log(1/\nu))})$ time algorithm that given a graph with the promise that either (i) $\Phi_G(\delta) \geq 1 - \nu$ or (ii) $\Phi_G(\delta^2) \leq 0.5$ decides which is the case.
\end{corollary}
\begin{proof}
For $q = O(\log(1/\nu))$ we find from Theorem~\ref{thm:sse-hyper} that in case (i), $\Vert V_{\geq 0.4} \Vert_{2 \to q} \leq 2/ \sqrt{\delta}$, while in case (ii) $\Vert V_{\geq 0.4} \Vert_{2 \to q} \geq 0.1/ \delta^{1 - 2/q}$.  Thus it sufficies to obtain a $(2/ \sqrt{\delta},  0.1/ \delta^{1 - 2/q})$-approximation for the $2 \to q$ norm to solve the problem, and by Theorem~\ref{thm:subexp} this can be achieved in time $\exp( n^{O(\log(1/\nu))})$ for sufficiently small $\delta$.
\end{proof}

\section*{Conclusions}


This work motivates further study of the complexity of approximating hypercontractive norms such as the $2\to 4$ norm. A particulary interesting question is
what is the complexity of obtaining a good approximation for the $2\to 4$ norm and what's the relation of this problem to the \smallsetexpansion problem. Our
work leaves possible at least the following three scenarios: \textbf{(i)} both these problems can be solved in quasipolynomial time, but not faster, which
would mean that the UGC as stated is essentially false but a weaker variant of it is true, \textbf{(ii)} both these problems are \np-hard to solve (via a
reduction with polynomial blowup) meaning that the UGC is true,  and \textbf{(iii)} the \smallsetexpansion and \uniquegames problems are significantly easier
than the $2\to 4$ problem with the most extreme case being that the former two problems can be solved in polynomial time and the latter is \np-hard and hence
cannot be done faster than subexponential time.  This last scenario would mean that one can improve on the subexponential algorithm for the $2\to 4$ norm for
general instances by using the structure of instances arising from the \smallsetexpansion reduction of Theorem~\ref{thm:sse-hyper} (which indeed seem quite
different from the instances arising from the hardness reduction of Theorem~\ref{thm:hardness}). In any case we hope that further study of the complexity of
computing hypercontractive norms can lead to a better understanding of the boundary between hardness and easiness for \uniquegames and related problems.

\paragraph{Acknowledgments} We thank Pablo Parrilo for useful
discussions, Daniel Dadush for catching an error in an earlier version
of \secref{random}
and the anonymous STOC referees for numerous comments that greatly improved the presentation of this paper. Aram Harrow was
funded by NSF grants 0916400, 0829937, 0803478, 1111382, DARPA QuEST
contract FA9550-09-1-0044 and ARO contract W911NF-12-1-0486. 
Jonathan Kelner was partially supported by NSF awards 1111109 and 0843915.

\appendix

\section{More facts about pseudo-expectation} \label{app:pseudo-expectation}

In this section we note some additional facts about pseudo-expectation functionals that are useful in this paper.

\begin{lemma}
  \label{lem:boundedness-relation}
  The relation $P^2\sle P$ holds if and only if $0\sle P\sle 1$.
  Furthermore, if $P^2 \sle P$ and $0 \sle Q \sle P$, then $Q^2\sle Q$.
\end{lemma}

\begin{proof}
  If $P\sge 0$, then $P\sle 1$ implies $P^2\sle P$.
  (Multiplying both sides with a sum of squares preserves the order.)
  On the other hand, suppose $P^2\sle P$.
  Since $P^2\sge 0$, we also have $P\sge 0$.
  Since $1-P=P-P^2+(1-P)^2$, the relation $P^2\sle P$ also implies $P\sle 1$.

  For the second part of the lemma, suppose $P^2\sle P$ and $0\sle Q\sle
  P$.
  Using the first part of the lemma, we have $P\sle 1$.
  It follows that $0\sle Q\sle 1$, which in turn implies $Q^2\sle Q$ (using
  the other direction of the first part of the lemma).
\end{proof}

\begin{fact}
  If $f$ is a $d$-f.r.v. over $\R^\cU$ and $\set{P_v}_{v\in
    \cU}$ are polynomials of degree at most $k$, then $g$ with
  $g(v)=P_v(f)$ is a level-$(d/k)$ fictitious random variable over $\R^{\cU}$.
  (For a polynomial $Q$ of degree at most $d/k$, the pseudo-expectation is
  defined as
  \begin{math}
    \pE_g Q(\set{g(v)}_{v\in \cU}) \seteq \pE_f Q(\set{P_v(f)}_{v\in \cU})\mper
  \end{math}%
  )
\end{fact}

\Dnote{is the fact correct? the level might drop to $d/k$, no?}
\Bnote{I think you're right, and changed this from $d-k$ to $d/k$ to be on the safe side. Are we actually using this fact anywhere?}
\Anote{Does \thmref{ug-value-bound} use it?}
\begin{lemma}\label{lem:linear-cs} For $f,g \in \L2(\cU)$,
\[
\iprod{f,g} \sle \tfrac{1}{2}\norm{f}^2 +\tfrac{1}{2}\norm{g}^2 \mper
\]
\end{lemma}
\begin{proof} The right-hand side minus the LHS equals the square polynomial $\tfrac{1}{2}\iprod{f-g,f-g}$
\end{proof}

\begin{lemma}[Cauchy-Schwarz inequality]\label{lem:cauchy-Schwarz}
  If $(f,g)$ is a level-$2$ fictitious random variable over $\R^\cU\times \R^{\cU}$,
  then
  \begin{displaymath}
    \pE_{f,g} \iprod{f,g}
    \le \sqrt{\pE_f \snorm{f}}\cdot \sqrt{\pE_g
      \snorm{g}}
    \mper
  \end{displaymath}
\end{lemma}

\begin{proof}
  Let $\bar f =f/\sqrt{\pE_f \snorm{f}}$ and $\bar g=g/\sqrt{\pE_g
    \snorm{g}}$.
  Note $\pE_{\bar f} \snorm {\bar f} = \pE_{\bar g}\snorm{\bar g}=1$.
  Since by Lemma~\ref{lem:linear-cs}, $\iprod{\bar f,\bar g}\sle \half \snorm{\bar f} + \half \snorm{\bar
    g}$, we can conclude the desired inequality,
  \begin{displaymath}
    \pE_{f,g} \iprod{f,g} = \sqrt{\pE_f \snorm{f}}\cdot \sqrt{\pE_g
      \snorm{g}} \pE_{\bar f,\bar g} \iprod{\bar f,\bar g}
    \le \sqrt{\pE_f \snorm{f}}\cdot \sqrt{\pE_g  \snorm{g}}
    \cdot \underbrace{\Paren{\tfrac12 \pE_{\bar f} \snorm{\bar f} + \tfrac
        12\E_{\bar g} \snorm{\bar g}}}_{=1}
    \mper\qedhere
  \end{displaymath}
\end{proof}

\begin{corollary}[\Holder's inequality]\label{lem:holders}
If $(f,g)$ is a \lrvar4 over $\R^\cU\times \R^{\cU}$, then
\begin{align*}
\pE_{f, g}  \E_{u \in \cU} f(u)g(u)^3 \leq \left(\pE_f \norm{f}_4^4\right)^{1/4} \left(\pE_g \norm{g}_4^4 \right)^{3/4} .
\end{align*}
\end{corollary}
\begin{proof}
Using \pref{lem:cauchy-Schwarz} twice, we have
\begin{align*}
\pE_{f, g}  \E_{u \in \cU} f(u)g(u)^3 \leq & \left(\pE_{f, g} \E_{u \in \cU} f(u)^2 g(u)^2\right)^{1/2} \left(\pE_{g} \norm{g}_4^4\right)^{1/2}
\leq \left(\pE_f \norm{f}_4^4\right)^{1/4} \left(\pE_g \norm{g}_4^4 \right)^{3/4} .
\end{align*}
\end{proof}

\section{Norm bound implies small-set expansion} \label{app:hyper-imp-sse}

In this section, we show that an upper bound on $2\to q$ norm of the projector to the top eigenspace of a graph implies that the graph is a small-set expander.
This proof appeared elsewhere implicitly~\cite{KhotV05,ODonnell07} or explicitly~\cite{BarakGHMRS11} and is presented here only for completeness.  We use the
same notation from Section~\ref{sec:sse}. Fix a graph $G$ (identified with its normalized adjacency matrix), and $\lambda \in (0,1)$, letting $V_{ \geq
\lambda}$ denote the subspace spanned by eigenfunctions with eigenvalue at least $\lambda$.

If $p,q$ satisfy $1/p+1/q=1$ then $\norm{x}_p = \max_{y: \norm{y}_q \leq 1} |\iprod{x,y}|$. Indeed, $|\iprod{x,y}| \leq \norm{x}_p\norm{y}_q$ by \Holder's
inequality, and by choosing $y_i = \sign(x_i)|x_i|^{p-1}$ and normalizing one can see this equality is tight. In particular, for every $x \in L(\cU)$,
$\norm{x}_q = \max_{y:\norm{y}_{q/(q - 1)}\leq 1} |\iprod{x,y}|$ and $\norm{y}_{q/(q - 1)} = \max_{\norm{x}_q \leq 1} |\iprod{x,y}|$. As a consequence
\[
\norm{A}_{2\to q} = \max_{\norm{x}_2 \leq 1} \norm{Ax}_q = \max_{\norm{x}_2 \leq 1, \norm{y}_{q/(q - 1)} \leq 1} |\iprod{Ax,y}| = \max_{\norm{y}_{q/(q - 1)}\leq 1} |\iprod{A^Ty,x}| = \norm{A^T}_{q/(q- 1)\to 2}
\]

Note that if $A$ is a projection operator, $A=A^T$. Thus, part~1 of Theorem~\ref{thm:sse-hyper} follows from the following lemma:

\begin{lemma}\label{lem:hyper-to-sse} Let $G=(V,E)$ be  regular graph and $\lambda \in (0,1)$. Then, for every $S \subseteq V$,
\[
\bd(S) \geq 1-\lambda - \norm{V_{\lambda}}_{q/(q-1) \to 2}^2\mu(S)^{(q-2)/q}
\]
\end{lemma}
\begin{proof} Let $f$ be the characteristic function of $S$, and write $f = f' + f''$ where $f'\in V_{\lambda}$ and  $f''=f-f'$ is the projection to the eigenvectors with value less than $\lambda$. Let $\mu = \mu(S)$. We know that
\begin{equation}
\bd(S) = 1 - \iprod{f,Gf}/\normt{f}^2 = 1 - \iprod{f,Gf}/\mu \;, \label{eq:expansioneval}
\end{equation}
And $\norm{f}_{q/(q-1)} = \left( \E f(x)^{q/(q-1)} \right)^{(q-1)/q} = \mu^{(q-1)/q}$, meaning that $\norm{f'} \leq \norm{V_{\lambda}}_{q/(q-1)\to 2}
\mu^{(q-1)/q}$. We now write
\begin{eqnarray}
\iprod{f,Gf} =  \iprod{f',Gf'} + \iprod{f'',Gf''} \leq \normt{f'}^2 + \lambda\normt{f''}^2 &\leq&  \norm{\cV}_{q/(q-1)\to 2}^2\norm{f}_{q/(q-1)}^2 + \lambda\mu \nonumber \\ &\leq& \norm{\cV}_{2\to q}^2 \mu^{2(q-1)/q} + \lambda \mu\mper
\end{eqnarray}
Plugging this into (\ref{eq:expansioneval}) yields the result.
\end{proof}

\section{Semidefinite Programming Hierarchies} \label{app:hierarchies}

\newcommand{\inst}{\Im}
\newcommand{\sos}{\mathrm{sos}}
\newcommand{\lass}{\mathrm{lass}}

In this section, we compare different SDP hierarchies and discuss some of
their properties.

\subsection{Example of Max Cut}

In this section, we compare the SoS hierarchy and Lasserre hierarchy at the example of Max Cut.
(We use a formulation of Lasserre's hierarchy similar to the one in
\cite{Schoenebeck08}.)
\Dnote{maybe rephrase for clarity. in this section we want to compare our
  formulation of Lasserre's hierarchy to the usual formulation of this
  hierarchy in the computer science literature.}
It will turn out that these different formulations are equivalent up to
(small) constant factors in the number of levels.
We remark that the same proof with syntactic modifications shows that our
SoS relaxation of Unique Games is equivalent to the corresponding Lasserre relaxation.

Let $G$ be a graph (an instance of Max Cut) with vertex set
$V=\set{1,\ldots,n}$.
The level-$d$ Lasserre relaxation for $G$, denoted $\lass_d(G)$, is the
following semidefinite program over vectors $\set{v_S}_{S\sse
  [n],~\card{S}\le d}$,
\begin{align*}
  \lass_d(G)\colon & \maximize& & \sum_{(i,j)\in G} \cnorm{v_i - v_j}^2
  \\
  & \text{ subject to }&& \ciprod{v_S,v_T}=\ciprod{v_{S'},v_{T'}}
  \quad\text{ for   all sets with } S\Delta T = S'\Delta T'\mcom\\
  & && \cnorm{v_\eset}^2=1\mper
\end{align*}

The level-$d$ SoS relaxation for $G$, denoted $\sos_d(G)$, is the following
semidefinite program over $d$-\pef $\pE$ (and $d$-\frv $x$ over $\R^V$),
\begin{align*}
  \sos_d(G)\colon & \maximize& &  \pE_x \sum_{(i,j)\in G} (x_i-x_j)^2\\
  & \text{ subject to }&& \pE_x (x_i^2-1)^2=0\quad\text{ for all }i\in V\mper\qquad\qquad
\end{align*}

\paragraph{From Lasserre to SoS}
\label{sec:from-lasserre-sos}

Suppose $\set{v_S}$ is a solution to $\lass_d(G)$.
For a polynomial $P$ over $\R^V$, we obtain a multilinear polynomial $P'$
by successively replacing squares $x_i^2$ by $1$.
(In other words, we reduce $P$ modulo the ideal generated by the
polynomials $x_i^2-1$ with $i\in V$.)
We define a $d$-\pef $\pE$ by setting $\pE P = \sum_{\card{S}\le
  d}c_S\iprod{v_\eset,v_S}$, where $\set{c_S}_{\card{S}\le d}$ are the
coefficients of the polynomial $P'=\sum_{\card{S}\le d} c_S \prod_{i\in
  S}x_i$ obtained by making $P$ multilinear.
The functional $\pE$ is linear (using $(P+Q)'=P'+Q'$) and satisfies the
normalization condition.
We also have $\pE (x_i^2-1)^2=0$ since $(x_i^2-1)^2=0$ modulo $x_i^2-1 $.
Since $\pE_x (x_i-x_j)^2 = \cnorm{v_i-v_j}^2$ for all $i,j\in V$ (using
$\iprod{v_\eset,v_{ij}}=\iprod{v_i,v_j}$), our solution for $\sos_d(G)$ has
the same objective value as our solution for $\lass_d(G)$.
It remains to verify positivity.
Let $P^2$ be a polynomial of degree at most $d$.
We may assume that $P$ is multilinear, so that $P=\sum_{\card{S}\le d} c_S x_S$
Therefore $P^2=\sum_{S,T} c_Sc_Tx_Sx_T$ and $\pE P^2=\sum_{S,T}c_Sc_T
\iprod{v_\eset,v_{S\Delta T}}$.
Using the property $\ciprod{v_\eset,v_{S\Delta T}}=\ciprod{v_S,v_T}$, we
conclude $\pE P^2 = \sum_{S,T}c_S c_T \iprod{v_S,v_T} = \cnorm{\sum_S c_S
  v_S}^2\ge 0$.

\paragraph{From SoS to Lasserre}
\label{sec:from-sos-lasserre}

Let $\pE$ be a solution to $\sos_d(G)$.
We will construct a solution for $\lass_{d/2}(G)$ (assuming $d$ is even).
Let $d'=d/2$.
For $\alpha\in \N^n$, let $x^\alpha$ be the monomial $\prod_{i\in[n]}
x_i^{\alpha_i}$.
The polynomials $\set{x^\alpha}_{\card{\alpha}\le d'}$ form a basis of the
space of degree-$d'$ polynomials over $\R^n$.
Since $\pE P^2\ge 0$ for all polynomials $P$ of degree at most
$d'$, the matrix $(\pE x^\alpha x^\beta)_{\card{\alpha},\card{\beta}\le
  d'}$ is positive semidefinite.
Hence, there exists vectors $v_\alpha$ for $\alpha$ with $\card{\alpha}\le
d'$ such that $\pE x^\alpha x^\beta=\iprod{v_\alpha,v_\beta}$.
We claim that the vectors $v_\alpha$ with $\alpha\in\bits^n$ and
$\card{\alpha}\le d$ form a solution for $\lass_d(G)$.
The main step is to show that $\iprod{v_\alpha,v_\beta}$ depends only on
$\alpha+\beta\mod 2$.
Since $\iprod{v_\alpha,v_\beta} = \pE x^{\alpha+\beta}$, it is enough to
show that $\pE$ satisfies $\pE x^\gamma = \pE x^{\gamma \mod 2}$.
Hence, we want to show $\pE x^2 P = \pE P$ for all polynomials (with
appropriate degree).
Indeed, by \pref{lem:pseudo-expectation-cauchy-schwarz}, $\pE (x^2-1)\cdot P\le
\sqrt{\pE (x^2-1)^2}\sqrt{\pE P^2}=0$.

\Dnote{

\subsection{Comparsion with hierarchies of Parrilo and Lasserre}

\subsection{Strong Duality}

}

\addreferencesection
\bibliographystyle{hyperalpha}
\bibliography{hypercontract,quantum}

\end{document}

%% file: macros.tex
\newcommand{\nc}{\newcommand}
\nc{\rnc}{\renewcommand}

\usepackage{etex}

\usepackage{xspace,enumerate}

\usepackage[dvipsnames]{xcolor}

\usepackage[full]{textcomp}

\usepackage[american]{babel}

\usepackage{mathtools}

\usepackage{bm}

\usepackage{amsthm}

\newtheorem{theorem}{Theorem}[section]
\newtheorem*{theorem*}{Theorem}

\newtheorem{claim}[theorem]{Claim}
\newtheorem*{claim*}{Claim}

\newtheorem*{proposition*}{Proposition}
\newtheorem{lemma}[theorem]{Lemma}
\newtheorem*{lemma*}{Lemma}
\newtheorem{corollary}[theorem]{Corollary}

\newtheorem*{conjecture*}{Conjecture}

\newtheorem{fact}[theorem]{Fact}
\newtheorem*{fact*}{Fact}

\newtheorem*{hypothesis*}{Hypothesis}

\theoremstyle{definition}
\newtheorem{definition}[theorem]{Definition}

\newtheorem{remark}[theorem]{Remark}

\usepackage[varg]{txfonts} 
\renewcommand{\mathbb}{\varmathbb}

\ifnum\showkeys=1
\usepackage[color]{showkeys}
\fi

\ifnum\showcolorlinks=1
\usepackage[
pagebackref,
colorlinks=true,
urlcolor=blue,
linkcolor=blue,
citecolor=OliveGreen,
]{hyperref}
\fi

\ifnum\showcolorlinks=0
\usepackage[
colorlinks=false,
pdfborder={0 0 0}
]{hyperref}
\fi

\usepackage{prettyref}

\newcommand{\savehyperref}[2]{\texorpdfstring{\hyperref[#1]{#2}}{#2}}

\newrefformat{eq}{\savehyperref{#1}{\textup{(\ref*{#1})}}}
\newrefformat{lem}{\savehyperref{#1}{Lemma~\ref*{#1}}}
\newrefformat{def}{\savehyperref{#1}{Definition~\ref*{#1}}}
\newrefformat{thm}{\savehyperref{#1}{Theorem~\ref*{#1}}}
\newrefformat{cor}{\savehyperref{#1}{Corollary~\ref*{#1}}}
\newrefformat{cha}{\savehyperref{#1}{Chapter~\ref*{#1}}}
\newrefformat{sec}{\savehyperref{#1}{Section~\ref*{#1}}}
\newrefformat{app}{\savehyperref{#1}{Appendix~\ref*{#1}}}
\newrefformat{tab}{\savehyperref{#1}{Table~\ref*{#1}}}
\newrefformat{fig}{\savehyperref{#1}{Figure~\ref*{#1}}}
\newrefformat{hyp}{\savehyperref{#1}{Hypothesis~\ref*{#1}}}
\newrefformat{alg}{\savehyperref{#1}{Algorithm~\ref*{#1}}}
\newrefformat{rem}{\savehyperref{#1}{Remark~\ref*{#1}}}
\newrefformat{item}{\savehyperref{#1}{Item~\ref*{#1}}}
\newrefformat{step}{\savehyperref{#1}{step~\ref*{#1}}}
\newrefformat{conj}{\savehyperref{#1}{Conjecture~\ref*{#1}}}
\newrefformat{fact}{\savehyperref{#1}{Fact~\ref*{#1}}}
\newrefformat{prop}{\savehyperref{#1}{Proposition~\ref*{#1}}}
\newrefformat{prob}{\savehyperref{#1}{Problem~\ref*{#1}}}
\newrefformat{claim}{\savehyperref{#1}{Claim~\ref*{#1}}}
\newrefformat{relax}{\savehyperref{#1}{Relaxation~\ref*{#1}}}
\newrefformat{red}{\savehyperref{#1}{Reduction~\ref*{#1}}}
\newrefformat{part}{\savehyperref{#1}{Part~\ref*{#1}}}

\newcommand{\Sref}[1]{\hyperref[#1]{\S\ref*{#1}}}

\usepackage{nicefrac}

\newcommand{\half}{\nicefrac12}

\ifnum\usemicrotype=1
\usepackage{microtype}
\fi

\ifnum\showauthornotes=1
\newcommand{\Authornote}[2]{{\sffamily\small\color{red}{[#1: #2]}}}
\newcommand{\Authorcomment}[2]{{\sffamily\small\color{gray}{[#1: #2]}}}
\newcommand{\Authorstartcomment}[1]{\sffamily\small\color{gray}[#1: }

\newcommand{\Authorfnote}[2]{\footnote{\color{red}{#1: #2}}}
\newcommand{\Authorfixme}[1]{\Authornote{#1}{\textbf{??}}}
\newcommand{\Authormarginmark}[1]{\marginpar{\textcolor{red}{\fbox{\Large #1:!}}}}
\else
\newcommand{\Authornote}[2]{}
\newcommand{\Authorcomment}[2]{}
\newcommand{\Authorstartcomment}[1]{}

\newcommand{\Authorfnote}[2]{}
\newcommand{\Authorfixme}[1]{}
\newcommand{\Authormarginmark}[1]{}
\fi

\newcommand{\Dnote}{\Authornote{D}}

\newcommand{\Anote}{\Authornote{A}}

\newcommand{\Bnote}{\Authornote{B}}

\ifnum\showfixme=0

\fi

\usepackage{boxedminipage}

\DeclareBoldMathCommand[heavy]\boldlangle{\left\langle}
\DeclareBoldMathCommand[bold]\boldrangle{\right\rangle}
\DeclareBoldMathCommand[bold]\boldlvert{\left\lVert}
\DeclareBoldMathCommand[heavy]\boldrvert{\right\rVert}

\newcommand{\paren}[1]{(#1)}
\newcommand{\Paren}[1]{\left(#1\right)}

\newcommand{\Abs}[1]{\left\lvert#1\right\rvert}

\newcommand{\card}[1]{\lvert#1\rvert}

\newcommand{\set}[1]{\{#1\}}

\newcommand{\norm}[1]{\lVert#1\rVert}
\newcommand{\cnorm}[1]{{\pmb{\lVert}} #1 {\pmb{\rVert}}}
\newcommand{\cNorm}[1]{\cnorm{#1}} 

\newcommand{\cBignorm}[1]{\cnorm{#1}}

\newcommand{\normt}[1]{\norm{#1}_2}

\newcommand{\snorm}[1]{\norm{#1}^2}

\newcommand{\iprod}[1]{\langle#1\rangle}
\newcommand{\ciprod}[1]{\pmb{\langle} #1\pmb{\rangle}}

\newcommand{\cIprod}[1]{\ciprod{#1}}
\newcommand{\cbigiprod}[1]{\ciprod{#1}}

\newcommand{\Esymb}{\mathbb{E}}
\newcommand{\Psymb}{\mathbb{P}}
\newcommand{\Vsymb}{\mathbb{V}}

\DeclareMathOperator*{\E}{\Esymb}
\DeclareMathOperator*{\Var}{\Vsymb}
\DeclareMathOperator*{\ProbOp}{\Psymb}

\renewcommand{\Pr}{\ProbOp}

\newcommand{\ot}{\otimes}

\newcommand{\textparen}[1]{\text{(#1)}}
\newcommand{\using}[1]{\textparen{using #1}}
\ifx\because\undefined
\newcommand{\because}[1]{\textparen{because #1}}
\else
\renewcommand{\because}[1]{\textparen{because #1}}
\fi

\newcommand{\sge}{\succeq}
\newcommand{\sle}{\preceq}

\newcommand{\maximize}{\mathop{\textrm{maximize}}}

\newcommand{\bits}{\{0,1\}}
\newcommand{\sbits}{\{\pm1\}}

\newcommand{\super}[2]{#1^{\paren{#2}}}

\newcommand{\inv}[1]{{#1^{-1}}}

\newcommand{\seteq}{\mathrel{\mathop:}=}

\newcommand{\from}{\colon}

\newcommand{\mper}{\,.}
\newcommand{\mcom}{\,,}

\newcommand\bdot\bullet

\ifx\mathds\undefined 
\newcommand{\Ind}{\mathbb I}
\else
\newcommand{\Ind}{\mathds 1}
\fi

\renewcommand{\th}{\textsuperscript{th}\xspace}

\DeclareMathOperator{\Inf}{Inf}
\DeclareMathOperator{\inj}{inj}

\DeclareMathOperator{\tr}{Tr}

\DeclareMathOperator{\poly}{poly}

\DeclareMathOperator{\supp}{supp}

\DeclareMathOperator{\sign}{sign}
\DeclareMathOperator{\conv}{conv}

\DeclareMathOperator{\rank}{rank}

\DeclareMathOperator{\tsum}{{\textstyle \sum}}

\newcommand{\Z}{\mathbb Z}
\newcommand{\N}{\mathbb N}
\newcommand{\R}{\mathbb R}
\newcommand{\C}{\mathbb C}
\newcommand{\F}{\mathbb F}

\newcommand{\GF}[1]{\mathbb F_{#1}}

\newcommand{\problemmacro}[1]{\texorpdfstring{\textsc{#1}}{#1}\xspace}

\newcommand{\uniquegames}{\problemmacro{Unique Games}}

\newcommand{\smallsetexpansion}{\problemmacro{Small-Set Expansion}}

\newcommand{\cD}{\mathcal D}

\newcommand{\cI}{\mathcal I}

\newcommand{\cP}{\mathcal P}

\newcommand{\cR}{\mathcal R}
\newcommand{\cS}{\mathcal S}

\newcommand{\cU}{\mathcal U}
\newcommand{\cV}{\mathcal V}
\newcommand{\cW}{\mathcal W}
\newcommand{\cX}{\mathcal X}
\newcommand{\cY}{\mathcal Y}
\newcommand{\cZ}{\mathcal Z}

\newcommand{\bbR}{\mathbb R}

\renewcommand{\leq}{\leqslant}
\renewcommand{\le}{\leqslant}
\renewcommand{\geq}{\geqslant}
\renewcommand{\ge}{\geqslant}

\ifnum\showdraftbox=1
\newcommand{\draftbox}{\begin{center}
  \fbox{%
    \begin{minipage}{2in}%
      \begin{center}%
          \Large\textsc{Working Draft}\\%
        Please do not distribute%
      \end{center}%
    \end{minipage}%
  }%
\end{center}
\vspace{0.2cm}}
\else
\newcommand{\draftbox}{}
\fi

\let\epsilon=\varepsilon

\numberwithin{equation}{section}

\newcommand{\MYstore}[2]{%
  \global\expandafter \def \csname MYMEMORY #1 \endcsname{#2}%
}

\newcommand{\MYload}[1]{%
  \csname MYMEMORY #1 \endcsname%
}

\newcommand{\MYnewlabel}[1]{%
  \newcommand\MYcurrentlabel{#1}%
  \MYoldlabel{#1}%
}

\newcommand{\MYdummylabel}[1]{}

\newcommand{\torestate}[1]{%
  \let\MYoldlabel\label%
  \let\label\MYnewlabel%
  #1%
  \MYstore{\MYcurrentlabel}{#1}%
  \let\label\MYoldlabel%
}

\newcommand{\restatetheorem}[1]{%
  \let\MYoldlabel\label
  \let\label\MYdummylabel
  \begin{theorem*}[Restatement of \prettyref{#1}]
    \MYload{#1}
  \end{theorem*}
  \let\label\MYoldlabel
}

\newcommand{\restatelemma}[1]{%
  \let\MYoldlabel\label
  \let\label\MYdummylabel
  \begin{lemma*}[Restatement of \prettyref{#1}]
    \MYload{#1}
  \end{lemma*}
  \let\label\MYoldlabel
}

\newcommand{\restateprop}[1]{%
  \let\MYoldlabel\label
  \let\label\MYdummylabel
  \begin{proposition*}[Restatement of \prettyref{#1}]
    \MYload{#1}
  \end{proposition*}
  \let\label\MYoldlabel
}

\newcommand{\restatefact}[1]{%
  \let\MYoldlabel\label
  \let\label\MYdummylabel
  \begin{fact*}[Restatement of \prettyref{#1}]
    \MYload{#1}
  \end{fact*}
  \let\label\MYoldlabel
}

\newcommand{\restate}[1]{%
  \let\MYoldlabel\label
  \let\label\MYdummylabel
  \MYload{#1}
  \let\label\MYoldlabel
}

\newcommand{\addreferencesection}{
  \phantomsection
  \addcontentsline{toc}{section}{References}
}

\newcommand{\sse}{\subseteq}
\newcommand{\ra}{\rightarrow}
\newcommand{\e}{\epsilon}
\newcommand{\eps}{\epsilon}
\newcommand{\eset}{\emptyset}

\let\origparagraph\paragraph
\renewcommand{\paragraph}[1]{\origparagraph{#1.}}

\allowdisplaybreaks

\newcommand{\cclassmacro}[1]{\texorpdfstring{\textbf{#1}}{#1}\xspace}
\newcommand{\p}{\cclassmacro{P}}
\newcommand{\np}{\cclassmacro{NP}}

\DeclareMathOperator{\Sep}{Sep}
\DeclareMathOperator{\PPT}{PPT}
\DeclareMathOperator{\DPS}{DPS}
\DeclareMathOperator{\Ext}{Ext}

\newcommand{\eq}[1]{(\ref{eq:#1})}
\def\be{\begin{equation}}
\def\ee{\end{equation}}
\def\ba#1\ea{\begin{align}#1\end{align}}
\def\ban#1\ean{\begin{align*}#1\end{align*}}

\def\benum{\begin{enumerate}}
\def\eenum{\end{enumerate}}
\def\bit{\begin{itemize}}
\def\eit{\end{itemize}}

\newcommand{\secref}[1]{Section~\ref{sec:#1}}

\newcommand{\lemref}[1]{Lemma~\ref{lem:#1}}
\newcommand{\thmref}[1]{Theorem~\ref{thm:#1}}

\newcommand{\corref}[1]{Corollary~\ref{cor:#1}}

\newcommand{\ifexp}[1]{}  
\newcommand{\ifcount}[1]{#1} 
\newcommand{\bh}{\mathbf{h}}
\newcommand{\bS}{\mathbf{S}}

\def\begsub#1#2\endsub{\begin{subequations}\label{eq:#1}\begin{align}#2\end{align}\end{subequations}}
\nc\qand{\qquad\text{and}\qquad}